\documentclass[conference]{IEEEtran}
\usepackage{etoolbox}
\newtoggle{longversion}
\toggletrue{longversion}
\newtoggle{isitlongversion}
\usepackage{url}

\usepackage{hyperref}
\hypersetup{
  colorlinks   = true, %Colours links instead of ugly boxes
  urlcolor     = black!50!red, %Colour for external hyperlinks
  linkcolor    = black!50!brown, %Colour of internal links
  citecolor   = black!50!brown %Colour of citations
}
\usepackage{pgfplots}
\pgfplotsset{compat=1.15}
\usepackage{amsmath}
\usepackage{amssymb}
\usepackage{amsthm}
\usepackage{verbatim}
\usepackage{bm}
\usepackage{color,graphicx,xcolor}
\usepackage{mdwtab}
\usepackage{subfigure}
\usepackage{mathtools,tikz}
\mathtoolsset{showonlyrefs} 
\usepackage{hhline}
\usepackage{multirow}
\usepackage{pdfpages}
\usepackage{enumitem}

\iftoggle{longversion}{
\onecolumn
\newcommand{\longonly}[1]{#1} %stuff to put only in long version
\newcommand{\isitlongonly}[1]{} %stuff to put only in ISIT long version
\newcommand{\isitshortonly}[1]{} %stuff to put only in ISIT short version
\newcommand{\shortonly}[1]{} %stuff to put only in short version
}
{
\newcommand{\longonly}[1]{} %stuff to put only in long version
\newcommand{\shortonly}[1]{#1} %stuff to put only in short version

\iftoggle{isitlongversion}
{
\onecolumn
\newcommand{\isitlongonly}[1]{#1} %stuff to put only in ISIT long version
\newcommand{\isitshortonly}[1]{} %stuff to put only in ISIT short version

}
{
\twocolumn
\setlength{\abovedisplayskip}{0.3em}
\setlength{\belowdisplayskip}{0.3pt}
\setlength{\abovedisplayshortskip}{0pt}
\setlength{\belowdisplayshortskip}{0pt}
\newcommand{\isitlongonly}[1]{} %stuff to put only in ISIT long version
\newcommand{\isitshortonly}[1]{#1} %stuff to put only in ISIT short version

\usepackage[skip=0.4em]{caption}
\DeclareCaptionFormat{myformat}{\fontsize{8}{9}\selectfont#1#2#3}
\captionsetup{format=myformat}
\setlength{\belowcaptionskip}{-1.6em}
}
}

\IEEEoverridecommandlockouts
\usepackage{cite}
\usepackage{amsmath,amssymb,amsfonts,amsthm,xspace}
\usepackage{algorithmic}
\usepackage{graphicx}
\usepackage{textcomp}

\usetikzlibrary{arrows,shapes.geometric}
\allowdisplaybreaks
\newtheorem{remark}{Remark}

\newtheorem{thm}{Theorem}
\newtheorem{defn}{Definition}
\newtheorem{example}{Example}
\newtheorem{lemma}[thm]{Lemma}
\newtheorem{corollary}[thm]{Corollary}
\newtheorem{prop}[thm]{Proposition}
\newcommand{\subparagraph}[1]{\par {\em\underline{#1:}}}

\usepackage{pgffor}
\foreach \x in {a,...,z}{%
\expandafter\xdef\csname vec\x \endcsname{\noexpand\ensuremath{\noexpand\bm{\x}}}
}

\foreach \x in {A,...,Z}{%
\expandafter\xdef\csname vec\x \endcsname{\noexpand\ensuremath{\noexpand\bm{\x}}}
}

\foreach \x in {A,...,Z}{%
\expandafter\xdef\csname c\x \endcsname{\noexpand\ensuremath{\noexpand\mathcal{\x}}}
}

\foreach \x in {A,...,Z}{%
\expandafter\xdef\csname bb\x \endcsname{\noexpand\ensuremath{\noexpand\mathbb{\x}}}
}

\newcommand{\defineqq}{\ensuremath{\stackrel{\text{\tiny def}}{=}}}

\usetikzlibrary{shapes,arrows,positioning,decorations.pathreplacing,calc}

\def\msg{\ensuremath{m}} %message
\def\msgh{\ensuremath{\hat{\msg}}} %reconstruction
\def\nummsg{\mbox{$N$}} %number of messages
\def\auth{\ensuremath{\mathrm{auth}}}

\newcommand{\mch}{\ensuremath{W_{Z|X,Y}}\xspace}

\newcommand{\indep}{\raisebox{0.05em}{\rotatebox[origin=c]{90}{$\models$}}\xspace}

\newcommand{\na}{\ensuremath{\text{hon}}\xspace}
\newcommand{\malone}{\ensuremath{\text{mal \one}}}
\newcommand{\maltwo}{\ensuremath{\text{mal \two}}}

\newcommand{\spoofable}{\text{spoofable}\xspace}

\newcommand{\spoofability}{\text{spoofability}\xspace}

\newcommand{\inb}[1]{\left\{#1\right\}}
\newcommand{\inp}[1]{\left(#1\right)}
\newcommand{\insq}[1]{\left[#1\right]}

\def\reliable{\ensuremath{\mathrm{reliable}}}
\def\MAC{\ensuremath{\mathrm{MAC}}}
\def\AVMAC{\ensuremath{\mathrm{AV-MAC}}}

\newcommand{\mo}{\ensuremath{m_{\one}}\xspace}
\newcommand{\mt}{\ensuremath{m_{\two}}\xspace}
\newcommand\independent{\protect\mathpalette{\protect\independenT}{\perp}}
\def\independenT#1#2{\mathrel{\rlap{$#1#2$}\mkern2mu{#1#2}}}

\def\one{\ensuremath{\mathsf{A}}\xspace} %blame 1
\def\two{\ensuremath{\mathsf{B}}\xspace} %blame 2

\def\oneb{\ensuremath{\mathbf{a}}\xspace} %blame 1
\def\twob{\ensuremath{\mathbf{b}}\xspace} %blame 2

\newcommand{\red}[1]{{\textcolor{red}{#1}}}

\def\BibTeX{{\rm B\kern-.05em{\sc i\kern-.025em b}\kern-.08em
	T\kern-.1667em\lower.7ex\hbox{E}\kern-.125emX}}
\begin{document}

\title{Communication With Adversary Identification in Byzantine Multiple Access Channels
\thanks{N. Sangwan and V. Prabhakaran acknowledge support of
the Department of Atomic Energy, Government of India,
under project no. RTI4001. N. Sangwan’s work was
additionally supported by the Tata Consultancy Services (TCS)
foundation through the TCS Research Scholar Program. Work
of B. K. Dey was supported in part by Bharti Centre for
Communication in IIT Bombay. V. Prabhakaran's work was also
supported by the Science \& Engineering Research Board, India
through project MTR/2020/000308.}}

\author{\IEEEauthorblockN{Neha Sangwan}
\IEEEauthorblockA{TIFR, India}
\and
\IEEEauthorblockN{Mayank Bakshi}
\IEEEauthorblockA{Huawei, Hong Kong}
\and
\IEEEauthorblockN{Bikash Kumar Dey}
\IEEEauthorblockA{IIT Bombay, India}
\and
\IEEEauthorblockN{Vinod M. Prabhakaran}
\IEEEauthorblockA{TIFR, India}
}

%Uncomment next 6 lines for the shorter version
% \makeatletter
% \patchcmd{\@maketitle}
%   {\addvspace{0.5\baselineskip}\egroup}
%   {\addvspace{-1.5\baselineskip}\egroup}
%   {}
%   {}
% \makeatother

\maketitle

\isitshortonly{\vspace{-1in}}
\begin{abstract}
We introduce the problem of determining the identity of a byzantine user (internal adversary) in a communication system. We consider a two-user discrete memoryless multiple access channel where either user may deviate from the prescribed behaviour. Owing to the noisy nature of the channel, it may be overly restrictive to attempt to detect all deviations. In our formulation, we only require detecting deviations which impede the decoding of the non-deviating user's message. When neither user deviates, correct decoding is required. When one user deviates, the decoder must  either output a pair of messages of which the message of the {non-deviating} user is correct or identify the deviating user. The users and the receiver do not share any randomness. The results include a characterization of the set of channels where communication is feasible, and an inner and outer bound on the capacity region. 

%\red{Further, the positivity condition is compared with symmetrizability and overwritability conditions (see \cite{NehaBDPITW19} and \cite{NehaBDPISIT19}) which appear in other adversary models.}
\end{abstract}

%\begin{IEEEkeywords}
%Multiple Access Channels, adversary, Byzantine users
%\end{IEEEkeywords}

%!TeX root=paper.tex
%\noindent\textcolor{blue}{\textbf{Note to the reviewers:} There are two typos in Section III in the submitted version. They have been corrected in this version. Please see Section~\ref{sec:feasibility}.}
\section{Introduction}\label{sec:intro}
In many modern wireless communication applications (e.g., the Internet of Things), devices with varying levels of security are connected over a shared communication medium.  Compromised devices may allow an adversary to disrupt the communication of other devices. This motivates the question we study in this paper -- is it possible to design a communication system in which malicious actions by compromised devices can be detected so that such devices can be isolated or taken offline?

We consider a two-user Multiple Access Channel (MAC) where either user may deviate from the prescribed behaviour. Owing to the noisy nature of the channel, it may be overly restrictive to attempt to detect all deviations. Indeed, it suffices to detect only such deviations which impede the correct decoding of the other user’s message. We formulate a communication problem for the MAC with the following decoding guarantee (Fig.~\ref{fig:authcomMAC}): the decoder outputs either a pair of messages {\em or} declares one of the users to be deviating. When both users are honest, the decoder must output the correct message pair with high probability (w.h.p.); when exactly one user deviates, w.h.p., the decoder must either correctly detect the deviating user or output a message pair of which the message of the other (honest) user is correct (see Section~\ref{sec:model}). No guarantees are made if both users deviate. Thus, we require that a deviating user cannot cause a decoding error for the other user without getting caught.
Throughout this paper, we assume that encoders and decoder do not share any randomness.

For comparison, consider the stronger guarantee  of {\em reliable communication} where the decoder outputs a pair of messages such that the message(s) of non-deviating user(s) is correct w.h.p. \cite{NehaBDPITW19}. While achieving this clearly satisfies the requirements of the present model, it might be too demanding. For example, in a binary erasure MAC\footnote{The binary erasure MAC has binary inputs $X, Y$ and outputs $Z=X+Y$ where $+$ is real addition.}~\cite[pg.~83]{YHKEG}, a deviating user can run an independent copy of the honest user's encoder and inject a spurious message which will appear equally plausible to the decoder as the honest user's actual message (also see section~\ref{sec:comparison}). Thus, reliable communication is impossible over the binary erasure MAC. However, our results, when specialized to this channel, will show that communication with adversary identification is possible. 
That is, under our coding scheme it is impossible for a byzantine user to mount a successful attack without getting caught. In fact, for the binary erasure MAC, we show that the capacity region of communication with adversary identification is the same as the (non-adversarial) capacity region of the binary erasure MAC (see Section~\ref{sec:example_tightness}).

{Another decoding guarantee that is weaker than} the present model allows the decoder  to declare adversarial interference (in the presence of malicious user(s)) without identifying the adversary. We called this {\em authenticated communication} and characterized its  feasibility condition and capacity region in\cite{NehaBDPISIT19}. The feasibility condition is called {\em overwritability}, a notion which was introduced by Kosut and Kliewer for network coding \cite{KK2} and AVCs \cite{KosutKITW18}. 

{The present model lies between the models for reliable communication and authenticated communication in a byzantine MAC. However, obtaining results here appears to be significantly more challenging.}
On the one hand, for reliable communication over the two-user MAC, we may treat the channel from each user to the decoder as an arbitrarily varying channel (AVC) \cite{BBT60} with the other user's input as state. Hence, the users may send their messages using the corresponding AVC codes~\cite{CsiszarN88}. Thus, the rectangular region defined by the capacities of the two AVCs is achievable\footnote{In fact, this rectangular region defined by the capacities of the two AVCs is the reliable communication capacity region since a deviating user can act exactly like the adversary in the AVC of the other user. Note that the AVCs for binary erasure MAC have zero AVC capacity.\label{ftn:reliable}}.  
On the other hand, for authenticated communication over the two-user MAC, our achievable strategy in \cite{NehaBDPISIT19} involved an unauthenticated communication phase using a non-adversarial MAC code followed by separate (short) authentication phases for each user’s decoded message. Failure to authenticate a user’s message implies the presence of an adversary (though not its identity since the user whose message is being authenticated might have deviated to cause the authentication failure).
In both 
%reliable and authenticated communication
the cases above, the decoder, when it accounts for the byzantine nature of the users, deals with the users one at a time. 
However, similar decoding strategies seem to be insufficient for adversary identification.
Determining the identity of a deviating user requires dealing with the byzantine nature of both users simultaneously, thereby complicating the decoder design (see Section~\ref{sec:feasibility}).

We characterize the infeasibility of communication with adversary identification using a channel condition we call {\em \spoofability} (see Fig.~\ref{fig:spoof1}). 
It allows a deviating user to mount an attack which can be confused with an attack of the other user and which  introduces a spurious message that can be confused with the actual message of the (other) honest user.
%In a \spoofable channel, the adversarial user can mount an attack which hides the identity of the deviating user as well as confuses the decoder between the sent message and a spurious message of the honest user.
When the channel is not \spoofable, a deterministic code in the style of \cite{CsiszarN88} can provide positive rates to both the users (Theorem~\ref{thm:main_result}). Our outer bound is in terms of the capacity of an Arbitrarily Varying-MAC{\cite{Jahn81}}~(Theorem~\ref{eq:outer_bound}). Further, a comparison is drawn between \spoofability and the feasibility conditions for the reliable communication and authenticated communication models.

\paragraph*{Related works} There is a long line of works in the information theory literature on communication in the  presence of external adversaries (see \cite{survey} for a survey).
Communication in systems with byzantine users has also received some attention \cite{Jaggi7,KTong,KK2,Yener,NehaBDPITW19,NehaBDPISIT19}. 
%Networks coding with byzantine attacks on nodes and edges has been studied \cite{Jaggi7,KTong,KK2}. He and Yener  \cite{Yener} considered a Gaussian two-hop network with an eavesdropping and byzantine adversarial relay where the receiver is required to decode with message secrecy and detect byzantine attack.
Message authentication codes where the users have pre-shared keys and
communicate over noiseless channels have been extensively
studied~\cite{SimmonsCRYPTO84,MaurerIT00,Gungor16}. Message authentication over noisy
channels has also been considered~\cite{LaiEPIT09,Jiang14,Gungor16,TuLIT18}.
There has also been some recent work on authenticated communication over channels in which an external adversary may be present; in the presence of the adversary, the decoder may declare adversarial interference instead of decoding  \cite{KosutKITW18,BKKGYu,Graves16,BeemerCNS20} (In a 2-user MAC model in \cite{BeemerCNS20} when declaring the presence of an adversary, the decoder is required to decode at least one user's message.). 
These models are different from the present model, where, when declaring the presence of an (internal) adversary, we also require the decoder to output its identity.

%!TeX root=paper.tex
\iffalse
\section{Notation}
We will use the following notation: The random variables $X,X', \tilde{X}$ are all distributed over a finite alphabet \cX. Similarly, the random variables $Y,Y', \tilde{Y}$ are distributed over a finite alphabet \cY. The random variable $Z$ is distributed over a finite alphabet \cZ. For an alphabet $\cX$, $\cP^n_{\cX}$ denotes the set of all empirical distributions over blocklength $n$. For a distribution $P_X\in \cP^n_{\cX}$, $T^n_{X}$ denotes the set of typical sequences with relative frequencies specified by $P_X$. For a distribution $P_{XY}\in \cP^n_{\cX\times \cY}$, the set $T^n_{X|Y}(\cdot|\vecy)$ denotes $\{\vecx: (\vecx,\vecy)\in T^n_{XY}\}$. For a vector $\vecx\in \cX^n$, let $\vecx(t)$ be the $t^{\text{th}}$ component of $\vecx$. All the exponents and logarithms are with respect to base $2$.
\fi

\section{System Model}\label{sec:model}
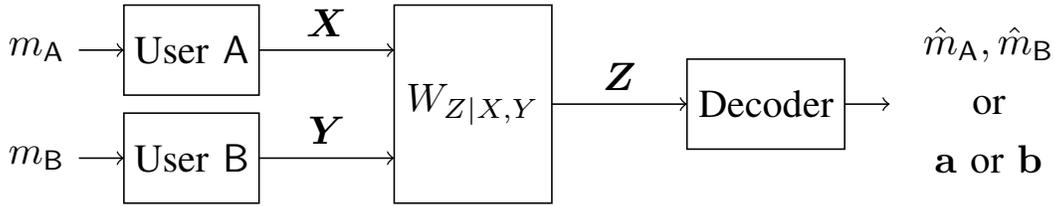
\begin{figure}[h]\centering
\resizebox{0.8\columnwidth}{!}{\begin{tikzpicture}[scale=0.4]
	\draw (2,0) rectangle ++(3,2) node[pos=.5]{ User \two};
	\draw (2,2.4) rectangle ++(3,2) node[pos=.5]{ User \one};
	\draw (10-2,0) rectangle ++(3.5,4.4) node[pos=.5]{ $W_{Z|X,Y}$};
	\draw (16.5-2,1.2) rectangle ++(3.5,2) node[pos=.5]{ Decoder};
	\draw[->] (1,1) node[anchor=east]{ $\msg_{\two}$} -- ++ (1,0) ;
	\draw[->] (5,1) -- node[above] { $\vecY$} ++ (3,0);
	
	\draw[->] (1,3.4) node[anchor=east]{ $\msg_{\one}$} -- ++ (1,0) ;
	\draw[->] (5,3.4) -- node[above] { $\vecX$} ++ (3,0);

	\draw[->] (13.5-2,2.2) -- node[above] { $\vecZ$} ++ (3,0);
	\draw[->] (19-1,2.2) -- ++ (1,0) node[anchor=west]{\begin{array}{c} \msgh_{\one},\msgh_{\two}\\ \text{or}\\ \oneb\text{ or }\twob\end{array}};
\end{tikzpicture}}
\caption{MAC with byzantine users: Reliable decoding of both the messages is required when neither user deviates. When a user (say, user \two) deviates, the decoded message should either be correct for the honest user or the decoder should identify the deviating user (by outputting $\twob$) with high probability.}\label{fig:authcomMAC}
\end{figure}
\vspace{0.2 cm}
\paragraph*{Notation}
For a set $\cS\in \bbR^{k}$, let $\mathsf{conv}(\cS)$ and $\mathsf{int}(\cS)$ denote its convex closure and interior respectively. {Let $\vecx\in \cX^n$ (resp. $\vecX$ distributed over $\cX^n$) denote the $n$-length vectors (resp. $n$-length random vectors).}  For a distribution $P_X$ on $\cX$, let $T^n_{X}$ denote the set of all $n$-length sequences $\vecx\in \cX^n$ with empirical distribution $P_X$. {$\textsf{Unif}(\cA)$ denotes the uniform distribution over the set $\cA$.} For a two-user MAC $W(.|.,.)$, we will use $\cC_{\MAC}(W)$ (or simply $\cC_{\MAC}$) to denote its (non-adversarial) capacity region. We will use $W^n$ to denote the $n$-fold product of the channel $W$.

Consider a two-user discrete memoryless Multiple Access Channel (MAC)  as shown in Fig.~\ref{fig:authcomMAC}. User $\one$ has input alphabet $\cX$ and user $\two$ has input alphabet $\cY$. The output alphabet of the channel is $\cZ$. The sets $\cX$, $\cY$ and $\cZ$ are finite. We study communication in a MAC where either user may deviate from the communication protocol by sending any sequence of its choice from its input alphabet. {While doing so, the deviating user is unaware of other user's input}. We will refer to this channel model  as a {\em MAC with byzantine users}. 
\begin{defn}[Adversary identifying code]\label{defn:code}
An $(\nummsg_{\one},\nummsg_{\two},n)$  {\em deterministic adversary identifying code} for a MAC with byzantine users consists of the following: 
\begin{enumerate}[label=(\roman*)]
\item Two message sets, $\mathcal{M}_i = \{1,\ldots,\nummsg_i\}$, $i=\one,\two$,
\item Two deterministic encoders, $f_{\one}^{(n)}:\mathcal{M}_{\one}\rightarrow \mathcal{X}^n$ and $f_{\two}^{(n)}:\mathcal{M}_{\two}\rightarrow\mathcal{Y}^n$, and
\item A deterministic decoder, $\phi^{(n)}:\mathcal{Z}^n\rightarrow(\mathcal{M}_{\one}\times\mathcal{M}_{\two})\cup\{\oneb,\, \twob\}.$ 
\end{enumerate}
\end{defn}
The output symbol \oneb indicates that user \one is adversarial. Similarly, \twob indicates that user \two is adversarial. The {\em average probability of error} $P_{e}(f^{(n)}_{\one},f^{(n)}_{\two},\phi^{(n)})$ is the maximum of the average probabilities of error in the following three cases: (1) both users are honest, (2) user \one is adversarial, and (3) user \two is adversarial. When both users are honest, the decoded messages should be correct with high probability (w.h.p.). Let $\cE_{\mo,\mt} = \inb{\vecz:\phi^{(n)}(\vecz)\neq(\mo, \mt)}$ denote the corresponding error event.  The average error {probability} when both users are honest is
\shortonly{
\begin{align*}
&P_{e,\na}\hspace{-0.25em} \defineqq\frac{1}{N_{\one}\cdot N_{\two}} \sum_{\substack{(\mo, \mt)\in\\ \mathcal{M}_{\one}\times\mathcal{M}_{\two}}}W^n\inp{\cE_{\mo,\mt}|f_{\one}^{(n)}(\mo), f_{\two}^{(n)}(\mt)}.
\end{align*}
When user \one is adversarial, the decoder's output, w.h.p., should either be the symbol $\oneb$ or a pair of messages of which the message of user \two is correct. The error event $\cE_{\mt} \defineqq \inb{\vecz:\phi^{(n)}(\vecz)\notin\inp{\cM_{\one}\times\{\mt\}}\cup{\{\oneb\}}}$. The average probability of error when user \one is adversarial is 
\begin{align*}
&P_{e,\malone} \defineqq \max_{\vecx\in\cX^n} \left(\frac{1}{N_{\two}}\sum_{m_{\two}\in \mathcal{M}_{\two}}W^n\inp{\cE_{\mt}|\vecx, f_{\two}^{(n)}(\mt)}\right).
\end{align*} Similarly, for $\cE_{\mo} \defineqq \inb{\vecz:\phi^{(n)}(\vecz)\notin\inp{\{\mo\}\times\cM_{\two}}\cup{\{\twob\}}}$, the average probability of error when user \two is adversarial is 
\begin{align*}
&P_{e,\maltwo} \defineqq \max_{\vecy\in\cY^n}\left(\frac{1}{N_{\one}}\sum_{m_{\one}\in \mathcal{M}_{\one}}W^n\inp{\cE_{\mo}|f_{\one}^{(n)}(\mo), \vecy}\right).
\end{align*}}
\longonly{
\begin{align}
&P_{e,\na}\hspace{-0.25em} \defineqq\frac{1}{N_{\one}\cdot N_{\two}} \sum_{\substack{(\mo, \mt)\in\\ \mathcal{M}_{\one}\times\mathcal{M}_{\two}}}W^n\inp{\cE_{\mo,\mt}|f_{\one}^{(n)}(\mo), f_{\two}^{(n)}(\mt)}.\label{eq:na}
\end{align}
When user \one is adversarial, the decoder's output, w.h.p., should either be the symbol $\oneb$ or a pair of messages of which the message of user \two is correct. The error event $\cE_{\mt} \defineqq \inb{\vecz:\phi^{(n)}(\vecz)\notin\inp{\cM_{\one}\times\{\mt\}}\cup{\{\oneb\}}}$. The average probability of error when user \one is adversarial is 
\begin{align}
&P_{e,\malone} \defineqq \max_{\vecx\in\cX^n} \left(\frac{1}{N_{\two}}\sum_{m_{\two}\in \mathcal{M}_{\two}}W^n\inp{\cE_{\mt}|\vecx, f_{\two}^{(n)}(\mt)}\right).\label{eq:mal1}
\end{align} Similarly, for $\cE_{\mo} \defineqq \inb{\vecz:\phi^{(n)}(\vecz)\notin\inp{\{\mo\}\times\cM_{\two}}\cup{\{\twob\}}}$, the average probability of error when user \two is adversarial is 
\begin{align}
&P_{e,\maltwo} \defineqq \max_{\vecy\in\cY^n}\left(\frac{1}{N_{\one}}\sum_{m_{\one}\in \mathcal{M}_{\one}}W^n\inp{\cE_{\mo}|f_{\one}^{(n)}(\mo), \vecy}\right).\label{eq:mal2}
\end{align}}
We define the \emph{average probability of error} as
\begin{align*}
P_{e}(f_{\one}^{(n)},f_{\two}^{(n)},\phi^{(n)})\defineqq \max{\inb{P_{e,\na},P_{e,\malone},P_{e,\maltwo}}}.
\end{align*}
\longonly{{Note that the probability of error under a randomized attack is the weighted average of the probabilities of errors under the different deterministic attacks and hence maximized by a deterministic attack. Thus, $P_{e,\maltwo}$ is an upper bound on the probability of error for any attack by user \two, deterministic or random}} \shortonly{Note that $P_{e,\maltwo}$ is clearly an upper bound on the probability of error for any attack by user \two, deterministic or random. } Similarly, $P_{e,\malone}$ is an upper bound for any attack by user \one.  Thus, the probability of error under deterministic attacks is same as that under randomized attacks. 

%%%%% Figure for spoofability so that it appears before the defns in results.tex %%%%%
\begin{figure*}
\centering
\subfigure[]{
\begin{tikzpicture}[scale=0.5]
	\draw (1.7-0.8-0.3,2.9-0.1) rectangle ++(2.3,1.2) node[pos=.5]{\footnotesize ${Q^n_{Y|\tilde{X},\tilde{Y}}}$};
	\draw (4.1,4) rectangle ++(1.5,1.5) node[pos=.5]{ $W^n$};	
	\draw[ ->] (0.4-0.3,3.0) node[anchor=east]{ $\tilde{\vecy}\phantom{'}$} -- ++ (0.5,0) ;
	\draw[ ->] (0.4-0.3,3.7+0.1) node[anchor=east]{ $\tilde{\vecx}\phantom{'}$} -- ++ (0.5,0) ;
	\draw[->] (3.4,5.1) -- ++(0.7,0);
	\draw[->] (3.4,4.4) -- ++(0.7,0);
	\draw[-] (3.4,4.4) -- ++ (0,-1);
	\draw[-] (3.4,5.1) -- ++ (0,1);
	\draw[-] (2.9, 3.4)  --  ++(0.5,0);
	\draw[-] (0.4-0.3,6.1) node[anchor=east]{ $\vecx'$} -- ++ (3.3,0);
	\draw[->] (5.4+1-0.3-0.5,4.75) --  ++ (0.5,0)node[anchor= west] { $\vecz$};
	\end{tikzpicture}\label{fig:spoof1a} 
}
\subfigure[]{
\begin{tikzpicture}[scale=0.5]
	\draw (1.7-0.8-0.3,2.9-0.1) rectangle ++(2.3,1.2) node[pos=.5]{\footnotesize ${Q^n_{Y|\tilde{X},\tilde{Y}}}$};
	\draw (4.1,4) rectangle ++(1.5,1.5) node[pos=.5]{ $W^n$};	
	\draw[ ->] (0.4-0.3,3.0) node[anchor=east]{ $\tilde{\vecy}\phantom{'}$} -- ++ (0.5,0) ;
	\draw[ ->] (0.4-0.3,3.7+0.1) node[anchor=east]{ $\vecx'$} -- ++ (0.5,0) ;
	\draw[->] (3.4,5.1) -- ++(0.7,0);
	\draw[->] (3.4,4.4) -- ++(0.7,0);
	\draw[-] (3.4,4.4) -- ++ (0,-1);
	\draw[-] (3.4,5.1) -- ++ (0,1);
	\draw[-] (2.9, 3.4)  --  ++(0.5,0);
	\draw[-] (0.4-0.3,6.1) node[anchor=east]{ $\tilde{\vecx}\phantom{'}$} -- ++ (3.3,0);
	\draw[->] (5.4+1-0.3-0.5,4.75) --  ++ (0.5,0)node[anchor= west] { $\vecz$}; 
\end{tikzpicture}\label{fig:spoof1b} 
}
\subfigure[]{
\begin{tikzpicture}[scale=0.5]
	\draw (1.7-1.1,5.6-0.1) rectangle ++(2.5,1.2) node[pos=.5]{\footnotesize ${Q^n_{X|\tilde{X},X'}}$};
	\draw (4.1,4) rectangle ++(1.5,1.5) node[pos=.5]{ $W^n$};	
	\draw[ ->] (0.1,5.8-0.1) node[anchor=east]{ ${\vecx'}$} -- ++ (0.5,0) ;
	\draw[ ->] (0.1,6.4+0.1) node[anchor=east]{ $\tilde{\vecx}\phantom{'}$} -- ++ (0.5,0) ;
	\draw[->] (3.4,5.1) -- ++(0.7,0);
	\draw[->] (3.4,4.4) -- ++(0.7,0);
	\draw[-] (3.4,4.4) -- ++ (0,-1);
	\draw[-] (3.4,5.1) -- ++ (0,1);
	\draw[-] (3.1, 6.1)  --  ++(0.3,0);
	\draw[-] (0.1,3.4) node[anchor=east]{ $\tilde{\vecy}\phantom{'}$} -- ++ (3.3,0);
	\draw[->] (5.4+1-0.3-0.5,4.75) --  ++ (0.5,0)node[anchor= west] { $\vecz$};
\end{tikzpicture}\label{fig:spoof1c} 
}\caption{When \eqref{eq:spoof1} holds, for $(\vecx', \tilde{\vecx}, \tilde{\vecy})\in \cX^n\times\cX^n\times\cY^n$, the output distributions in the three cases above will be the same.} \label{fig:spoof1}
\end{figure*}
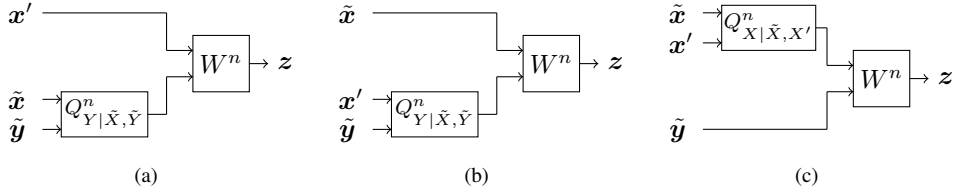

\longonly{
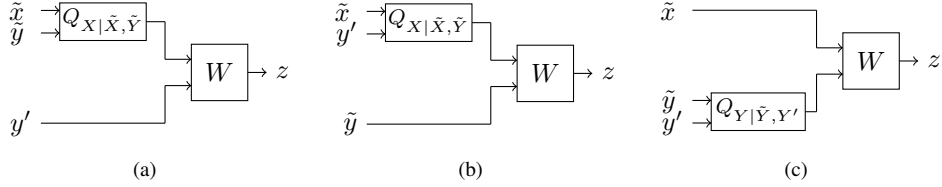
\begin{figure*}
\centering
\subfigure[]{
\begin{tikzpicture}[scale=0.5]
	\draw (1.7-0.8-0.3,5.6) rectangle ++(2.3,1) node[pos=.5]{\footnotesize ${Q_{X|\tilde{X},\tilde{Y}}}$};
	\draw (4.1,4) rectangle ++(1.5,1.5) node[pos=.5]{ $W$};	
	\draw[ ->] (0.1,5.8) node[anchor=east]{ $\tilde{y}\phantom{'}$} -- ++ (0.5,0) ;
	\draw[ ->] (0.1,6.4) node[anchor=east]{ $\tilde{x}\phantom{'}$} -- ++ (0.5,0) ;
	\draw[->] (3.4,5.1) -- ++(0.7,0);
	\draw[->] (3.4,4.4) -- ++(0.7,0);
	\draw[-] (3.4,4.4) -- ++ (0,-1);
	\draw[-] (3.4,5.1) -- ++ (0,1);
	\draw[-] (2.9, 6.1)  --  ++(0.5,0);
	\draw[-] (0.1,3.4) node[anchor=east]{ $y'$} -- ++ (3.3,0);
	\draw[->] (5.4+1-0.3-0.5,4.75) --  ++ (0.5,0)node[anchor= west] { $z$};
\end{tikzpicture}\label{fig:spoof2a} 
}
\subfigure[]{
\begin{tikzpicture}[scale=0.5]
	\draw (1.7-0.8-0.3,5.6) rectangle ++(2.3,1) node[pos=.5]{\footnotesize ${Q_{X|\tilde{X},\tilde{Y}}}$};
	\draw (4.1,4) rectangle ++(1.5,1.5) node[pos=.5]{ $W$};	
	\draw[ ->] (0.1,5.8) node[anchor=east]{ ${y'}$} -- ++ (0.5,0) ;
	\draw[ ->] (0.1,6.4) node[anchor=east]{ $\tilde{x}\phantom{'}$} -- ++ (0.5,0) ;
	\draw[->] (3.4,5.1) -- ++(0.7,0);
	\draw[->] (3.4,4.4) -- ++(0.7,0);
	\draw[-] (3.4,4.4) -- ++ (0,-1);
	\draw[-] (3.4,5.1) -- ++ (0,1);
	\draw[-] (2.9, 6.1)  --  ++(0.5,0);
	\draw[-] (0.1,3.4) node[anchor=east]{ $\tilde{y}$} -- ++ (3.3,0);
	\draw[->] (5.4+1-0.3-0.5,4.75) --  ++ (0.5,0)node[anchor= west] { $z$};
\end{tikzpicture}\label{fig:spoof2b} 
}
\subfigure[]{
\begin{tikzpicture}[scale=0.5]
	\draw (1.7-1.1,2.9) rectangle ++(2.5,1) node[pos=.5]{\footnotesize ${Q_{Y|\tilde{Y},Y'}}$};
	\draw (4.1,4) rectangle ++(1.5,1.5) node[pos=.5]{ $W$};	
	\draw[ ->] (0.1,3.1) node[anchor=east]{ ${y'}$} -- ++ (0.5,0) ;
	\draw[ ->] (0.1,3.7) node[anchor=east]{ $\tilde{y}\phantom{'}$} -- ++ (0.5,0) ;
	\draw[->] (3.4,5.1) -- ++(0.7,0);
	\draw[->] (3.4,4.4) -- ++(0.7,0);
	\draw[-] (3.4,4.4) -- ++ (0,-1);
	\draw[-] (3.4,5.1) -- ++ (0,1);
	\draw[-] (3.1, 3.4)  --  ++(0.3,0);
	\draw[-] (0.1,6.1) node[anchor=east]{ $\tilde{x}\phantom{'}$} -- ++ (3.3,0);
	\draw[->] (5.4+1-0.3-0.5,4.75) --  ++ (0.5,0)node[anchor= west] { $z$}; 
\end{tikzpicture}\label{fig:spoof2c} 
}
\caption{A MAC is \two-\spoofable if for each $\tilde{x},\, \tilde{y},\, y',  \,z$ the conditional output distributions $P(z|\tilde{x},\tilde{y},y')$ in \ref{fig:spoof2a}, \ref{fig:spoof2b} and \ref{fig:spoof2c} are the same.} \label{fig:spoof2}
\end{figure*}
}
%%%%%%%%%%%%%%%%%%%%%%%%%%%%%%%%%%%%%%%%%%%%%%%%%%%%%%%%%%%%%%%%%%%%%%%%%%%%%%%%%%%%%%

\begin{defn}[Achievable rate pair for and capacity region of communication with adversary identification]
 $(R_{\one},R_{\two})$ is an {\em achievable rate pair for communication with adversary identification} if there exists a sequence of $(\lfloor2^{nR_{\one}}\rfloor,\lfloor2^{nR_{\two}}\rfloor,n)$ adversary identifying codes $\{f_{\one}^{(n)},f_{\two}^{(n)},\phi^{(n)}\}_{n=1}^\infty$ such that $\lim_{n\rightarrow\infty}P_{e}(f_{\one}^{(n)},f_{\two}^{(n)},\phi^{(n)})\rightarrow0.$ The {\em  capacity region of communication with adversary identification} $\cC$ is the closure of the set of all such achievable rate pairs. Let $C^{}_{\one}$ (resp. $C^{}_{ \two}$) be defined as the supremum of the set $\{R_{\one}: (R_{\one}, R_{\two})\in \cC \text{ for some }R_{\two}\}$ (resp. $\{R_{\two}: (R_{\one}, R_{\two})\in \cC \text{ for some }R_{\one}\}$).
\end{defn}

\section{Feasibility of communication with adversary identification}\label{sec:feasibility}
%\noindent\textcolor{blue}{\textbf{Note to the reviewers:} Two typos, one in \eqref{eq:spoof1} and another in the paragraph after Definition~\ref{defn:spoof}, have been corrected ($y$ changed to $\tilde{y}$). The changes are written in blue.}
\begin{defn}\label{defn:spoof}
A MAC \mch is \one-{\em \spoofable} if there exist distributions ${Q_{Y|\tilde{X},\tilde{Y}}}$ and ${Q_{X|\tilde{X},X'}}$ such that $\forall\,x', \,\tilde{x},\, \tilde{y},\, z,$
\begin{align}\label{eq:spoof1}
&\sum_{y}Q_{Y|\tilde{X},\tilde{Y}}(y|\tilde{x},\tilde{y})\mch(z|x',y) \nonumber\\
&= \sum_{y}Q_{Y|\tilde{X},\tilde{Y}}(y|x',\tilde{y})\mch(z|\tilde{x},y) \nonumber\\
& = \sum_{x}Q_{X|\tilde{X},X'}(x|\tilde{x},x')\mch(z|x,{\tilde{y}}).
\end{align}

A MAC \mch is \two-{\em \spoofable} \longonly{({see Fig.~\ref{fig:spoof2}.})} if there exist distributions ${Q_{X|\tilde{X},\tilde{Y}}}$ and ${Q_{Y|\tilde{Y},Y'}}$ such that $\forall\,\tilde{x},\, \tilde{y}, \,y', \, z,$
\begin{align}\label{eq:spoof2}
&\sum_{x}Q_{X|\tilde{X},\tilde{Y}}(x|\tilde{x},\tilde{y})\mch(z|x,y') \nonumber\\
&= \sum_{x}Q_{X|\tilde{X},\tilde{Y}}(x|\tilde{x},y')\mch(z|x,\tilde{y}) \nonumber\\
& = \sum_{y}Q_{Y|\tilde{Y},Y'}(y|\tilde{y},y')\mch(z|\tilde{x},y).
\end{align}
A MAC is {\em \spoofable} if it is either \one- or \two-\spoofable. 

\end{defn}

{%We will argue that for an \one-\spoofable channel, communication is not possible by user-\one. 
When \eqref{eq:spoof1} holds, for a triple $(\vecx', \tilde{\vecx}, \tilde{\vecy})\in \cX^n\times\cX^n\times\cY^n$, the output distributions in the following three cases are the same (see Fig.~\ref{fig:spoof1}.): (a) User \one sends $\vecx'$ and user \two sends $\vecY\sim Q^n_{Y|\tilde{X},\tilde{Y}}(.|\tilde{\vecx},\tilde{\vecy})$, i.e., $\vecY$ is distributed as the output of the memoryless channel $Q_{Y|\tilde{X},\tilde{Y}}$ on inputs $\tilde{\vecx}$ and $\tilde{\vecy}$; (b) User \one sends $\tilde{\vecx}$ and user \two sends $\vecY\sim Q^n_{Y|\tilde{X},\tilde{Y}}(.|\vecx',\tilde{\vecy})$; (c) User \two sends $\tilde{\vecy}$ and user \one sends $\vecX\sim Q^n_{X|\tilde{X},X'}(.|\tilde{\vecx},\vecx')$.
 Hence, for a given code $(f_{\one}, f_{\two}, \phi)$ and independent $M_{\one}\sim\textsf{Unif}(\cM_{\one})$, $M'_{\one}\sim\textsf{Unif}(\cM_{\one})$ and $M_{\two}\sim\textsf{Unif}(\cM_{\two})$, 
 the output distributions in the following three cases are the same: (a) User \one is honest and sends $f_{\one}(M_{\one})$ and user \two is adversarial and attacks with $\vecY\sim Q^n_{Y|\tilde{X},\tilde{Y}}(.|f_{\one}(M_{\one}'),f_{\two}(M_{\two}))$; (b) User \one is honest and sends $f_{\one}(M'_{\one})$ and user \two is adversarial and attacks with $\vecY\sim Q^n_{Y|\tilde{X},\tilde{Y}}(.|f_{\one}(M_{\one}),f_{\two}(M_{\two})$; (c) User \two is honest and sends $f_{\two}(M_{\two})$ and user \one is adversarial and attacks with $\vecX\sim Q^n_{X|\tilde{X},X'}(.|f_{\one}(M_{\one}),f_{\one}(M_{\one}'))$. 
Thus, the decoder cannot determine the adversarial user reliably, nor can it differentiate between $M_{\one}$ and $M'_{\one}$ as the input of user \one. }
In\shortonly{~\cite{lver}}\longonly{ Lemma~\ref{thm:converse}}, we formally argue that for an \one-\spoofable MAC, no non-zero rate can be achieved for user-\one.

\iffalse
Consider \eqref{eq:spoof1} as shown in Fig.~\ref{fig:spoof1}. There are two attack strategies, ${Q_{X|\tilde{X},X'}}$ by user \one and ${Q_{Y|\tilde{X},\tilde{Y}}}$ by user \two such that the output distributions in the following three cases are the same: (a) User \one is honest and sends $x'$ and user \two is adversarial and attacks with $Q_{Y|\tilde{X},\tilde{Y}}(.|\tilde{x},\tilde{y})$; (b) User \one is honest and sends $\tilde{x}$ and user \two is adversarial and attacks with $Q_{Y|\tilde{X},\tilde{Y}}(.|x',\tilde{y})$; (c) User \two is honest and sends $\tilde{y}$ and user \one is adversarial and attacks with $Q_{X|\tilde{X},X'}(.|\tilde{x},x')$. 
Since the output is consistent with these attack strategies of both the users, the decoder cannot determine the adversarial user reliably, nor can it differentiate between $x'$ and $\tilde{x}$ as the input of user \one (see Fig.~\ref{fig:spoof1a} and \ref{fig:spoof1b}). 
In\shortonly{~\cite{lver}}\longonly{ Lemma~\ref{thm:converse}}, we formally argue that for an \one-\spoofable MAC, no non-zero rate can be achieved for user-\one.
\fi

Our first result states that, in fact, non-\spoofability characterizes the MACs in which users can work at positive rates of communication with adversary identification.  
\begin{thm}\label{thm:main_result}
If a MAC is \one-\spoofable (resp. \two-\spoofable), communication with adversary identification from user-\one (resp. user-\two) is impossible. Specifically, for any $(N_{\one}, N_{\two}, n)$ adversary identifying code with $N_{\one}\geq 2$ (resp. $N_{\two}\geq 2$), the probability of error is at least $1/12$.  If a MAC is neither \one-\spoofable nor \two-\spoofable, then its capacity region has a non-empty interior ($\mathsf{int}(\cC)\neq \emptyset)$, that is, both users can communicate reliably with adversary identification at positive rates. 
\end{thm}
\longonly{The proof of the theorem is given in Appendix~\ref{sec:proof_thm1}.}
\begin{corollary}
$\mathsf{int}(\cC) = \emptyset$ if and only if a MAC is \spoofable.
\end{corollary}
\begin{remark}
Theorem~\ref{thm:main_result} does not cover the case when exactly one user is \spoofable. In particular, if the MAC is \one-\spoofable (and thus, $C^{}_{\one} = 0$), but not \two-\spoofable, can $C^{}_{\two}>0$? A similar case is also open for Arbitrarily Varying Multiple Access MAC (AV-MAC) (see \cite{AhlswedeC99}). When encoders have private randomness, this can be resolved as was recently shown by Pereg and Steinberg \cite{PeregS19}. A similar resolution is possible for the present problem. 
We can use encoders with private randomness to show that  $C^{}_{\one} >0$ (resp. $C^{}_{ \two} >0$) if and only if the MAC is not \one-\spoofable\ (resp. not \two-\spoofable). 
\end{remark}
%\vspace{-1 cm}
In the interest of space, we limit the discussion of achievability to an informal description of the decoder. See \shortonly{\cite{lver}}\longonly{Lemma~\ref{thm:detCodesPositivity}} for a complete proof.  
For input distributions $P_{\one}$ and $P_{\two}$ on \cX\ and \cY\ respectively, the decoder works by collecting potential candidates for the messages sent by each user. 
A message $\mo$ is deemed a {\em candidate} for user \one if it is typical with some (attack) vector $\vecy$ and the output vector $\vecz$ according to the channel law (i.e., for some $\eta>0$, $\inp{f_{\one}(m_{\one}), \vecy, \vecz} \in T^{n}_{XYZ}$ such that $D\inp{P_{XYZ}||P_{\one}P_YW}\leq\eta$). 
We further prune the list of candidates by only keeping the ones which can account for all other candidates that can lead to ambiguity at the decoder. 
For example, for a candidate $\mo$, suppose there are two other candidates $\tilde{m}_{\one}$ and $\tilde{m}_{\two}$ of user \one and user \two respectively. The decoder is confused between \mo and $\tilde{m}_{\one}$, so it cannot reliably choose an output message for user \one. Neither can it adjudge one of the users to be adversarial as both users have valid message candidates. 
In order to get around this, we require that for  every pair of candidates ($\tilde{m}_{\one},\tilde{m}_{\two}$) such that $\inp{f_{\one}(\mo), \vecy,  f_{\one}(\tilde{m}_{\one}), f_{\two}(\tilde{m}_{\two}), \vecz}$$\in $$T^n_{XY\tilde{X}\tilde{Y}Z}$, the condition  $I(\tilde{X}\tilde{Y};XZ|Y)<\eta$ holds.   %Thus, to assert the candidacy of $\mo$, we need to ensure that it can explain away\footnote{we can probably use this phrase here for this informal description} every pair of candidates ($\tilde{m}_{\one},\tilde{m}_{\two})\in \cM_{\one}\times\cM_{\two}$ which pass check 1. 
%In particular, for $\inp{f_{\one}(\mo), \vecy,  f_{\one}(\tilde{m}_{\one}), f_{\two}(\tilde{m}_{\two}), \vecz}$$\in $$T^n_{XY\tilde{X}\tilde{Y}Z}$, we require the Markov chain $XZ-Y-\tilde{X}\tilde{Y}$ to nearly hold (i.e. $I(\tilde{X}\tilde{Y};XZ|Y)<\eta$). 
%In other words, $XZ$ is almost conditionally independent of $\tilde{X}\tilde{Y}$ given $Y$. 
{Under this condition, we may infer that the channel output \vecz\ was likely not caused by the pair $(\tilde{m}_{\one}$, $\tilde{m}_{\two}$),
rather, $(\tilde{m}_{\one}$, $\tilde{m}_{\two})$ is more likely to be part of the attack strategy employed by user \two to produce its input  vector \vecy.}
Similarly, if there  is a pair of candidates $(\tilde{m}_{\two 1},\tilde{m}_{\two 2})$ of user \two,  the decoder can neither reliably decode user \two's message, nor can it implicate either user. 
Thus, we require that for every pair of candidates ($\tilde{m}_{\two 1}$,$\tilde{m}_{\two 2}$) of user \two such that $\inp{f_{\one}(\mo), \vecy, f_{\two}(\tilde{m}_{\two 1}), f_{\two}(\tilde{m}_{\two 2}),\vecz}$$\in $$T^n_{XY\tilde{Y}_1\tilde{Y}_2Z}$, the condition $I(\tilde{Y}_1\tilde{Y}_2;XZ|Y)< \eta$ holds. Let $D_{\one}(\eta, \vecz)$ be the set of all candidates of user \one which pass these checks. 
%For a parameter $\eta$$>$$0$, let $\cD_{\eta}$ be the set of joint distributions defined as
%$\cD_\eta$ $\defineqq$  $\inb{P_{XYZ}\in \cP^n_{\cX\times \cY \times \cZ}:\, D\inp{P_{XYZ}||P_XP_YW}\leq\eta }$.
%For the  received output sequence $\vecz$, let $D_{\one}(\eta, \vecz)$ be defined as the set of messages $m_{\one}\in\mathcal{M}_{\one}$ such that there exists $\vecy\in \cY^n$ satisfying the following conditions:
%\begin{enumerate}[label=(\roman*)]
%\item $\inp{f_{\one}(m_{\one}), \vecy, \vecz} \in T^{n}_{XYZ}$ for some $P_{XYZ}\in \cD_{\eta}.$ \label{check:1}
%\item If there exists $(\tilde{m}_{\one},\, \tilde{m_{\two}})$$\in$$ \cM_{\one}\times \cM_{\two}, \, \tilde{m}_{\one}\neq m_{\one}$ and  $(\vecy', \, \vecx')\in \cY^n\times \cX^n$ such that $\inp{f_{\one}(\mo), \vecy, f_{\one}(\tilde{m}_{\one}), \vecy',\vecx', f_{\two}(\tilde{m}_{\two}), \vecz}$$\in $$T^n_{XY\tilde{X}Y'X'\tilde{Y}Z}$, $P_{\tilde{X}Y'Z}$$\in$$ \cD_{\eta}$ and  $P_{X'\tilde{Y}Z}$$\in $$\cD_{\eta}$, then $I(\tilde{X}\tilde{Y};XZ|Y)<\eta$. 
%\item If there exists $\tilde{m}_{\two 1},\, \tilde{m}_{\two 2}$$\in$$\cM_{\two}$, $\tilde{m}_{\two 1}\neq \tilde{m}_{\two 2}$ and $\vecx'_1, \, \vecx'_2$$\in$$ \cX^n$ such that $\inp{f_{\one}(\mo), \vecy,\vecx'_1, f_{\two}(\tilde{m}_{\two 1}),\vecx'_2, f_{\two}(\tilde{m}_{\two 2}),\vecz}$$\in $$T^n_{XYX'_1\tilde{Y}_1X'_2\tilde{Y}_2Z}$, $P_{X'_1\tilde{Y}_1Z}$$\in$$ \cD_{\eta}$ and  $P_{X'_2\tilde{Y}_2Z}$$\in$$ \cD_{\eta}$, then $I(\tilde{Y}_1\tilde{Y}_2;XZ|Y)< \eta$.\label{check:3}
%\end{enumerate}
{We define $D_{\two}(\eta, \vecz)$ analogously by interchanging the roles of users \one and \two.}
The decoder is as follows:
\begin{align*}
\phi(\vecz) \defineqq\begin{cases}(\mo,\mt) &\text{if }D_{\one}(\eta, \vecz)\times D_{\two}(\eta, \vecz) = \{(\mo,\mt)\},\\ \oneb \text{ (\small blame \one)} &\text{if }|D_{\one}(\eta, \vecz)| = 0, \, |D_{\two}(\eta, \vecz)| \neq 0,\\ \twob\text{ (\small blame \two)} &\text{if }|D_{\two}(\eta, \vecz)| = 0, \, |D_{\one}(\eta, \vecz)| \neq 0,\\(1,1) &\text{if }|D_{\one}(\eta, \vecz)| = |D_{\two}(\eta, \vecz)| = 0.
\end{cases}
\end{align*}In the spirit of \cite[Lemma 4]{CsiszarN88}, we show in \shortonly{\cite{lver}}\longonly{Appendix~\ref{sec:proof_thm1}} that for a non-\spoofable MAC, there exists a  small enough $\eta>0$ such that if $|D_{\one}(\eta, \vecz)|, |D_{\two}(\eta, \vecz)| >0$ then $|D_{\one}(\eta, \vecz)|$ = $|D_{\two}(\eta, \vecz)|$ = 1. Thus, the decoder definition covers all the cases. We also show that $|D_{\one}(\eta, \vecz)|$ = $|D_{\two}(\eta, \vecz)| = 0$ is a low probability event.  In\shortonly{~\cite{lver}}\longonly{ Appendix~\ref{sec:proof_thm1}}, we analyze the error probability of the decoder and show that for non-\spoofable channels it can support positive rates for both users.

\section{Capacity region}
\subsection{Inner bound}
For distributions $P_{\one}$ and $P_{\two}$  over $\cX$ and $\cY$ respectively, we define  $\cP(P_{\one}, P_{\two}) \defineqq \{P_{XY\tilde{X}\tilde{Y}Z}: P_{X\tilde{Y}Z}=P_{\one}\times P_{\tilde{Y}}\times W \text{ for some }P_{\tilde{Y}} \text{ and }P_{\tilde{X}YZ} = P_{\tilde{X}}\times P_{\two}\times W \text{ for some }P_{\tilde{X}}\}$.
Let $\cR_{1}(P_{\one}, P_{\two})$ be the set of rate pairs $(R_{\one}, R_{\two})$ such that
\begin{align}\label{eq:inner_bd_1}
R_{\one}&\leq \min_{P_{XY\tilde{X}\tilde{Y}Z} \in \cP(P_{\one}, P_{\two})} I(X;Z)\nonumber\\
R_{\two}&\leq \min_{P_{XY\tilde{X}\tilde{Y}Z} \in \cP(P_{\one}, P_{\two}):X\independent Y} I(Y;Z|X).
\end{align} 
Similarly, let $\cR_{2}(P_{\one}, P_{\two})$ be the set of rate pairs given by
\begin{align}\label{eq:inner_bd_2}
R_{\one}&\leq \min_{P_{XY\tilde{X}\tilde{Y}Z} \in \cP(P_{\one}, P_{\two}):X\independent Y} I(X;Z|Y)\nonumber\\
R_{\two}&\leq \min_{P_{XY\tilde{X}\tilde{Y}Z} \in \cP(P_{\one}, P_{\two})} I(Y;Z).
\end{align}

\begin{thm}[Achievable rate region]\label{thm:inner_bd}
When $\mathsf{int}(\cC)\neq \emptyset$, 
\begin{align*}
\mathsf{conv}(\cup_{P_{\one}, P_{\two}}\inp{\cR_1(P_{\one}, P_{\two})\cup\cR_{2}(P_{\one}, P_{\two})})\subseteq \cC.
\end{align*}

\end{thm}
The proof uses a slighly modified version of the decoder used in Theorem~\ref{thm:main_result}. This modification imposes the additional condition $X\indep Y$ on the distribution. \shortonly{Please refer to~\cite{lver} for details.}\longonly{Please see Appendix~\ref{sec:inner_bd_proof}}.

\subsection{Outer bound}\label{sec:outer_bd}
The outer bound is provided in terms of the capacity of an Arbitrarily Varying Multiple Access Channel (AV-MAC).
An AV-MAC $\cW = \{W(z|x,y,s), (x, y, z)\in \cX\times\cY\times\cZ: s\in \cS\}$$\subseteq \bbR^{|\cX|\times|\cY|\times|\cZ|}$ is a family of MACs parameterized by the set of state symbols $\cS$ (see \cite{Jahn81}). The state of an AV-MAC can vary arbitrarily during the transmission. We use $\cC^{}_{\AVMAC}(\cW)$ (or simply $\cC^{}_{\AVMAC}$) to denote the deterministic capacity region of an AV-MAC $\cW$. 
\begin{defn} \label{defn:outerboundAVMAC}
For a MAC $W$, let  $\tilde{\cW}_W$ be the set of MACs $\tilde{W}$ such that for some distributions $Q_{X'|X}$ and $Q_{Y'|Y}$ on $\cX\times\cX$ and $\cY\times\cY$ respectively and 
for all $x, y, z\in \cX\times\cY\times \cZ$,
\begin{align}\label{eq:outer_bound}
\tilde{W}(z|x,y) &= \sum_{x'}Q_{X'|X}(x'|x)W(z|x',y) \nonumber\\
&= \sum_{y'}Q_{Y'|Y}(y'|y)W(z|x,y')
\end{align}
\end{defn}

Notice that $W\in \tilde{\cW}_W$ by choosing trivial distributions $Q_{X'|X}(x|x) = 1$ for all $x$ and $Q_{Y'|Y}(y|y) = 1$ for all $y$.
The set $\tilde{\cW}_W$ is convex because for every $(Q_{X'|X}, Q_{Y'|Y})$ and $(Q'_{X'|X}, Q'_{Y'|Y})$  satisfying \eqref{eq:outer_bound}, the pair $(\alpha Q_{X'|X}+ (1-\alpha) Q'_{X'|X}, \alpha Q_{Y'|Y} + (1-\alpha) Q'_{Y'|Y})$, $\alpha \in [0,1]$ also satisfies \eqref{eq:outer_bound}.
To get an outer bound, let us consider a situation where user \one is malicious and attacks in the following manner: it runs its encoder on a uniformly distributed message from its message set, then passes the output of the encoder through $\prod_{i=1}^{n}Q^i_{X'|X}$ where for all $i\in [1:n]$, $(Q^i_{X'|X}, Q^i_{Y'|Y})$ satisfy \eqref{eq:outer_bound} for some $Q^i_{Y'|Y}$. The output of $\prod_{i=1}^{n}Q^i_{X'|X}$ is finally sent to the MAC ${W}$ as input by user \one. 
At the receiver, it is not clear if user \one attacked using $\prod_{i=1}^{n}Q^i_{X'|X}$ or user \two attacked using $\prod_{i=1}^{n}Q^i_{Y'|Y}$. Hence, the malicious user cannot be identified reliably. So, the decoder must output a pair of messages. This implies that the capacity region $\cC$ must be a subset of the capacity region of the AV-MAC $\tilde{\cW}_W$ (Definition~\ref{defn:outerboundAVMAC}) parametrized by a pair of distributions $(Q_{X'|X}, Q_{Y'|Y})$ satisfying \eqref{eq:outer_bound}. 
This argument is formalized in \shortonly{the longer version~\cite{lver}.}\longonly{Appendix~\ref{sec:outer_bd_proof}}
The outer bound obtained in this manner is valid for any protocol: deterministic, stochastic (private randomness at the encoders) or randomized (independent randomness shared by each encoder with the decoder). 
\begin{thm}[Outer bound]\label{thm:outer_bd}
$\cC \subseteq \cC_{\AVMAC}(\tilde{\cW}_W)$. Moreover, there exists 
an AV-MAC ${\cW}_W$ such that $\cC_{\AVMAC}({\cW}_W)=\cC_{\AVMAC}(\tilde{\cW}_W)$ 
%a set ${\cW}_{W}$ such that $\mathsf{conv}({\cW}_W) = \mathsf{conv}(\tilde{\cW}_W)$ 
and $|{\cW}_W| \leq 2^{|\cX|^2 + |\cY|^2}$.
\end{thm}
%\noindent The existence of a finite set ${\cW}_{W}$ can be shown by simple geometric arguments \cite{lver}.
\noindent The existence of an AV-MAC ${\cW}_W$ with a finite state-space can be shown using the fact that the $\cC_{\AVMAC}(\cW)$ only depends on $\mathsf{conv}(\cW)$ and by simple geometric arguments\shortonly{~\cite{lver}.}\longonly{ (see Appendix~\ref{sec:outer_bd_proof}).}
\begin{remark}
Theorem~\ref{thm:outer_bd} also gives an outer bound for the capacity region under randomized codes (with independent randomness shared between each encoder and the decoder). \shortonly{See \cite{lver} for details.}
\end{remark}
\section{Examples and comparison with other models}
\subsection{Tightness of the inner bound for the Binary Erasure MAC}\label{sec:example_tightness}
We will show that for the binary erasure MAC~\cite[pg.~83]{YHKEG}, the inner bound on \cC\ given by Theorem~\ref{thm:inner_bd} is the same as its (non-adversarial) capacity region $\cC_{\MAC}$. Hence, it is tight.  We choose $P_{\one}$ and $ P_{\two}$ arbitrarily close to the uniform distribution $U$ on $\{0, 1\}$ while ensuring that  $P_{\one}\neq P_{\two}$. We show that $\cP(P_{\one}, P_{\two}) = \{P_{XY\tilde{X}\tilde{Y}Z}: P_{X\tilde{Y}Z}=P_{\one}\times P_{\two}\times W  \text{ and }P_{\tilde{X}YZ} = P_{\one}\times P_{\two}\times W \}$ and for $P_{XY\tilde{X}\tilde{Y}Z}\in \cP(P_{\one}, P_{\two})$ satisfying $X\indep Y$, $\tilde{X} = X$ and $\tilde{Y} = Y$.  Thus, \eqref{eq:inner_bd_1} evaluates to $R_{\one}\leq 0.5$ and $R_{\two} \leq 1$, and \eqref{eq:inner_bd_2} evaluates to $R_{\one}\leq 1$ and $R_{\two} \leq 0.5$. Using time sharing between these two rate pairs, we obtain the entire MAC region (This is the rate region $\cC$ in Fig.~\ref{ex:3}).  Please refer to \shortonly{the longer version~\cite{lver}}\longonly{Appendix~\ref{sec:examples}} for a complete argument.
\subsection{Comparison with related models}\label{sec:comparison}
In this section we contrast the present model with reliable communication and authenticated communication models.

\paragraph{Reliable communication in a MAC with byzantine users}\label{para:reliable} 

\longonly{We consider a MAC with a stronger decoding guarantee: the decoder, w.h.p,  outputs a message pair of which the message(s) of honest user(s) is correct. 
In the presence of a malicious user, the channel from the honest user to the receiver can be treated as an Arbitrarily Varying MAC (AVC) \cite{BBT60} with the input of other user as state. Thus, the capacity region is outer bounded by the rectangular region defined by the AVC capacities of the two users’ channels. Further, it is easy to see that this outer bound is achievable when both users use the corresponding AVC codes.
Csisz\'ar and Narayan show in \cite{CsiszarN88} that the capacity of an AVC is zero iff it is {\em symmetrizable}.}
\shortonly{We continue the discussion from the introduction (picking up from footnote~\ref{ftn:reliable} in page~\pageref{ftn:reliable}).} Communication is infeasible in an AVC if and only if it is {\em symmetrizable}\cite{CsiszarN88}. Translating this to the two-user MAC, we define a MAC to be {\em \two-symmetrizable} if there exists a distribution $P_{X|Y}$ such that
\begin{align}\label{eq:symmetrizable}
\sum_{x'\in \cX}P_{X|Y}(x|y')W(z|x,y) = \sum_{x'\in \cX}P_{X|Y}(x|y)W(z|x,y')
\end{align}
for all $(x, y, z)\in \cX\times\cY\times\cZ$. We define an \one-symmetrizable MAC analogously. A \textbf{symmetrizable} MAC is one which is either \one- or \two-symmetrizable. 
Thus, reliable communication by both users is feasible in a MAC if and only if it is not symmetrizable. We denote the reliable communication capacity of a MAC by $\cC_{\reliable}$.
%Using the results from\cite{Gubner, CsiszarN88}, we conclude that the reliable communication capacity region $\cC_{\reliable}$ of a MAC has non-empty interior if and only if it is not symmetrizable. %This model was extended in \cite{NehaBDPITW19} to a three-user MAC.
\paragraph{Authenticated communication in a MAC with byzantine users \cite{NehaBDPISIT19}}\label{para:weak} 
\longonly{This model considers a MAC with a weaker decoding guarantee: the decoder should reliably decode the messages when both users are honest. When one user is adversarial, the decoder either outputs a pair of messages of which the message of honest user is correct or it declares the presence of an adversary (without identifying it). 
In this case, the notion of an overwritable MAC characterizes the MACs with non-empty capacity region $\cC_{\auth}$ of authenticated communication.} {We} say that a MAC is {\em \two-overwritable} \cite[(1)]{NehaBDPISIT19} if there exists a distribution $P_{X'|X,Y}$ such that
\begin{align}\label{eq:overwritable}
\sum_{x'\in\cX}P_{X'|X,Y}(x'|x,y)W(z|x',y') = W(z|x,y) 
\end{align}
for all $y,y'\in\cY, \, x\in \cX$ and $z\in \cZ$. Similarly, we can define an \one-overwritable MAC.  If a MAC is either \one-~or \two-overwritable, we say that the MAC is \textbf{overwritable}. Authenticated communication by both users is not feasible in an overwritable MAC. Theorem~1 in \cite{NehaBDPISIT19} states that if the MAC is not overwritable, then authenticated communication capacity, $\cC_{\auth} = \cC_{\MAC}$.

\begin{prop}\label{prop:inclusions}
All overwritable MACs are \spoofable and all \spoofable MACs are symmetrizable. Furthermore, both these inclusions are strict.
\end{prop}
\noindent While the inclusions in Proposition~\ref{prop:inclusions} are obvious from the problem definitions and the feasibility results, we nonetheless provide a direct argument. Suppose a MAC is \two-overwritable with $P_{X'|X,Y}$ as the overwriting attack in \eqref{eq:overwritable}. For any distribution $Q_{Y}$ on \cY, let $Q_{X|\tilde{X}, \tilde{Y}}(x|\tilde{x}, \tilde{y})\defineqq \sum_{y}Q_{Y}(y)P_{X'|X,Y}(x|\tilde{x},y)$ for all $x, \tilde{x}, \tilde{y}$ and $Q_{Y|\tilde{Y},Y'}(y|\tilde{y},y')\defineqq Q_{Y}(y)$ for all $y,\tilde{y},y'$. Distributions $Q_{X|\tilde{X}, \tilde{Y}}$ and $Q_{Y|\tilde{Y},Y'}$ as defined satisfy \eqref{eq:spoof2}. Now, suppose a MAC \mch is \two-\spoofable with attacks $Q_{X|\tilde{X}, \tilde{Y}}$ and $Q_{Y|\tilde{Y},Y'}$ satisfying \eqref{eq:spoof2}. For all, $x, y$, let $P_{X|Y}(x|y)\defineqq Q_{X|\tilde{X}, \tilde{Y}}(x|\tilde{x},y)$ for any $\tilde{x}\in \cX$. It can be easily seen that the attack $P_{X|Y}$ as defined satisfies \eqref{eq:symmetrizable}.
Examples~\ref{ex:1} and~\ref{ex:2} below show strict inclusion (see Fig. \ref{fig:conditions}).
 \begin{example}[symmetrizable, but not \spoofable]\label{ex:1} Binary erasure MAC: {\em It has binary inputs $X, Y$ and outputs $Z=X+Y$ where $+$ is real addition. We show in \shortonly{the longer version~\cite{lver}}\longonly{Appendix~\ref{sec:BEC_not_spoof}} that this channel {is} not \spoofable. To show symmetrizability, we note that the distribution $P_{X|Y}(x|y) = 1$ for all $x=y$ is a symmetrizing attack in \eqref{eq:symmetrizable}.}
 \end{example}
\begin{example}[\spoofable, but not overwritable] \label{ex:2} Binary additive MAC: {\em It has binary inputs $X, Y$ and outputs $Z = X\oplus Y$ where $\oplus$ is the XOR operation. {To show \spoofability, note that the attacks ${Q_{X|\tilde{X},X'}}(x|\tilde{x}, x') = 1/2$ for all $x, \tilde{x}$ and $x'$, and ${Q_{Y|\tilde{X},\tilde{Y}}}(y|\tilde{x}, \tilde{y}) = 1/2$ for all $y, \tilde{x}$ and $\tilde{y}$, satisfy \eqref{eq:spoof1} because they result in the same uniform output distribution over $\cZ$ in all the three cases in \eqref{eq:spoof1}.}
% in all the three cases in Fig.~\ref{fig:spoof1}. 
We show in \shortonly{the longer version~\cite{lver}}\longonly{Appendix~\ref{sec:BAC_not_over}} that Example~\ref{ex:2} is not overwritable.}\end{example} 
\vspace{-0.2 cm}
\begin{figure}[h]
\centering
\begin{tikzpicture}[scale=0.45]
\draw (0,0) ellipse (6cm and 4cm) node[yshift = -1.2 cm, xshift = -0.7 cm]{\footnotesize symmetrizable} node[yshift = -1.5 cm, xshift = -0.7 cm]{\footnotesize MACs} ;%node[yshift = -1.2 cm, xshift = 1.5 cm]{\footnotesize $\cdot Z = X+Y$} ;
\draw (0.2,0.6) ellipse (4cm and 2.5cm) node[yshift = -0.6 cm, xshift = 0 cm]{\footnotesize spoofable };
\draw (0.2,0.6) ellipse (4cm and 2.5cm) node[yshift = -0.9 cm, xshift = 0 cm]{\footnotesize  MACs};
\draw (0.3, 1) ellipse (2.5 cm and 1 cm) node[xshift = -0.1 cm, yshift = 0.1 cm]{\footnotesize overwritable} node[yshift = -0.2 cm]{\footnotesize MACs};

\node[label={[label distance=0.01mm]80:\footnotesize \ref{ex:1}}] (A) at (1,-3.1){};
\node[label={[label distance=0.01mm]80:\footnotesize \ref{ex:2}}] (B) at (2.9,-0.4){};
%\node[label={[label distance=0.01mm]90:\footnotesize 3}] (C) at (2,0.7){};
\node[label=below:] (B2) at (6.8,0)[right]{\footnotesize Binary additive MAC};
\node[label=below:] (A1) at (3.5,-2.8){};
\node[label=below:] (A2) at (5,-2.1){};
\node[label=below:] (A3) at (7,-2)[right]{\footnotesize Binary erasure MAC};
\node[label=below:] (C3) at (7.5,2.5){};
\node[label=below:] (C1) at (3,2){};
\node[label=below:] (C2) at (5,2.3){};
%\node[label=below:$x_2$] (B) at (3,1){};
%\node[label=below:$x_1$] (C) at (1.25,0.65) {};
%\node[label=below:$\alpha$] (D) at (0.55,0.45) {};
%\node[label=below:$\beta$] (E) at (2,0.45) {};

%\draw[thick] plot [smooth, tension=0.75] coordinates {(A)     (B)};
%\draw[->] (C).. controls (C1) and (C2)..(C3);
\draw[->] (B)--(B2);
\draw[->] (A) .. controls (A1) and (A2) .. (A3);
\draw [fill=black] (A) circle (1.5pt);
\draw [fill=black] (B) circle (1.5pt);
%\draw [fill=black] (C) circle (1.5pt);  
\end{tikzpicture}\caption{The set of overwritable MACs is a strict subset of the set of \spoofable MACs which, in turn, is a strict subset of the set of symmetrizable MACs.}\label{fig:conditions} \end{figure}
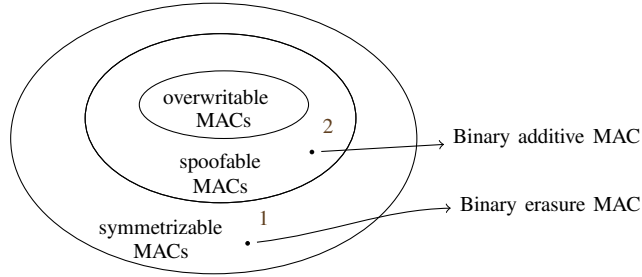
\vspace{0.3 cm}
We also note from the problem definitions that $\cC_{\reliable}\subseteq\cC\subseteq\cC_{\auth}\subseteq\cC_{\MAC}$.
Next we give an example of a channel for which $\cC_{\reliable}$, $\cC$ and $\cC_{\auth}$ are distinct.
The example is constructed by using the MACs in Examples~\ref{ex:1} and \ref{ex:2} in parallel. 
\begin{example}[$(Z_1,Z_2)=(X_1+Y_1, X_2\oplus Y_2)$]\label{ex:3}
For binary inputs $X_1, X_2, Y_1, Y_2$, the output $(Z_1,Z_2) = (X_1+Y_1,X_2\oplus Y_2)$.
 \end{example}
\noindent The channels  $Z_1=X_1+Y_1$ and $Z_2=X_2\oplus Y_2$ are both non-overwritable and symmetrizable. Since the MACs do not interact when used in parallel, we can show that the resultant MAC $(Z_1,Z_2)=(X_1+Y_1, X_2\oplus Y_2)$ is also non-overwritable and symmetrizable (see \shortonly{\cite{lver}}\longonly{Appendix~\ref{sec:ex3}}). Thus, $\cC_{\reliable} = \{0,0\}$ and  $\cC_{\auth} = \cC_{\MAC}$. 
To compute $\cC$, we note that the pair {$(Q_{X'|X}, Q_{Y'|Y})$} defined by $Q_{X'|X}((x_1, u)|(x_1,x_2)) = 0.5$ for all $u,x_1, x_2\in \{0,1\}$ and $Q_{Y'|Y}((y_1, v)|(y_1,y_2)) = 0.5$ for all $v,y_1, y_2\in \{0,1\}$ {satisfies} the conditions in~\eqref{eq:outer_bound}. The resulting channel $\tilde{W}$ has the same first component as $W$ (i.e., a binary erasure MAC) and a second component whose output $Z_2$ is independent of the inputs. By Theorem~\ref{thm:outer_bd}, $\cC$ is outer bounded by the (non-adversarial) capacity region of $\tilde{W}$ which is the capacity region of the binary erasure MAC. We can show that this outer bound is tight by using an adversary identifying code for the binary erasure MAC component $Z_1=X_1+Y_1$ (see Section~\ref{sec:example_tightness}) and any arbitrary inputs for the other component. Please see \shortonly{\cite{lver}}\longonly{Appendix~\ref{sec:ex3}} for details. The capacity regions under these three models are plotted in Fig.~\ref{fig:ex3}.
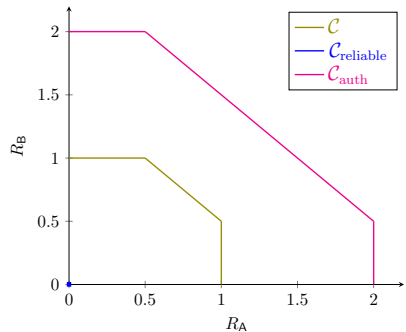
\begin{figure}[h]
\centering
\begin{tikzpicture}[scale=0.65]
\begin{axis}[
    legend style={font=\large},
    axis lines = left,
    xlabel = {$R_{\one}$},
    ylabel = {$R_{\two}$},
]

\addplot [thick,
    domain=0:0.5, 
    samples=50, smooth, 
    color=olive,
    ]
    {1};
\addlegendentry{\textcolor{olive}{\large $\cC$\phantom{abcde}}}
\addplot[thick,
    domain=0:0.5, 
    samples=50, smooth,
    color=blue,
    %mark=square*,
    %mark size=1pt
    ]
    coordinates {
    (0,0)};
\addlegendentry{\textcolor{blue}{\large $\cC_{\reliable}$}}

\addplot [thick,
    domain=0:0.5, 
    samples=50, smooth, 
    color=magenta,
]
{2};
\addlegendentry{\textcolor{magenta}{\large $\cC_{\auth}\phantom{ab}$}}
\addplot [thick,
    domain=0.5:1, 
    samples=50, smooth, 
    color=olive,
    ]
    {1.5-x};
\addplot [thick,
    domain=0.5:2, 
    samples=50, smooth, 
    color=magenta,
]
{2.5-x};
\addplot[thick, 
	samples=50, smooth,
	domain=0:6,
	olive] 
	coordinates {(1,0)(1,0.5)};
\addplot[thick, 
	samples=50, smooth,
	domain=0:6,
	magenta] 
	coordinates {(2,0)(2,0.5)};
\addplot[thick, 
	samples=50, smooth,
	domain=0:6,
	white] 
	coordinates {(2.1,0.1)(2.2,0.1)};
\addplot[thick, 
	samples=50, smooth,
	domain=0:6,
	white] 
	coordinates {(0.1,2.1)(0.2,2.2)};
\addplot[
    color=blue,
    mark=square*,
    mark size=1pt
    ]
    coordinates {
    (0,0)};
\end{axis}
\end{tikzpicture}\caption{Capacity regions for the MAC in Example~\ref{ex:3}: $\cC_{\reliable} = \{0,0\}$; $\cC = \cC_{\MAC}$ of $Z_1 = X_1+Y_1$; and $\cC_{\auth} = \cC_{\MAC}$ of $(Z_1,Z_2)=(X_1+Y_1, X_2\oplus Y_2)$. }\label{fig:ex3}\end{figure}

\longonly{%\onecolumn
\linespread{1.5}
\appendices
\section{Proof of Theorem~\ref{thm:main_result}}\label{sec:proof_thm1}
We first prove the converse.
\begin{lemma}\label{thm:converse}
 If a channel is \one-spoofable (resp. \two-\spoofable), then  for any $(N_{\one}, N_{\two}, n)$ strongly authenticating code with $N_{\one}\geq 2$ (resp. $N_{\two}\geq 2$), the probability of error is at least $1/12$.  
 \end{lemma}
\begin{proof}
The proof uses ideas from proof of \cite[Lemma 1, page 187]{CsiszarN88}. Suppose the channel satisfies \eqref{eq:spoof1}. A similar analysis can be done when channel satisfies \eqref{eq:spoof2}. Let $Q_{Y|\tilde{X}, \tilde{Y}}$ and $Q_{X|\tilde{X},X'}$ be attacks satisfying \eqref{eq:spoof1}. For any given $(N_{\one}, N_{\two},n) $ code $(f_{\one},f_{\two}, \phi)$, let $i, j\in \cM_{\one}$ be distinct. For $i \in \cM_{\one}$, let $\vecx_{i} = f_{\one}(i)$. Similarly, for $k\in \cM_{\two}$, let $\vecy_{k}$ denote $f_{\two}(k)$. Consider the following situations.
\begin{itemize}
	\item User \one sends $\vecx_{i}$. User \two is adversarial and its input to the channel is produced by passing $(\vecx_{j}, \vecy_k)$ through the n-fold product channel $Q_{Y|\tilde{X}, \tilde{Y}}$. For $\vecz\in \cZ^n$, the output distribution $\bbP(\vecz)$ (denoted by $P_{i,j,k}(\vecz)$) is given by 
	\begin{align}\label{eq:dist1}
	P_{i,j,k}(\vecz) = \prod_{t=1}^n\sum_{y\in\cY}Q_{Y|\tilde{X}, \tilde{Y}}(y|\vecx_{j}(t),\vecy_{k}(t))W_{Z|X,Y}(\vecz(t)|\vecx_{i}(t),y).
	\end{align}

	\item User \two sends  $\vecy_k$. User \one is adversarial and its input $\vecX_{i,j}$ to the channel is produced by passing $(\vecx_{i}, \vecx_{j})$ through the n-fold product channel $Q_{X|\tilde{X}, X'}$.  For $\vecz\in \cZ^n$, the output distribution $\bbP(\vecz)$ (denoted by $Q_{i,j,k}(\vecz)$) is given by 
\begin{align}\label{eq:dist2}
	Q_{i,j,k}(\vecz) = \prod_{t=1}^n\sum_{x\in\cX}Q_{X|\tilde{X}, X'}(x|\vecx_{i}(t),\vecx_{j}(t))W_{Z|X,Y}(\vecz(t)|x,\vecy_{k}(t)).
	\end{align}
\end{itemize}
By \eqref{eq:spoof1}, we see that for all $i,j\in \cM_{\one}, k\in \cM_{\two}$ and $\vecz\in \cZ^n$, $P_{i,j,k}(\vecz)=P_{j,i,k}(\vecz)=Q_{i,j,k}(\vecz)$. From \eqref{eq:mal2} and \eqref{eq:dist1}, we see that
\begin{align*}
P_{e,\maltwo}&\geq \frac{1}{N^2_{\one}\times N_{\two}}\sum_{i,j\in\cM_{\one}}\sum_{k\in \cM_{\two}}\sum_{\vecz:\phi_{\one}(\vecz)\notin\{i,\two\}}P_{i,j,k}(\vecz)
\end{align*} and 
\begin{align*}
P_{e,\maltwo}&\geq \frac{1}{N^2_{\one}\times N_{\two}}\sum_{i,j\in\cM_{\one}}\sum_{k\in \cM_{\two}}\sum_{\vecz:\phi_{\one}(\vecz)\notin\{j,\two\}}P_{j,i,k}(\vecz).
\end{align*} Using \eqref{eq:mal1} and \eqref{eq:dist2}, we obtain
\begin{align*}
P_{e,\malone}&\geq \frac{1}{N^2_{\one}\times N_{\two}}\sum_{i,j\in\cM_{\one}}\sum_{k\in \cM_{\two}}\sum_{\vecz:\phi_{\two}(\vecz)\notin\{k,\one\}}Q_{i,j,k}(\vecz).
\end{align*}
\noindent Thus,
\begin{align*}
3&P_{e}(f_{\one},f_{\two},\phi)\geq P_{e,\maltwo}+P_{e,\maltwo}+P_{e,\malone}\\
&\geq \frac{1}{N^2_{\one}\times N_{\two}}\sum_{i,j\in\cM_{\one}}\sum_{k\in \cM_{\two}}\inp{\sum_{\vecz:\phi_{\one}(\vecz)\notin\{i,\two\}}P_{i,j,k}(\vecz) + \sum_{\vecz:\phi_{\one}(\vecz)\notin\{j,\two\}}P_{j,i,k}(\vecz) +\sum_{\vecz:\phi_{\two}(\vecz)\notin\{k,\one\}}Q_{i,j,k}(\vecz) }\\
&\stackrel{\text{(a)}}{=}\frac{1}{N^2_{\one}\times N_{\two}}\sum_{i,j\in\cM_{\one}}\sum_{k\in \cM_{\two}}\inp{\sum_{\vecz:\phi_{\one}(\vecz)\notin\{i,\two\}}P_{i,j,k}(\vecz) + \sum_{\vecz:\phi_{\one}(\vecz)\notin\{j,\two\}}P_{i,j,k}(\vecz) +\sum_{\vecz:\phi_{\two}(\vecz)\notin\{k,\one\}}P_{i,j,k}(\vecz) }\\
&\geq\frac{1}{N^2_{\one}\times N_{\two}}\sum_{i,j\in\cM_{\one}, i\neq j}\sum_{k\in \cM_{\two}}\inp{\sum_{\vecz\in\cZ^n}P_{i,j,k}(\vecz)}\\
&=\frac{N_{\one}(N_{\one}-1)N_{\two}}{2N^2_{\one}\times N_{\two}}\\
& = \frac{N_{\one}-1}{2N_{\one}}\\
&\geq \frac{1}{4}
\end{align*}
where (a) follows by noting that $P_{i,j,k}(\vecz)=P_{j,i,k}(\vecz)=Q_{i,j,k}(\vecz)$. Thus, for any given code $(f_{\one}, f_{\two}, \phi)$, for a spoofable channel $P_e(f_{\one},f_{\two}, \phi)\geq \frac{1}{12}$.
A similar analysis follows when the channel is \two-\spoofable.
\end{proof}
Next, we show our positive result.

\iffalse
Let $R = (R_{\one},R_{\two})$ be the set of all rate pairs satisfying

For random variables $X',Y',X, Y $ and $Z$, and distributions $P_{X}$ and $P_{Y}$, let $S(P_{X},P_{Y})$ be the set of joint distributions defined as

$S(P_{X},P_{Y}) \defineqq \{P_{XYZX'Y'}: P_{XYZ} = P_{X}Q_{Y}W$ for some $Q_{Y}$,  $ P_{X'Y'Z} = Q_{X'}P_{Y}W$ for some $Q_{X'}\}$

Let $\cR$ be the cloure of the set of rate pairs $(R_{\one},\,R_{\two})$ such that for some $P_{X}$ and $P_{Y}$ the following hold
\begin{align}
R_{\one} &\leq \min_{P_{XYZX'Y'}\in S(P_{ X},P_{Y})}I(X;Z),\text{ and}\label{eq:R1}\\
R_{\two} &\leq  \min_{P_{XYZX'Y'}\in S(P_{X},P_{Y})}I(Y;Z).\label{eq:R2}
\end{align}
\begin{thm}
For a channel \mch, $\mathsf{int}(\cR)\neq \emptyset$ if and only if the channel is non spoofable. 
\end{thm}

\begin{thm}
If a channel is non-\spoofable, $\cR\subseteq \cR_{\mathsf{det}}$.
\end{thm}
Fix distributions $P_X$ and $P_Y$. For any $\delta_{\one}, \delta_{\two}>0$, and $R_{\one}=\min_{P_{XYZX'Y'}\in S(P_{ X},P_{Y})}I(X;Z)-\delta_{\one}$ and  $R_{\two}= \min_{P_{XYZX'Y'}\in S(P_{X},P_{Y})}I(Y;Z) - \delta_{\two}$, we will show the existence of a code which achieves these rates. Let $N_{\one}:= \exp\inp{nR_{\one}}$ and $N_{\two}:= \exp\inp{nR_{\two}}$ (for simplicity, assume that $N_{\one}$ and $N_{\two}$ are integers), and $\cM_{\one} = \{1,2,\ldots,N_{\one}\}$ and $\cM_{\two} = \{1,2,\ldots,N_{\two}\}$.\\
\fi

\begin{lemma}\label{thm:detCodesPositivity}
The rate region for deterministic codes is non-empty if the channel is {\em non-\spoofable}.
\end{lemma}
\begin{proof}

\noindent {\em Encoding.} 
For some $P_{\one}$ and $P_{\two}$ satisfying $\min_{x\in \cX}P_{\one}(x)>0$ and $\min_{y\in \cY}P_{\two}(y)>0$ respectively, and $\epsilon>0$ (TBD), the codebook is given by Lemma~\ref{lemma:codebook}.  For $\mo \in \cM_{\one}$, $f_{\one}(\mo) = \vecx_{\mo}$ and for $\mt \in \cM_{\two}$, $f_{\two}(\mt) = \vecy_{\mt}$.\\
{\em Decoding.} 
For a parameter $\eta>0$, let $\cD_{\eta}$ be the set of joint distributions defined as
$D_\eta \defineqq  \inb{P_{XYZ}\in \cP^n_{\cX\times \cY \times \cZ}:\, D\inp{P_{XYZ}||P_XP_YW}\leq\eta }$.
For the given codebook, the parameter $\eta$ and the received output sequence $\vecz$, let $D_{\one}(\eta, \vecz)$ be defined as the set of messages $m_{\one}\in\mathcal{M}_{\one}$ such that there exists $\vecy\in \cY^n$ satisfying the following conditions:
\begin{enumerate}[label=(\roman*)]
\item $\inp{f_{\one}(m_{\one}), \vecy, \vecz} \in T^{n}_{XYZ}$ for some $P_{XYZ}\in \cD_{\eta}.$ \label{check:1}
\item {For every} $(\tilde{m}_{\one},\, \tilde{m_{\two}})\in \cM_{\one}\times \cM_{\two}, \, \tilde{m}_{\one}\neq m_{\one}$ and  $(\vecy', \, \vecx')\in \cY^n\times \cX^n$ such that $\inp{f_{\one}(\mo), \vecy, f_{\one}(\tilde{m}_{\one}), \vecy',\vecx', f_{\two}(\tilde{m}_{\two}), \vecz}\in T^n_{XY\tilde{X}Y'X'\tilde{Y}Z}$, $P_{\tilde{X}Y'Z}\in \cD_{\eta}$ and  $P_{X'\tilde{Y}Z}\in \cD_{\eta}$, {we require that} $I(\tilde{X}\tilde{Y};XZ|Y)<\eta$.\label{check:2}
\item {For every} $\tilde{m}_{\two 1},\, \tilde{m}_{\two 2}\in \cM_{\two}$, and $\vecx'_1, \, \vecx'_2\in \cX^n$ such that $\inp{f_{\one}(\mo), \vecy,\vecx'_1, f_{\two}(\tilde{m}_{\two 1}),\vecx'_2, f_{\two}(\tilde{m}_{\two 2}),\vecz}\in T^n_{XYX'_1\tilde{Y}_1X'_2\tilde{Y}_2Z}$, $P_{X'_1\tilde{Y}_1Z}\in \cD_{\eta}$ and  $P_{X'_2\tilde{Y}_2Z}\in \cD_{\eta}$, {we require that} $I(\tilde{Y}_1\tilde{Y}_2;XZ|Y)< \eta$.\label{check:3}
\end{enumerate}
We define $D_{\two}(\eta, \vecz)$ analogously (by interchanging the roles of user \one and \two).
\begin{align*}
\phi(\vecz) \defineqq\begin{cases}(\mo,\mt), &\text{ if }D_{\one}(\eta, \vecz)\times D_{\two}(\eta, \vecz) = \{(\mo,\mt)\}\\ \oneb, &\text{ if }|D_{\one}(\eta, \vecz)| = 0, \, |D_{\two}(\eta, \vecz)| \neq 0\\\twob, &\text{ if }|D_{\two}(\eta, \vecz)| = 0, \, |D_{\one}(\eta, \vecz)| \neq 0\\(1,1)&\text{ otherwise}
																\end{cases}
\end{align*}

For small enough choice of $\eta>0$, Lemma~\ref{lemma:disambiguity} implies that if $|D_{\one}(\eta, \vecz)|, \, |D_{\two}(\eta, \vecz)| \geq 1$, then $|D_{\two}(\eta, \vecz)|=|D_{\one}(\eta, \vecz)| = 1$. {To see this, suppose $|D_{\one}(\eta, \vecz)|\geq 2$ and $|D_{\two}(\eta, \vecz)| \geq 1$. Let $m_{\one}, \tilde{m}_{\one}\in D_{\one}(\eta, \vecz)$ and $\tilde{m}_{\two}\in D_{\two}(\eta, \vecz)$. This implies that there exist $\vecx$, $\vecy$ and $\vecy'$ such that for $\inp{f_{\one}(\mo), \vecy, f_{\one}(\tilde{m}_{\one}), \vecy',\vecx', f_{\two}(\tilde{m}_{\two}), \vecz}\in T^n_{XY\tilde{X}Y'X'\tilde{Y}Z}$, $P_{XYZ}\in \cD_{\eta}$, $P_{\tilde{X}Y'Z}\in \cD_{\eta}$, $P_{X'\tilde{Y}Z}\in \cD_{\eta}$, $I(\tilde{X}\tilde{Y};XZ|Y)<\eta$, $I(X\tilde{Y};\tilde{X}Z|Y')<\eta$ and $I(X\tilde{X};\tilde{Y}Z|X')<\eta$. This is not possible because of Lemma \ref{lemma:disambiguity}.}

\begin{lemma}\label{lemma:disambiguity}
{For a channel which is not \one-spoofable,} there does not exist a distribution $P_{XY\tilde{X}Y'X'\tilde{Y}Z}\in \cP^n_{XY\tilde{X}Y'X'\tilde{Y}Z}$ with $\min_{x}P_X(x), \min_{\tilde{x}}P_{\tilde{X}}(\tilde{x}), \min_{\tilde{y}}P_{\tilde{Y}}(\tilde{y})\geq \alpha>0$ which{, for a small enough $\eta>0$,} satisfies the following:
\begin{enumerate}[label=(\Alph*)]
	\item $P_{XYZ}\in D_{\eta}$,\label{disamb:1}
	\item $P_{\tilde{X}Y'Z}\in D_{\eta}$,\label{disamb:2}
	\item $P_{X'\tilde{Y}Z}\in D_{\eta}$,\label{disamb:3}
	\item $I(\tilde{X}\tilde{Y};XZ|Y)<\eta$,\label{disamb:4}
	\item $I(X\tilde{Y};\tilde{X}Z|Y')<\eta$ and\label{disamb:5}
	\item $I(X\tilde{X};\tilde{Y}Z|X')<\eta$.\label{disamb:6}
\end{enumerate}
{Similarly, for a channel which is not \two-\spoofable, there does not exist a distribution $P_{X'_1\tilde{Y}_1X'_2\tilde{Y}_2XYZ}\in \cP^n_{X'_1\tilde{Y}_1X'_2\tilde{Y}_2XYZ}$ with $\min_{x}P_X(x), \min_{\tilde{y}_1}P_{\tilde{Y}_1}(\tilde{y}_1), \min_{\tilde{y}_2}P_{\tilde{Y}_2}(\tilde{y}_2)\geq \alpha>0$ which, for a small enough $\eta>0$, satisfies the following:
\begin{enumerate}[label=(\Alph*)]
	\item $P_{XYZ}\in D_{\eta}$,
	\item $P_{X'_1\tilde{Y}_1Z}\in D_{\eta}$,
	\item $P_{X'_2\tilde{Y}_2Z}\in D_{\eta}$,
	\item $I(\tilde{Y}_1\tilde{Y}_2;XZ|Y)<\eta$,
	\item $I(X\tilde{Y}_2;\tilde{Y}_1Z|X'_1)<\eta$ and
	\item $I(X\tilde{Y}_1;\tilde{Y}_2Z|X'_2)<\eta$.
\end{enumerate}}
\end{lemma}

\begin{proof}
Suppose for a channel which is not \one-\spoofable, there exists $P_{XY\tilde{X}Y'X'\tilde{Y}Z}\in \cP^n_{XY\tilde{X}Y'X'\tilde{Y}Z}$ which satisfies \ref{disamb:1}-\ref{disamb:6}.
Using~\ref{disamb:1} and \ref{disamb:4}, we obtain that
{
\begin{align*}
2\eta&\geq D(P_{XYZ}||P_XP_YW)+I(\tilde{X}\tilde{Y};XZ|Y)\\
&=D(P_{XYZ}||P_XP_YW_{Z|X,Y})+D(P_{XY\tilde{X}\tilde{Y}Z}||P_YP_{\tilde{X}\tilde{Y}|Y}P_{XZ|Y})\\
&=\sum_{x,y, \tilde{x}, \tilde{y}, z}P_{XY\tilde{X}\tilde{Y}Z}(x,y, \tilde{x}, \tilde{y}, z)\left(\log{\left\{\frac{P_{XYZ}(x, y, z)}{P_X(x)P_Y(y)W_{Z|X,Y}(z|x,y)}\right\}} + \log{\left\{\frac{P_{XY\tilde{X}\tilde{Y}Z}(x, y, \tilde{x}, \tilde{y}, z)}{P_Y(y)P_{\tilde{X}\tilde{Y}|Y}(\tilde{x}, \tilde{y}|y)P_{XZ|Y}(x,z|y)}\right\}}\right)\\
&=\sum_{x,y, \tilde{x}, \tilde{y}, z}P_{XY\tilde{X}\tilde{Y}Z}(x,y, \tilde{x}, \tilde{y}, z)\left(\log{\left\{\frac{P_{XYZ}(x, y, z)\times P_{XY\tilde{X}\tilde{Y}Z}(x, y, \tilde{x}, \tilde{y}, z)}{P_X(x)P_Y(y)W_{Z|X,Y}(z|x,y)\times P_Y(y)P_{\tilde{X}\tilde{Y}|Y}(\tilde{x}, \tilde{y}|y)P_{XZ|Y}(x,z|y)}\right\}}\right)\\
&=\sum_{x,y, \tilde{x}, \tilde{y}, z}P_{XY\tilde{X}\tilde{Y}Z}(x,y, \tilde{x}, \tilde{y}, z)\left(\log{\left\{\frac{ P_{XY\tilde{X}\tilde{Y}Z}(x, y, \tilde{x}, \tilde{y}, z)}{P_X(x)P_Y(y)W_{Z|X,Y}(z|x,y) P_{\tilde{X}\tilde{Y}|Y}(\tilde{x}, \tilde{y}|y)}\right\}}\right)\\
&=\sum_{x,y, \tilde{x}, \tilde{y}, z}P_{XY\tilde{X}\tilde{Y}Z}(x,y, \tilde{x}, \tilde{y}, z)\left(\log{\left\{\frac{ P_{XY\tilde{X}\tilde{Y}Z}(x, y, \tilde{x}, \tilde{y}, z)}{P_X(x) W_{Z|X,Y}(z|x,y) P_{Y\tilde{X}\tilde{Y}}(y, \tilde{x}, \tilde{y})}\right\}}\right)\\
& = D(P_{XY\tilde{X}\tilde{Y}Z}||P_XP_{\tilde{X}\tilde{Y}}P_{Y|\tilde{X}\tilde{Y}}W)\\
&\stackrel{(a)}{\geq} D(P_{X\tilde{X}\tilde{Y}Z}||P_XP_{\tilde{X}\tilde{Y}}V^1_{Z|X\tilde{X}\tilde{Y}})
\end{align*}
}
where $V^{(1)}_{Z|X\tilde{X}\tilde{Y}}(z|x,\tilde{x},\tilde{y}) \defineqq \sum_{y}P_{Y|\tilde{X}\tilde{Y}}(y|\tilde{x},\tilde{y}) W(z|x,y)$ and $(a)$ follows from the log sum inequality. Using Pinsker's inequality,
\begin{align}
d_{TV}\inp{P_{X\tilde{X}\tilde{Y}Z},P_XP_{\tilde{X}\tilde{Y}}V^{(1)}_{Z|X\tilde{X}\tilde{Y}}}< \sqrt{\eta}.\label{disambeq:1}
\end{align}
Similarly, using~\ref{disamb:2} and \ref{disamb:5}, we obtain 

\begin{align}
d_{TV}\inp{P_{\tilde{X}X\tilde{Y}Z},P_{\tilde{X}}P_{{X}\tilde{Y}}V^{(2)}_{Z|X\tilde{X}\tilde{Y}}}< \sqrt{\eta}\label{disambeq:2}
\end{align}
where $V^{(2)}_{Z|X\tilde{X}\tilde{Y}}(z|x,\tilde{x},\tilde{y}) \defineqq \sum_{y'}P_{Y'|{X}\tilde{Y}}(y'|{x},\tilde{y}) W(z|\tilde{x},y')$. Finally, using~\ref{disamb:3} and \ref{disamb:6}, we get
\begin{align}
d_{TV}\inp{P_{X\tilde{X}\tilde{Y}Z},P_{X\tilde{X}}P_{\tilde{Y}}V^{(3)}_{Z|X\tilde{X}\tilde{Y}}}< \sqrt{\eta}\label{disambeq:3}
\end{align}
where $V^{(3)}_{Z|X\tilde{X}\tilde{Y}}(z|x,\tilde{x},\tilde{y})\defineqq \sum_{x'}P_{X'|X\tilde{X}}(x'|x,\tilde{x})W(z|x',\tilde{y})$.

\begin{align*}
&2d_{TV}\inp{P_XP_{\tilde{X}\tilde{Y}}V^{(1)}_{Z|X\tilde{X}\tilde{Y}},P_XP_{\tilde{X}}P_{\tilde{Y}}V^{(1)}_{Z|X\tilde{X}\tilde{Y}}}\\
&=\sum_{x,\tilde{x},\tilde{y},z}\left|P_X(x)P_{\tilde{X}\tilde{Y}}(\tilde{x},\tilde{y})V^{(1)}_{Z|X\tilde{X}\tilde{Y}}(z|x,\tilde{x},\tilde{y})-P_X(x)P_{\tilde{X}}(\tilde{x})P_{\tilde{Y}}(\tilde{y})V^{(1)}_{Z|X\tilde{X}\tilde{Y}}(z|x,\tilde{x},\tilde{y})\right|\\
& = \sum_{x,\tilde{x},\tilde{y}}\inp{\sum_{z}V^{(1)}_{Z|X\tilde{X}\tilde{Y}}(z|x,\tilde{x},\tilde{y})}\left|P_X(x)P_{\tilde{X}\tilde{Y}}(\tilde{x},\tilde{y})-P_X(x)P_{\tilde{X}}(\tilde{x})P_{\tilde{Y}}(\tilde{y})\right|\\
& = \inp{\sum_xP_X(x)}\sum_{\tilde{x},\tilde{y}}\left|P_{\tilde{X}\tilde{Y}}(\tilde{x},\tilde{y})-P_{\tilde{X}}(\tilde{x})P_{\tilde{Y}}(\tilde{y})\right|\\
& = 2d_{TV}\inp{P_{\tilde{X}\tilde{Y}},P_{\tilde{X}}P_{\tilde{Y}}} <\sqrt{\eta} \text{ by using ~\eqref{disambeq:2}}.
\end{align*}
\iffalse
{
Using ~\eqref{disambeq:2}, we see that 
\begin{align*}
&d_{TV}\inp{P_{\tilde{X}\tilde{Y}},P_{\tilde{X}}P_{\tilde{Y}}} <\sqrt{\eta}
\end{align*}
Note that
\begin{align*}
2\sqrt{\eta} &> 2d_{TV}\inp{P_{\tilde{X}\tilde{Y}},P_{\tilde{X}}P_{\tilde{Y}}}\\
&= \sum_{\tilde{x},\tilde{y}}\left|P_{\tilde{X}\tilde{Y}}(\tilde{x},\tilde{y})-P_{\tilde{X}}(\tilde{x})P_{\tilde{Y}}(\tilde{y})\right|\\
&=\sum_{\tilde{x},\tilde{y}}\inp{\sum_{x,z}P_{XZ|\tilde{X}\tilde{Y}}(x, z|\tilde{x}, \tilde{y})}\left|P_{\tilde{X}\tilde{Y}}(\tilde{x},\tilde{y})-P_X(x)P_{\tilde{X}}(\tilde{x})P_{\tilde{Y}}(\tilde{y})\right|\\
&=\sum_{\tilde{x},\tilde{y},x,z}\left|P_{\tilde{X}\tilde{Y}}(\tilde{x},\tilde{y})P_{XZ|\tilde{X}\tilde{Y}}(x, z|\tilde{x}, \tilde{y})-P_{\tilde{X}}(\tilde{x})P_{\tilde{Y}}(\tilde{y})P_{XZ|\tilde{X}\tilde{Y}}(x, z|\tilde{x}, \tilde{y})\right|\\
&= 2d_{TV}\inp{P_{\tilde{X}\tilde{Y}XZ},P_{\tilde{X}}P_{\tilde{Y}}P_{XZ|\tilde{X}\tilde{Y}}}
\end{align*}
\fi
{
We use this and \eqref{disambeq:1}, to show
\begin{align}
3\sqrt{\eta}/2 &\geq d_{TV}\inp{P_{X\tilde{X}\tilde{Y}Z},P_XP_{\tilde{X}\tilde{Y}}V^{(1)}_{Z|X\tilde{X}\tilde{Y}}} + d_{TV}\inp{P_XP_{\tilde{X}\tilde{Y}}V^{(1)}_{Z|X\tilde{X}\tilde{Y}},P_XP_{\tilde{X}}P_{\tilde{Y}}V^{(1)}_{Z|X\tilde{X}\tilde{Y}}}\\
&\stackrel{\text{(a)}}{\geq} d_{TV}\inp{P_{X\tilde{X}\tilde{Y}Z},P_XP_{\tilde{X}}P_{\tilde{Y}}V^{(1)}_{Z|X\tilde{X}\tilde{Y}}}
\end{align} 
where (a) uses the triangle inequality.}
Thus, 
\begin{align}
d_{TV}\inp{P_{X\tilde{X}\tilde{Y}Z},P_XP_{\tilde{X}}P_{\tilde{Y}}V^{(1)}_{Z|X\tilde{X}\tilde{Y}}} \leq 3\sqrt{\eta}/2.\label{eq:sym1}
\end{align}
Similarly, using \eqref{disambeq:1} to show that $
d_{TV}\inp{P_{\tilde{X}}P_{X}P_{\tilde{Y}}V^{(2)}_{Z|X\tilde{X}\tilde{Y}},P_{\tilde{X}}P_{{X}\tilde{Y}}V^{(2)}_{Z|X\tilde{X}\tilde{Y}}}<\sqrt{\eta}/2$ and \eqref{disambeq:2}, we obtain
\begin{align}
d_{TV}\inp{P_{X\tilde{X}\tilde{Y}Z},P_XP_{\tilde{X}}P_{\tilde{Y}}V^{(2)}_{Z|X\tilde{X}\tilde{Y}}}\leq 3\sqrt{\eta}/2 \label{eq:sym2}
\end{align} and using \eqref{disambeq:1} to show that $d_{TV}\inp{P_{X}P_{\tilde{X}}P_{\tilde{Y}}V^{(3)}_{Z|X\tilde{X}\tilde{Y}},P_{X\tilde{X}}P_{\tilde{Y}}V^{(3)}_{Z|X\tilde{X}\tilde{Y}}}$ and \eqref{disambeq:3}, we obtain
\begin{align}
d_{TV}\inp{P_{X\tilde{X}\tilde{Y}Z},P_XP_{\tilde{X}}P_{\tilde{Y}}V^{(3)}_{Z|X\tilde{X}\tilde{Y}}}\leq 3\sqrt{\eta}/2.\label{eq:sym3}
\end{align}
{Suppose the channel is not \one-\spoofable (i.e. \eqref{eq:spoof1} does not hold), then there exists $\zeta>0$ such that for every $Q_{Y|\tilde{X} \tilde{Y}}$ and $Q_{X|\tilde{X}X'}$, at least one of the following two conditions hold:
\begin{align}
&\max_{x',\tilde{x}, \tilde{y}, z}\left|\sum_y Q_{Y|\tilde{X}\tilde{Y}}(y|\tilde{x}, \tilde{y})W_{Z|XY}(z|x',y)-\sum_y Q_{Y|\tilde{X}\tilde{Y}}(y|x', \tilde{y})W_{Z|XY}(z|\tilde{x}	,y)\right|>\zeta \label{eq:condition1}\\
&\max_{x',\tilde{x}, \tilde{y}, z}\left|\sum_y Q_{Y|\tilde{X}\tilde{Y}}(y|\tilde{x}, \tilde{y})W_{Z|XY}(z|x',y)-\sum_y Q_{X|\tilde{X}X'}(x|\tilde{x}, x')W_{Z|XY}(z|x,\tilde{y})\right|>\zeta \label{eq:condition2}\\
\end{align}
Suppose \eqref{eq:condition2} holds. We use \eqref{eq:sym1} and \eqref{eq:sym3} to write the following:
\begin{align*}
d_{TV}&\inp{P_XP_{\tilde{X}}P_{\tilde{Y}}V^{(1)}_{Z|X\tilde{X}\tilde{Y}}, P_XP_{\tilde{X}}P_{\tilde{Y}}V^{(3)}_{Z|X\tilde{X}\tilde{Y}}} \\
&\leq d_{TV}\inp{P_{X\tilde{X}\tilde{Y}Z},P_XP_{\tilde{X}}P_{\tilde{Y}}V^{(1)}_{Z|X\tilde{X}\tilde{Y}}} +d_{TV}\inp{P_{X\tilde{X}\tilde{Y}Z},P_XP_{\tilde{X}}P_{\tilde{Y}}V^{(3)}_{Z|X\tilde{X}\tilde{Y}}} \\
&\leq 3\sqrt{\eta}.
\end{align*}
Thus, 
\begin{align*}
&\max_{x,\tilde{x}, \tilde{y}, z}\alpha^3\left|\sum_{y}P_{Y|\tilde{X}\tilde{Y}}(y|\tilde{x},\tilde{y}) W(z|x,y) -  \sum_{x'}P_{X'|X\tilde{X}}(x'|x,\tilde{x})W(z|x',\tilde{y})\right|\\
&\leq\max_{x,\tilde{x}, \tilde{y}, z}P_X(x)P_{\tilde{X}}(\tilde{x})P_{\tilde{Y}}(\tilde{y})\left|\sum_{y}P_{Y|\tilde{X}\tilde{Y}}(y|\tilde{x},\tilde{y}) W(z|x,y) -  \sum_{x'}P_{X'|X\tilde{X}}(x'|x,\tilde{x})W(z|x',\tilde{y})\right|\\
&=\max_{x,\tilde{x}, \tilde{y}, z}\left|P_X(x)P_{\tilde{X}}(\tilde{x})P_{\tilde{Y}}(\tilde{y})\sum_{y}P_{Y|\tilde{X}\tilde{Y}}(y|\tilde{x},\tilde{y}) W(z|x,y) -  P_X(x)P_{\tilde{X}}(\tilde{x})P_{\tilde{Y}}(\tilde{y})\sum_{x'}P_{X'|X\tilde{X}}(x'|x,\tilde{x})W(z|x',\tilde{y})\right|\\
&\leq \sum_{{x,\tilde{x}, \tilde{y}, z}}\left|P_X(x)P_{\tilde{X}}(\tilde{x})P_{\tilde{Y}}(\tilde{y})\sum_{y}P_{Y|\tilde{X}\tilde{Y}}(y|\tilde{x},\tilde{y}) W(z|x,y) -  P_X(x)P_{\tilde{X}}(\tilde{x})P_{\tilde{Y}}(\tilde{y})\sum_{x'}P_{X'|X\tilde{X}}(x'|x,\tilde{x})W(z|x',\tilde{y})\right|\\
&= 2d_{TV}\inp{P_XP_{\tilde{X}}P_{\tilde{Y}}V^{(1)}_{Z|X\tilde{X}\tilde{Y}}, P_XP_{\tilde{X}}P_{\tilde{Y}}V^{(3)}_{Z|X\tilde{X}\tilde{Y}}}\\
&\leq 3\sqrt{\eta}.
\end{align*}
This contradicts \eqref{eq:condition2} for $\zeta>3\sqrt{\eta}/\alpha^3$.
Next, we consider the case when \eqref{eq:condition1} holds. In this case, for any $P_{Y|\tilde{X}\tilde{Y}}$ and $P_{Y'|{X}\tilde{Y}}$,
\begin{align*}
&2\max_{x,\tilde{x}, \tilde{y}, z}\left|\sum_{y}P_{Y|\tilde{X}\tilde{Y}}(y|\tilde{x},\tilde{y}) W(z|x,y)-\sum_{y'}P_{Y'|{X}\tilde{Y}}(y'|{x},\tilde{y}) W(z|\tilde{x},y')\right|\\
&= \max_{x,\tilde{x}, \tilde{y}, z}\left|\sum_{y}P_{Y|\tilde{X}\tilde{Y}}(y|\tilde{x},\tilde{y}) W(z|x,y)-\sum_{y'}P_{Y'|{X}\tilde{Y}}(y'|{x},\tilde{y}) W(z|\tilde{x},y')\right| \\
&\qquad + \max_{x,\tilde{x}, \tilde{y}, z}\left|\sum_{y}P_{Y'|{X}\tilde{Y}}(y|\tilde{x},\tilde{y}) W(z|x,y) - \sum_{y'}P_{Y|\tilde{X}\tilde{Y}}(y'|{x},\tilde{y}) W(z|\tilde{x},y')\right|\\
&\geq 2\max_{x,\tilde{x}, \tilde{y}, z}\left|\sum_{y}\inp{\frac{P_{Y|\tilde{X}\tilde{Y}}(y|\tilde{x},\tilde{y})+P_{Y'|{X}\tilde{Y}}(y|\tilde{x},\tilde{y})}{2}}W(z|x,y)-\sum_{y'}\inp{\frac{P_{Y'|{X}\tilde{Y}}(y'|{x},\tilde{y})+P_{Y|\tilde{X}\tilde{Y}}(y'|{x},\tilde{y}))}{2}} W(z|\tilde{x},y')\right| \\
&\stackrel{\text{(a)}}{=} 2\max_{x,\tilde{x}, \tilde{y}, z}\left|\sum_{y}Q_{Y|\tilde{X}\tilde{Y}}(y|\tilde{x},\tilde{y})W(z|x,y)-\sum_{y'}Q_{Y|\tilde{X}\tilde{Y}}(y'|{x},\tilde{y}) W(z|\tilde{x},y')\right|\\
\stackrel{\text{(b)}}{\geq} 2\zeta
\end{align*}
where (a) follows by defining $Q_{Y|\tilde{X}\tilde{Y}}\defineqq \frac{P_{Y'|{X}\tilde{Y}}(y'|{x},\tilde{y})+P_{Y|\tilde{X}\tilde{Y}}(y'|{x},\tilde{y}))}{2}$ and (b) follows from \eqref{eq:condition1}.
Thus, 
\begin{align}\max_{x,\tilde{x}, \tilde{y}, z}\left|\sum_{y}P_{Y|\tilde{X}\tilde{Y}}(y|\tilde{x},\tilde{y}) W(z|x,y)-\sum_{y'}P_{Y'|{X}\tilde{Y}}(y'|{x},\tilde{y}) W(z|\tilde{x},y')\right|\geq \zeta. \label{eq:condition3}
\end{align}
Using \eqref{eq:sym1} and \eqref{eq:sym2}, we can show that \eqref{eq:condition3} (and thus, \eqref{eq:condition1}) does not hold for $\zeta>3\sqrt{\eta}/\alpha^3$.\\
This completes the proof of the first statement. The proof of the second statement is along the same lines as the proof of the first statement. It can be obtained by interchanging the roles of users \one and \two and making the following replacements in the above proof: $X\rightarrow \tilde{Y}_1, \, Y\rightarrow X_1', \, \tilde{X}\rightarrow \tilde{Y}_2, \, {Y}'\rightarrow X'_2, \, X'\rightarrow Y, \text{ and }\tilde{Y}\rightarrow X$.}
\end{proof}

Fix $R_{\one} = R_{\two} = \delta$ for some positive $\delta$ (TBD).
We start by showing that $P_{e,\na}$ can be upper bounded by sum of $P_{e,\malone}$ and $P_{e,\maltwo}$. So, we only need to analyse the case when a user is malicious. To show this, we note that $\cE_{\mo, \mt}  = \cE_{\mo}\cup\cE_{\mt}$. Thus,
\begin{align*}
&P_{e,\na}\hspace{-0.25em} =
\frac{1}{N_{\one}\cdot N_{\two}} \sum_{(\mo, \mt)\in \mathcal{M}_{\one}\times\mathcal{M}_{\two}}W^n\inp{\cE_{\mo}\cup\cE_{\mt}|f_{\one}(\mo), f_{\two}(\mt)}\\
&\leq\frac{1}{N_{\one}\cdot N_{\two}} \sum_{(\mo, \mt)\in \mathcal{M}_{\one}\times\mathcal{M}_{\two}}\Big(W^n\inp{\cE_{\mo}|f_{\one}(\mo), f_{\two}(\mt)}+W^n\inp{\cE_{\mt}|f_{\one}(\mo), f_{\two}(\mt)}\Big)\\
&=\frac{1}{ N_{\two}}\sum_{\mt\in\cM_{\two}}\inp{\frac{1}{N_{\one}} \sum_{\mo\in \mathcal{M}_{\one}}W^n\inp{\cE_{\mo}|f_{\one}(\mo), f_{\two}(\mt)}}\\
&\qquad \qquad+\frac{1}{N_{\one}}\sum_{\mo\in\cM_{\one}}\inp{ \frac{1}{N_{\two}}\sum_{\mt\in\mathcal{M}_{\two}}W^n\inp{\cE_{\mt}|f_{\one}(\mo), f_{\two}(\mt)}}\\
&\leq P_{e,\malone} +P_{e,\maltwo}.
\end{align*}
So, if $P_{e,\malone}$ and $P_{e,\maltwo}$ are small, $P_{e,\na}$ is also small. We will first analyse $P_{e,\maltwo}$. Suppose user \two attacks with an attack vector $\vecy\in \cY^n$. For some $\eta/3>\epsilon>0$, we define the following sets.
\begin{align*}
\cH_1&= \inb{m_{\one}: (\vecx_{m_{\one}}, \vecy)\in \cup_{P_{XY}\in \cP^n_{\cX\times\cY}}T_{XY}^n, I(X;Y)> \epsilon}\\
\cH_2&= \inb{m_{\one}: (\vecx_{m_{\one}}, \vecy)\in \cup_{P_{XY}\in \cP^n_{\cX\times\cY}}T_{XY}^n, I(X;Y)\leq \epsilon}
\end{align*}
For notational convenience, let $\phi(\vecz) = (\phi_{\one}(\vecz), \phi_{\two}(\vecz))$ and the output symbols $\oneb = (\one,\one)$ and $\twob = (\two,\two)$. Thus, the decoder always outputs a pair. 
\begin{align}
P_{e,\maltwo} \leq \frac{1}{N_{\one}}|\cH_1| + \frac{1}{N_{\one}}\sum_{m_{\one}\in \cH_2}\inp{\sum_{P_{XYZ}\in \cD_{\eta}^c}\sum_{\vecz\in T^n_{Z|XY}(\vecx_{m_{\one}},\vecy)}W^n(\vecz|\vecx_{m_{\one}}, \vecy)} \nonumber\\ + \frac{1}{N_{\one}}\sum_{m_{\one}\in \cH_2}\inp{\sum_{P_{XYZ}\in \cD_{\eta}}\sum_{\vecz\in T^n_{Z|XY}(\vecx_{m_{\one}},\vecy), \phi_{\one}(\vecz)\notin\inb{m_{\one}, \two}}W^n(\vecz|\vecx_{m_{\one}}, \vecy)}\label{eq:error}
\end{align}
The first term on the RHS is upper bounded by
\begin{align*}
|\cP^n_{\cX\times \cY}|\times\frac{|\inb{m_{\one}: (\vecx_{m_{\one}}, \vecy)\in T_{XY}^n, \, I(X;Y)> \epsilon}|}{N_{\one}}
\end{align*}
which goes to zero as $n\rightarrow \infty$ by~\eqref{codebook:1} and noting that there are only polynomially many types. Analysing the second term, for $m_{\one}\in \cH_2$ and $P_{XYZ}\in \cD_{\eta}^c$,  
\begin{align*}
\sum_{P_{XYZ}\in \cD_{\eta}^c}\sum_{\vecz\in T^n_{Z|XY}(\vecx_{m_{\one}},\vecy)}W^n(\vecz|\vecx_{m_{\one}}, \vecy) &\leq |\cD_{\eta}^c|\exp{\inp{-nD(P_{XYZ}||P_{XY}W)}}\\
& = |\cD_{\eta}^c|\exp{\inp{-n\inp{D(P_{XYZ}||P_{X}P_{Y}W) -I(X;Y)}}}\\
& \leq |\cD_{\eta}^c|\exp{\inp{-n\inp{\eta-\epsilon}}} \rightarrow 0 \text{ when }\epsilon<\eta.
\end{align*}
We are left to analyse the last term. For $(\vecx_{\mo}, \vecy, \vecz)\in P_{XYZ}$ such that $P_{XYZ}\in \cD_{\eta}$ and $\mo\in \cH_2$, $\phi_{\one}(\vecz)\notin\inb{m_{\one}, \two}$ when one of the following happens (follows from Lemma~\ref{lemma:disambiguity}).
\begin{itemize}
	\item $|D_{\one}(\eta, \vecz)| = |D_{\two}(\eta, \vecz)| = 1$, but $\mo\notin D_{\one}(\eta, \vecz)$.
	\item $|D_{\one}(\eta, \vecz)| = 0$.
\end{itemize}
To formalize this, we define the following sets. For $\mo\in \cM_{\one}$,
\begin{align*}
\cG_{\mo} &= \inb{\vecz: (\vecx_{\mo}, \vecy, \vecz)\in P_{XYZ}, P_{XYZ}\in \cD_{\eta}, I(X, Y)\leq \epsilon}\\
\cG_{\mo,0} &= \cG_{\mo} \cap \inb{\vecz: \phi_{\one}(\vecz)\notin\inb{m_{\one}, \two} }\\
\cG_{\mo,1} &= \cG_{\mo} \cap\inb{\vecz:  |D_{\one}(\eta, \vecz)| = |D_{\two}(\eta, \vecz)| = 1 , \mo\notin D_{\one}(\eta, \vecz) }\\
\cG_{\mo,2} &= \cG_{\mo} \cap\inb{\vecz:  |D_{\one}(\eta, \vecz)| = 0}\\
\cG_{\mo,3} &= \cG_{\mo} \cap\inb{\vecz: \mo\notin D_{\one}(\eta, \vecz) }
\end{align*}
We are interested in $\cG_{\mo,0}$. Note that $\cG_{\mo,0} \subseteq \cG_{\mo,1}\cup \cG_{\mo,2}\subseteq \cG_{\mo,3}$. So, it suffices to upper bound the probability of $\cG_{\mo,3}$ when $\vecx_{\mo}$ is sent by user \one and $\vecy$ by user \two. From the definition of $D_{\one}(\eta, \vecz)$, we see that $\cG_{\mo,3}$ is the set of $\vecz \in \cZ^n$ which satisfy decoding condition~\ref{check:1} (this is because $\vecz\in\cG_{\mo,3}$ implies $\vecz\in \cG_{\mo}$) but do not satisfy either decoding condition~\ref{check:2} or decoding condition~\ref{check:3}. We capture this by defining the following sets of distributions:
\begin{align*}
\cP_1& =  \{P_{X\tilde{X}Y\tilde{Y}Z}\in \cP^n_{\cX\times\cX\times\cY\times\cY\times\cZ}: P_{XYZ}\in \cD_{\eta},I(X;Y)\leq \epsilon,\, P_{\tilde{X}Y'Z}\in \cD_{\eta} \text{ for some }Y',\\
&\qquad  \, P_{X'\tilde{Y}Z}\in \cD_{\eta} \text{ for some }X', P_{X}=P_{\tilde{X}}=P_{\one}, P_{\tilde{Y}} = P_{\two}\text{ and }I(\tilde{X}\tilde{Y};XZ|Y)\geq\eta\}\\
\cP_2& = \{P_{X\tilde{Y}_1\tilde{Y}_2YZ}\in \cP^n_{\cX\times\cY\times\cY\times\cY\times\cZ}: P_{XYZ}\in \cD_{\eta},I(X;Y)\leq \epsilon,\, P_{X'_1\tilde{Y}_1Z}\in \cD_{\eta} \text{ for some }X'_1,\\
&\qquad \, P_{X'_2\tilde{Y}_2Z}\in \cD_{\eta} \text{ for some }X'_2, P_{X}=P_{\one}, P_{\tilde{Y}_1}=P_{\tilde{Y}_2} = P_{\two}\text{ and }I(\tilde{Y}_1\tilde{Y}_2;XZ|Y)\geq\eta\}.
\end{align*}
For $P_{X\tilde{X}Y\tilde{Y}Z}\in \cP_1$ and $P_{X\tilde{Y}_1\tilde{Y}_2YZ}\in \cP_2 $, let
\begin{align*}
\cE_{\mo,1}(P_{X\tilde{X}Y\tilde{Y}Z}) & = \big\{\vecz: \exists(\tilde{m}_{\one},\, \tilde{m}_{\two})\in \cM_{\one}\times \cM_{\two}, \, \tilde{m}_{\one}\neq m_{\one}, \,  \inp{\vecx_{\mo},\vecx_{\tilde{m}_{\one}}, \vecy,   \vecy_{\tilde{m}_{\two}}, \vecz}\in T^n_{X\tilde{X}Y\tilde{Y}Z} \big\} \text{ and }\\
\cE_{\mo,2}(P_{X\tilde{Y}_1\tilde{Y}_2YZ}) & = \big\{\vecz: \exists \tilde{m}_{\two 1},\, \tilde{m}_{\two 2}\in \cM_{\two}, \tilde{m}_{\two 1}\neq \tilde{m}_{\two 2},\,\inp{\vecx_{\mo},  \vecy_{\tilde{m}_{\two 1}}, \vecy_{\tilde{m}_{\two 2}},\vecy,\vecz}\in T^n_{X\tilde{Y}_1\tilde{Y}_2YZ}\big\}.
\end{align*}
Note that $\cG_{\mo,3}= \inp{\cup_{P_{X\tilde{X}Y\tilde{Y}Z}\in \cP_1}\cE_{\mo,1}(P_{X\tilde{X}Y\tilde{Y}Z})}\cup \inp{\cup_{P_{X\tilde{Y}_1\tilde{Y}_2YZ}\in \cP_2}\cE_{\mo,2}(P_{X\tilde{Y}_1\tilde{Y}_2YZ})}$.

Thus, the last term in \eqref{eq:error}  can be analysed as below.
\begin{align}
&\frac{1}{N_{\one}}\sum_{m_{\one}\in \cH_2}\inp{\sum_{P_{XYZ}\in \cD_{\eta}}\sum_{\vecz\in T^n_{Z|XY}(\vecx_{m_{\one}},\vecy), \phi_{\one}(\vecz)\notin\inb{m_{\one}, \two}}W^n(\vecz|\vecx_{m_{\one}}, \vecy)}\nonumber\\
\leq&\frac{1}{N_{\one}}\sum_{\mo\in \cH_2}\sum_{P_{X\tilde{X}Y\tilde{Y}Z}\in \cP_1}  W^n\inp{\cE_{\mo,1}(P_{X\tilde{X}Y\tilde{Y}Z})|\vecx_{\mo}, \vecy} \nonumber\\
&\qquad \qquad+ \frac{1}{N_{\one}}\sum_{\mo\in \cH_2}\sum_{P_{X\tilde{Y}_1\tilde{Y}_2YZ}\in \cP_2}  W^n\inp{\cE_{\mo,2}(P_{X\tilde{Y}_1\tilde{Y}_2YZ})|\vecx_{\mo}, \vecy}. \label{err:upperbound_p}
\end{align}
We see that $|\cP_1|$ and $|\cP_2|$ are at most polynomial and clearly $|\cH_2|\leq N_{\one}$. So, it will  suffice to uniformly upper bound $W^n\inp{\cE_{\mo,1}(P_{X\tilde{X}Y\tilde{Y}Z})|\vecx_{\mo}, \vecy}$ and $W^n\inp{P_{X\tilde{Y}_1\tilde{Y}_2YZ})|\vecx_{\mo}, \vecy}$ by a term exponentially decreasing in $n$ for all $P_{X\tilde{X}Y\tilde{Y}Z}\in \cP_1$ and $P_{X\tilde{Y}_1\tilde{Y}_2YZ}\in \cP_2 $. 
We start with the first term in the RHS of~\eqref{err:upperbound_p}. By using~\eqref{codebook:2b}, we see that for $P_{X\tilde{X}Y\tilde{Y}Z}\in \cP_1$ such that
\begin{align*}
I\inp{X;\tilde{X}\tilde{Y}Y}> |R_{\one}- I(\tilde{X};\tilde{Y}Y)|^{+}+|R_{\two}-I(\tilde{Y};Y)|^{+}+\epsilon\,
\end{align*}
\begin{align*}
\frac{\left|\inb{\mo:(\vecx_{\mo}, \vecx_{\tilde{m}_{\one}}, \vecy_{\mt}, \vecy)\in T^n_{X\tilde{X}\tilde{Y}Y} \text{ for some }\tilde{m}_{\one}\neq \mo\text{ and some }\mt} \right|}{N_{\one}} \leq \exp\inb{-n\epsilon/2}.
\end{align*}
So, 
\begin{align*}
&\frac{1}{N_{\one}}\sum_{\mo\in \cH_2} W^n\inp{\cE_{\mo,1}(P_{X\tilde{X}Y\tilde{Y}Z})|\vecx_{\mo}, \vecy} \\
& = \frac{1}{N_{\one}}\sum_{\substack{\mo:(\vecx_{\mo}, \vecx_{\tilde{m}_{\one}}, \vecy_{\mt}, \vecy)\in T^n_{X\tilde{X}\tilde{Y}Y},\\ \tilde{m}_{\one}\in \cM_{\one},\tilde{m}_{\one}\neq \mo,\mt\in\cM_{\two}}} \sum_{\vecz\in T^{n}_{Z|X\tilde{X}Y\tilde{Y}}(\vecx_{\mo},\vecx_{\tilde{m}_{\one}},\vecy,\vecy_{\tilde{m}_{\two}})}W^n\inp{\vecz|\vecx_{\mo}, \vecy}\\
&\leq \exp\inb{-n\epsilon/2}.
\end{align*}
Thus, it is sufficient to consider distributions $P_{X\tilde{X}Y\tilde{Y}Z}\in \cP_1$ for which 
\begin{align}
I\inp{X;\tilde{X}\tilde{Y}Y}\leq |R_{\one}- I(\tilde{X};\tilde{Y}Y)|^{+}+|R_{\two}-I(\tilde{Y};Y)|^{+}+\epsilon\label{eq:1_p}
\end{align}
For $P_{X\tilde{X}Y\tilde{Y}Z}\in \cP_1$ satisfying~\eqref{eq:1_p},
\begin{align}
&\sum_{\vecz\in \cE_{\mo,1}(P_{X\tilde{X}Y\tilde{Y}Z})}W^n(\vecz|\vecx_{\mo}, \vecy)\nonumber\\
&\qquad\leq\sum_{\substack{\tilde{m}_{\one}, \tilde{m}_{\two}:\\(\vecx_{\mo}, \vecx_{\tilde{m}_{\one}}, \vecy_{\tilde{m}_{\two}},\vecy)\in T^{n}_{X\tilde{X}\tilde{Y}Y}}}\sum_{\vecz:(\vecx_{\mo}, \vecx_{\tilde{m}_{\one}}, \vecy_{\tilde{m}_{\two}},\vecy, \vecz)\in T^{n}_{X\tilde{X}\tilde{Y}YZ}}W^n(\vecz|\vecx_{\mo}, \vecy)\nonumber\\
&\qquad\leq \sum_{\substack{\tilde{m}_{\one}, \tilde{m}_{\two}:\\(\vecx_{\mo}, \vecx_{\tilde{m}_{\one}}, \vecy_{\tilde{m}_{\two}},\vecy)\in T^{n}_{X\tilde{X}\tilde{Y}Y}}}\frac{|T^{n}_{Z|X\tilde{X}\tilde{Y}Y}(\vecx_{\mo},\vecx_{\tilde{m}_{\one}}, \vecy_{\tilde{m}_{\two}}, \vecy)|}{|T^n_{Z|XY}(\vecx_{\mo},\vecy)|}\nonumber\\
&\qquad \leq \sum_{\substack{\tilde{m}_{\one}, \tilde{m}_{\two}:\\(\vecx_{\mo}, \vecx_{\tilde{m}_{\one}}, \vecy_{\tilde{m}_{\two}},\vecy)\in T^{n}_{X\tilde{X}\tilde{Y}Y}}}\frac{\exp\inp{nH(Z|X\tilde{X}\tilde{Y}Y)}}{(n+1)^{-|\cX||\cY||\cZ|}\exp\inp{nH(Z|XY)}}\nonumber\\
&\qquad \leq \sum_{\substack{\tilde{m}_{\one}, \tilde{m}_{\two}:\\(\vecx_{\mo}, \vecx_{\tilde{m}_{\one}}, \vecy_{\tilde{m}_{\two}},\vecy)\in T^{n}_{X\tilde{X}\tilde{Y}Y}}}\exp\inp{-n\inp{I(Z;\tilde{X}\tilde{Y}|XY)-\epsilon}} \text{ for large }n.\nonumber\\
&\qquad\stackrel{\text{(a)}}{\leq}\exp\inp{n\inp{|R_{\one}- I(\tilde{X};\tilde{Y}XY)|^{+}+|R_{\two}-I(\tilde{Y};XY)|^{+}-I(Z;\tilde{X}\tilde{Y}|XY)+2\epsilon}}\label{eq:upperbound2_p}
\end{align}
where (a) follows using~\eqref{codebook:3b}. 
We see that
\begin{align*}
I(Z;\tilde{X}\tilde{Y}|XY) &= I(XZ;\tilde{X}\tilde{Y}|Y)-I(X;\tilde{X}\tilde{Y}|Y)\\
&\stackrel{\text{(a)}}{\geq} \eta-I(X;\tilde{X}\tilde{Y}Y)\\
&\stackrel{\text{(b)}}{\geq} \eta- |R_{\one}- I(\tilde{X};\tilde{Y}Y)|^{+}-|R_{\two}-I(\tilde{Y};Y)|^{+}-\epsilon
\end{align*} 
where (a) uses the condition $I(XZ;\tilde{X}\tilde{Y}|Y)\geq \eta$ from definition of $\cP_1$ and the fact that $I(X;\tilde{X}\tilde{Y}Y)\geq I(X;\tilde{X}\tilde{Y}|Y)$ and (b) follows from \eqref{eq:1_p}.
This implies that
\begin{align*}
&\sum_{\vecz\in \cE_{\mo,1}(P_{X\tilde{X}Y\tilde{Y}Z})}W^n(\vecz|\vecx_{\mo}, \vecy)
\\&\leq \exp\inp{n\inp{|R_{\one}- I(\tilde{X};\tilde{Y}XY)|^{+}+|R_{\two}-I(\tilde{Y};XY)|^{+}+|R_{\one}- I(\tilde{X};\tilde{Y}Y)|^{+}+|R_{\two}-I(\tilde{Y};Y)|^{+}-\eta+3\epsilon}}\\
&\leq \exp\inp{n\inp{4\delta-\eta+3\epsilon}}\\
&\rightarrow 0\text{ when }\eta>3\epsilon+4\delta.
\end{align*}

Now, we move on to the second term in the RHS of~\eqref{err:upperbound_p}. We see that by using~\eqref{codebook:4}, it is sufficient to consider distribution $P_{X\tilde{Y}_1\tilde{Y}_2YZ}\in \cP_2$ for which 
\begin{align}
I\inp{X;\tilde{Y}_1\tilde{Y}_2Y}\leq|R_{\two}-I(\tilde{Y}_1;Y)|^{+}+|R_{\two}-I(\tilde{Y}_2;\tilde{Y}_1 Y)|^{+} +\epsilon.\label{eq:2_p}
\end{align}
For $P_{X\tilde{Y}_1\tilde{Y}_2YZ}\in \cP_2$ satisfying~\eqref{eq:2_p},

\begin{align}
&\sum_{\vecz\in \cE_{\mo,\vecy, 2}(P_{X\tilde{Y}_1\tilde{Y}_2YZ})}W^n(\vecz|\vecx_{\mo}, \vecy)\nonumber\\
&\qquad\leq\sum_{\substack{\tilde{m}_{\two 1},  \tilde{m}_{\two 2}:\\(\vecx_{\mo}, \vecy_{\tilde{m}_{\two 1}}, \vecy_{\tilde{m}_{\two 2}},\vecy)\in T^{n}_{X\tilde{Y}_1\tilde{Y}_2Y}}}\sum_{\vecz:(\vecx_{\mo}, \vecy_{\tilde{m}_{\two 1}}, \vecy_{\tilde{m}_{\two 2}},\vecy, \vecz)\in T^{n}_{X\tilde{Y}_1\tilde{Y}_2YZ}}W^n(\vecz|\vecx_{\mo}, \vecy)\nonumber\\
&\qquad\leq \sum_{\substack{\tilde{m}_{\two 1}, \tilde{m}_{\two 2}:\\(\vecx_{\mo}, \vecy_{\tilde{m}_{\two 1}}, \vecy_{\tilde{m}_{\two 2}},\vecy)\in T^{n}_{X\tilde{Y}_1\tilde{Y}_2Y}}}\frac{|T^{n}_{Z|X\tilde{Y}_1\tilde{Y}_2Y}(\vecx_{\mo},\vecy_{\tilde{m}_{\two 1}}, \vecy_{\tilde{m}_{\two 2}}, \vecy)|}{|T^n_{Z|XY}(\vecx_{\mo},\vecy)|}\nonumber\\
&\qquad \leq \sum_{\substack{\tilde{m}_{\two 1}, \tilde{m}_{\two 2}:\\(\vecx_{\mo}, \vecy_{\tilde{m}_{\two 1}}, \vecy_{\tilde{m}_{\two 2}},\vecy)\in T^{n}_{X\tilde{Y}_1\tilde{Y}_2Y}}}\frac{\exp\inp{nH(Z|X\tilde{Y}_1\tilde{Y}_2Y)}}{(n+1)^{-|\cX||\cY||\cZ|}\exp\inp{nH(Z|XY)}}\nonumber\\
&\qquad \leq \sum_{\substack{\tilde{m}_{\two 1}, \tilde{m}_{\two 2}:\\(\vecx_{\mo}, \vecy_{\tilde{m}_{\two 1}}, \vecy_{\tilde{m}_{\two 2}},\vecy)\in T^{n}_{X\tilde{Y}_1\tilde{Y}_2Y}}}\exp\inp{-n\inp{I(Z;\tilde{Y}_1\tilde{Y}_2|XY)-\epsilon}} \text{ for large }n.\nonumber\\
&\qquad\stackrel{\text{(a)}}{\leq}\exp\inp{n\inp{|R_{\two}-I(\tilde{Y}_1;XY)|^{+}+|R_{\two}-I(\tilde{Y}_2;\tilde{Y}_1 XY)|^{+}-I(Z;\tilde{Y}_1\tilde{Y}_2|XY)+2\epsilon}}\\
&\qquad\stackrel{\text{(b)}}{\leq}\exp\Big(n\Big(|R_{\two}-I(\tilde{Y}_1;XY)|^{+}+|R_{\two}-I(\tilde{Y}_2;\tilde{Y}_1 XY)|^{+}+|R_{\two}-I(\tilde{Y}_1;Y)|^{+}\\
&\qquad\qquad+|R_{\two}-I(\tilde{Y}_2;\tilde{Y}_1 Y)|^{+}-\eta+3\epsilon\Big)\Big)\\
&\qquad\leq \exp\inp{n\inp{4\delta-\eta+3\epsilon}}\\
&\qquad\rightarrow 0\text{ when }\eta>3\epsilon+4\delta.
\end{align}
where (a) follows using~\eqref{codebook:5} and (b) follows from \eqref{eq:2_p} and definition of $\cP_{2}$.\\
Similarly, we can show that if $\eta>3\epsilon+4\delta$ the probability of error goes to zero with $n$ when user \one is malicious.

\end{proof}

\section{Codebook for Theorem~\ref{thm:main_result} and~\ref{thm:inner_bd}} 
%\textcolor{red}{two equations here are redundant (not used in achievabilitiy)}\\
\begin{lemma}[codebook lemma]\label{lemma:codebook}
Suppose $\mathcal{X,Y,Z}$ are finite. Let $P_{\one}\in \cP^n_{\cX}$ and $P_{\two}\in \cP^n_{\cY}$. For any $\epsilon>0$, there exists $n_0(\epsilon)$ such that for all $n\geq n_0(\epsilon),\, N_{\one}, N_{\two}\geq \exp(n\epsilon)$, there exists codebooks $\{\vecx_1, \vecx_2, \ldots, \vecx_{N_{\one}}\}$ of type $P_{\one}$ and $\{\vecy_1, \vecy_2, \ldots, \vecy_{N_{\two}}\}$ of type $P_{\two}$ such that for every $\vecx, \vecx'\in \cX^n$  and $\vecy, \vecy'\in \cY^n$, and joint types $P_{X\tilde{X}\tilde{Y}Y}\in \cP^n_{\cX\times \cX\times \cY\times \cY}$ and $P_{X'\tilde{Y}_1\tilde{Y}_2Y'}\in \cP^n_{\cX\times \cY\times \cY\times \cY}$ such that $P_{X} = P_{\tilde{X}} = P_{X'} = P_{\one}$, $P_{\tilde{Y}_1} = P_{\tilde{Y}_1} = P_{\two}$, $(\vecx, \vecy)\in T^n_{XY}$ and $(\vecx', \vecy')\in T^n_{X'Y'}$, and for $R_{\one}\defineqq (1/n)\log N_{\one}$ and $R_{\two}\defineqq (1/n)\log N_{\two}$ where $R_{\one}\leq H(X)$ and $R_{\two}\leq H(\tilde{Y}_1)=H(\tilde{Y}_2)$, the following holds: 
\begin{align}
&\frac{\left|\inb{m_{\one}: (\vecx_{\mo}, \vecy)\in T_{XY}^n}\right|}{N_{\one}}\leq \exp\inb{-n\epsilon/2},\text{ if }I(X;Y)> \epsilon\label{codebook:1}\\
%&\frac{\left|\inb{\mo:(\vecx_{\mo}, \vecx_{\tilde{m}_{\one}}, \vecy_{\mt}, \vecy)\in T^n_{X\tilde{X}\tilde{Y}Y} \text{ for some }\tilde{m}_{\one}\neq \mo\text{ and some }\mt} \right|}{N_{\one}} \leq \exp\inb{-n\epsilon/2} \nonumber\\
%& \qquad \text{ if }I(X;\tilde{X}\tilde{Y}Y)-|R_{\one}- I(\tilde{X};Y)|^{+}-|R_{\two}-I(\tilde{Y};\tilde{X}Y)|^{+}       >\epsilon\label{codebook:2a}\\
&\frac{\left|\inb{\mo:(\vecx_{\mo}, \vecx_{\tilde{m}_{\one}}, \vecy_{\mt}, \vecy)\in T^n_{X\tilde{X}\tilde{Y}Y} \text{ for some }\tilde{m}_{\one}\neq \mo\text{ and some }\mt} \right|}{N_{\one}} \leq \exp\inb{-n\epsilon/2} \nonumber\\
& \qquad \text{ if }I(X;\tilde{X}\tilde{Y}Y)-|R_{\one}- I(\tilde{X};\tilde{Y}Y)|^{+}-|R_{\two}-I(\tilde{Y};Y)|^{+}       >\epsilon\label{codebook:2b}\\
%&\left|\inb{(\tilde{m}_{\one}, \tilde{m}_{\two}):(\vecx, \vecx_{\tilde{m}_{\one}}, \vecy_{\tilde{m}_{\two}}, \vecy)\in T^n_{X\tilde{X}\tilde{Y}Y}}\right|\nonumber\\
%&\qquad\leq \exp\inb{n\inp{|R_{\one}- I(\tilde{X};XY)|^{+}+|R_{\two}-I(\tilde{Y};\tilde{X}XY)|^{+}+\epsilon}}\label{codebook:3a}\\
&\left|\inb{(\tilde{m}_{\one}, \tilde{m}_{\two}):(\vecx, \vecx_{\tilde{m}_{\one}}, \vecy_{\tilde{m}_{\two}}, \vecy)\in T^n_{X\tilde{X}\tilde{Y}Y}}\right|\nonumber\\
&\qquad\leq \exp\inb{n\inp{|R_{\one}- I(\tilde{X};\tilde{Y}XY)|^{+}+|R_{\two}-I(\tilde{Y};XY)|^{+}+\epsilon}}\label{codebook:3b}\\
&\frac{\left|\inb{\mo:(\vecx_{\mo}, \vecy_{\tilde{m}_{\two 1}}, \vecy_{\tilde{m}_{\two 2}}, \vecy')\in T^n_{X'\tilde{Y}_1\tilde{Y}_2Y'} \text{ for some }\tilde{m}_{\two 1},\tilde{m}_{\two 2}, } \right|}{N_{\one}} \leq \exp\inb{-n\epsilon/2} \nonumber\\
& \qquad \text{ if }I(X';\tilde{Y}_1\tilde{Y}_2Y')-|R_{\two}-I(\tilde{Y}_1;Y')|^{+}-|R_{\two}-I(\tilde{Y}_2;\tilde{Y}_1 Y')|^{+} >\epsilon\label{codebook:4}\\
&\text{and }\left|\inb{(\tilde{m}_{\two 1}, \tilde{m}_{\two 2}):(\vecx', \vecy_{\tilde{m}_{\two 1}}, \vecy_{\tilde{m}_{\two 2}}, \vecy')\in T^n_{X'\tilde{Y}_1\tilde{Y}_2Y'} }\right|\nonumber\\
&\qquad\leq \exp\inb{n\inp{|R_{\two}-I(\tilde{Y}_1;X'Y')|^{+}+|R_{\two}-I(\tilde{Y}_2;\tilde{Y}_1 X'Y')|^{+}+\epsilon}}.\label{codebook:5}
\end{align}
Analogous statements hold when the roles of users \one and \two are interchanged.
\end{lemma}

\begin{proof}
This proof is along the lines of the proof of \cite[Lemma 3]{CsiszarN88}. We will generate the codebook by a random experiment. For fixed $\vecx, \vecx'$, $ \vecy,\vecy'$,$P_{X\tilde{X}\tilde{Y}Y}$ and $P_{X'\tilde{Y}_1\tilde{Y}_2Y'}$ satisfying the conditions of the Lemma, we will show that the probability that each of the  statements~\eqref{codebook:1}~-~\eqref{codebook:5} does not hold falls doubly exponentially in $n$. Since, $|\cX^n|$, $|\cY^n|$, $|\cP_{\cX\times \cX\times \cY\times \cY}|$ and $|\cP_{\cX\times \cY\times \cY\times \cY}|$ grow at most exponentially in $n$, a union bound will imply that the probability that any of the statements~\eqref{codebook:1}~-~\eqref{codebook:5} fail for some $\vecx, \vecx',\vecy,\vecy',$$P_{X\tilde{X}\tilde{Y}Y}$ and $P_{X'\tilde{Y}_1\tilde{Y}_2Y'}$ also falls doubly exponentially. This will show existence of a codebook which satisfies~\eqref{codebook:1}~-~\eqref{codebook:5}. The proof will employ \cite[Lemma A1]{CsiszarN88}, which is stated below.
\begin{lemma}\cite[Lemma A1]{CsiszarN88}\label{lemma:A1}
Let $Z_1, \ldots, Z_{N}$ be arbitrary random variables, and let $f_i(Z_1, \ldots, Z_i)$ be arbitrary with $0\leq f_i\leq 1$, $i = 1, \ldots, N$. Then the condition
\begin{align}\label{eq:A1}
E\insq{f_i(Z_1,\ldots, Z_i)|Z_1, \ldots, Z_{i-1}}\leq a \text{ a.s.,} \quad i = 1, \ldots, N, 
\end{align}
implies that 
\begin{align}\label{eq:A2}
\bbP\inb{\frac{1}{N}\sum_{i = 1}^{N}f_i(Z_1, \ldots, Z_i)>t}\leq \exp{\inb{-N(t-a\log{e})}}.
\end{align}
\end{lemma}

We denote the  type classes of $P_{\one}$ and $P_{\two}$ by $T^n_{\one}$ and $T^n_{\two}$ respectively. Let $\vecX_1, \vecX_2, \ldots, \vecX_{N_{\one}}$ be independent random vectors each uniformly distributed on $T^n_{\one}$ and $\vecY_1, \vecY_2, \ldots, \vecY_{N_{\two}}$ be another set of independent random vectors (independent of $\vecX_1, \vecX_2, \ldots, \vecX_{N_{\one}}$) with each element uniformly distributed on $T^n_{\two}$.  $(\vecX_1, \vecX_2, \ldots, \vecX_{N_{\one}})$ and $(\vecY_1, \vecY_2, \ldots, \vecY_{N_{\two}})$ are the random codebooks for user \one and \two respectively.
Fix $P_{X\tilde{X}\tilde{Y}Y}\in \cP_{\cX\times \cX\times \cY\times \cY}$, $P_{X'\tilde{Y}_1\tilde{Y}_2Y'}\in \cP^n_{\cX\times \cY\times \cY\times \cY}$, $\vecx,\vecx'\in T^n_{\one}$ and $\vecy,\vecy'\in \cY^n$ such that $P_X =P_X'= P_{\tilde{X}} = P_{\one}$, $P_{\tilde{Y}_1} = P_{\tilde{Y}_1} = P_{\two}$, $(\vecx,\vecy)\in T^n_{XY}$ and $(\vecx', \vecy')\in T^n_{X'Y'}$.

\noindent\underline{\em Analysis of \eqref{codebook:3b}}\\
%We will start with analysis of \eqref{codebook:3b}. 
Define
\begin{align}\label{func:g_i}
g_i(\vecy_1, \vecy_2, \ldots, \vecy_i) \defineqq \begin{cases}
    1, & \text{if } \vecy_i\in T^n_{\tilde{Y}|{X}Y}(\vecx,\vecy) \\
    0, & \text{otherwise,}
   \end{cases}
\end{align}
and for  $\tilde{\vecy}\in T^n_{\tilde{Y}|{X}Y}(\vecx,\vecy)$,
\begin{align}\label{func:h_i}
h^{\tilde{\vecy}}_i(\vecx_1, \vecx_2, \ldots, \vecx_i) \defineqq \begin{cases}
    1, & \text{if } \vecx_i\in T^n_{\tilde{X}|\tilde{Y}XY}(\tilde{\vecy}, \vecx,\vecy) \\
    0, & \text{otherwise.}
   \end{cases}
\end{align}
Let events $\cE, \cE_1$ and $\cE^{\tilde{\vecy}}_2$ be defined as 
\begin{align*}
\cE = &\Big\{\left|\inb{(\tilde{m}_{\one}, \tilde{m}_{\two}):(\vecx, \vecX_{\tilde{m}_{\one}}, \vecY_{\tilde{m}_{\two}}, \vecy)\in T^n_{X\tilde{X}\tilde{Y}Y}}\right|\\
&\qquad> \exp\inb{n\inp{|R_{\one}- I(\tilde{X};XY\tilde{Y})|^{+}+|R_{\two}-I(\tilde{Y};XY)|^{+}+\epsilon}}\Big\},\\
\cE_1 &= \inb{\sum_{i = 1}^{N_{\two}}g_i(\vecY_1, \vecY_2, \ldots, \vecY_i)>\exp{\inb{n\inp{|R_{\two}-I(\tilde{Y};XY)|^{+}+\frac{\epsilon}{2}}}}}\text{, and }\\
\cE_2^{\tilde{\vecy}} &= \inb{\sum_{j = 1}^{N_{\one}} h^{\tilde{\vecy}}_j(\vecX_1, \vecX_2, \ldots, \vecX_j)>\exp{\inb{n\inp{|R_{\one}-I(\tilde{X};\tilde{Y}XY)|^{+}+\frac{\epsilon}{2}}}}}.
\end{align*}
We note that
\begin{align*}
&\left|\inb{(\tilde{m}_{\one}, \tilde{m}_{\two}):(\vecx, \vecX_{\tilde{m}_{\one}}, \vecY_{\tilde{m}_{\two}}, \vecy)\in T^n_{X\tilde{X}\tilde{Y}Y}}\right|\\
&\qquad \qquad \qquad= \sum_{i = 1}^{N_{\two}}g_i(\vecY_1, \vecY_2, \ldots, \vecY_i)\inp{\sum_{j = 1}^{N_{\one}} h^{\vecY_i}_j(\vecX_1, \vecX_2, \ldots, \vecX_j)}.
\end{align*}
Thus, $ \cE \subseteq \inp{\cup_{\tilde{\vecy}\in T_{\tilde{Y}|XY}(\vecx, \vecy)}\cE_2^{\tilde{\vecy}}}\cup \cE_1$. In order to apply Lemma~\ref{lemma:A1} to \eqref{func:g_i} with $(\vecY_1, \ldots, \vecY_{N_{\two}})$ as the random variables $(Z_1, \ldots, Z_N)$, we note that 
\begin{align*}
E\insq{g_i(\vecY_1,\ldots, \vecY_i)|\vecY_1, \ldots, \vecY_{i-1}} = &\bbP\inb{\vecY_i\in T^n_{\tilde{Y}|{X}Y}(\vecx,\vecy)}\\
=&\frac{|T^n_{\tilde{Y}|{X}Y}(\vecx,\vecy)|}{|T^n_{\two}|}\\
\stackrel{\text{(a)}}{\leq}&\frac{\exp\inp{nH(\tilde{Y}|{X}Y)}}{(n+1)^{-|\cY|}\exp\inp{nH(\tilde{Y})}}\\
=&(n+1)^{|\cY|}\exp{\inp{-nI(\tilde{Y};XY)}},
\end{align*}
where (a) follows because $P_{\two}=P_{\tilde{Y}}$ and thus $|T^n_{\two}| = |T^n_{\tilde{Y}}|$. Taking $t =\frac{1}{N_{\two}}\exp{\inb{n\inp{|R_{\two}-I(\tilde{Y};XY)|^{+}+\frac{\epsilon}{2}}}}$ and $n\geq n_1(\epsilon)$, where $n_1(\epsilon) \defineqq $$\min{\inb{n:(n+1)^{|\cY|}\log{e}<\frac{1}{2}\exp(\frac{n\epsilon}{2})}},$
we see that $N_{\two}(t-a\log{e})\geq (1/2)\exp(n\frac{n\epsilon}{2})$. Using~\eqref{eq:A2}, this gives us
\begin{align} \label{eq:e1}
  \bbP(\cE_1) \leq \exp\inb{-\frac{1}{2}\exp\inb{\frac{n\epsilon}{2}}}.
\end{align}
Similarly, we apply Lemma~\ref{lemma:A1} to \eqref{func:h_i} with $(\vecX_1, \ldots, \vecX_{N_{\one}})$ as the random variables $(Z_1, \ldots, Z_N)$. We can show that $a= (n+1)^{|\cX|}\exp{\inp{-nI(\tilde{X};\tilde{Y}XY)}}$ satisfies \eqref{eq:A1}. We take $t =\frac{1}{N_{\one}}\exp{\inb{n\inp{|R_{\one}-I(\tilde{X};\tilde{Y}XY)|^{+}+\frac{\epsilon}{2}}}}$ and $n\geq n_2(\epsilon)$ where $n_2(\epsilon) \defineqq \min{\inb{n:(n+1)^{|\cX|}\log{e}<\frac{1}{2}\exp(\frac{n\epsilon}{2})}}$. This gives $N_{\one}(t-a\log{e})\geq (1/2)\exp(\frac{n \epsilon}{2})$ which, when plugged in \eqref{eq:A2}, gives 
\begin{align}\label{eq:e2}
\bbP\inp{\cE_2^{\tilde{\vecy}}} \leq \exp\inb{-\frac{1}{2}\exp\inb{\frac{n\epsilon}{2}}}.
\end{align}
Using \eqref{eq:e1} and \eqref{eq:e2}, 
\begin{align}\label{eq:prob1}
\bbP\inp{\cE}\leq \inp{|T^n_{\tilde{Y}|XY}(\vecx, \vecy)|+1}\exp\inb{-\frac{1}{2}\exp\inb{\frac{n\epsilon}{2}}}.
\end{align}
This shows that the probability that \eqref{codebook:3b} does not hold falls doubly exponentially. %Using a similar analysis, one can show that the probability with which \eqref{codebook:3a} does not hold also falls doubly exponentially. 

\noindent\underline{\em Analysis of \eqref{codebook:1}} \\
We will use the same arguments as used in obtaining \eqref{eq:e2}. We replace $\tilde{X}$ with $X$, $(\tilde{Y}, X, Y)$ with $Y$, to obtain
\begin{align*}
\bbP\inb{\left|\inb{m_{\one}: (\vecx_{\mo}, \vecy)\in T_{XY}^n, }\right|> \exp\inb{n\inp{|R_{\one}-I(X;Y)|^{+}+\frac{\epsilon}{2}}}}\leq \exp\inb{-\frac{1}{2}\exp\inb{\frac{n \epsilon}{2}}}.
\end{align*}
So, 
\begin{align*}
\bbP\inb{\frac{1}{N_{\one}}\left|\inb{m_{\one}: (\vecx_{\mo}, \vecy)\in T_{XY}^n, }\right|> \exp\inb{n\inp{|R_{\one}-I(X;Y)|^{+}-R_{\one}+\frac{\epsilon}{2}}}}\leq \exp\inb{-\frac{1}{2}\exp\inb{\frac{n \epsilon}{2}}}.
\end{align*}
We are given that $I(X;Y)> \epsilon$. When $R_{\one}>I(X;Y)$, we have  $|R_{\one}-I(X;Y)|^{+}$$-R_{\one}+\frac{\epsilon}{2}$ = $\frac{\epsilon}{2} - I(X;Y)\leq  -\frac{\epsilon}{2}$. When $R_{\one}\leq I(X;Y)$, we have  $|R_{\one}-I(X;Y)|^{+}$$-R_{\one}+\frac{\epsilon}{2}$ = $\frac{\epsilon}{2} - R_{\one}\leq  -\frac{\epsilon}{2}$ (because $R\geq \epsilon$). Thus 
\begin{align*}
&\bbP\inb{\frac{1}{N_{\one}}\left|\inb{m_{\one}: (\vecx_{\mo}, \vecy)\in T_{XY}^n, }\right|> \exp\inb{\frac{-n\epsilon}{2}}}\\
&\leq \bbP\inb{\frac{1}{N_{\one}}\left|\inb{m_{\one}: (\vecx_{\mo}, \vecy)\in T_{XY}^n, }\right|> \exp\inb{n\inp{|R_{\one}-I(X;Y)|^{+}-R_{\one}+\frac{\epsilon}{2}}}}\\
&\leq \exp\inb{-\frac{1}{2}\exp\inb{\frac{n \epsilon}{2}}}.
\end{align*}

\noindent\underline{\em Analyses of \eqref{codebook:2b}}\\
%We first prove~\eqref{codebook:2b}. 
For $i\in [1:N_{\one}]$, let $A_i$ be the set of indices $(j,k)\in [1:N_{\one}]\times[1:N_{\two}], j<i$ such that $(\vecx_j, \vecy_k)\in T^{n}_{\tilde{X}\tilde{Y}|Y}(\vecy)$ provided $|A_i|\leq \exp{\inb{n\inp{|R_{\one}- I(\tilde{X};\tilde{Y}Y)|^{+}+|R_{\two}-I(\tilde{Y};Y)|^{+}} +\frac{\epsilon}{4}}}$. Otherwise, $A_i = \emptyset$.
Let 
\begin{align}\label{func:f_i}
f_i^{[\vecy_1,\vecy_2,\ldots,\vecy_{N_{\two}}]}\inp{\vecx_1,\vecx_2,\ldots, \vecx_i} = \begin{cases}
    1, & \text{if } \vecx_i\in \cup_{(j,k)\in A_i} T^n_{X|\tilde{X}\tilde{Y}Y}(\vecx_j,\vecy_k,\vecy) \\
    0, & \text{otherwise.}
   \end{cases}
\end{align}
Then, 
\begin{align}
\bbP&\inb{\sum_{i=1}^{N_\one}f^{[\vecY_1,\vecY_2,\ldots,\vecY_{N_{\two}}]}_{i}\inp{\vecX_1,\vecX_2,\ldots, \vecX_i}\neq \left|\inb{i:\vecX_i\in T^n_{X|\tilde{X}\tilde{Y}Y}(\vecX_j,\vecY_k,\vecy) \text{ for some }j<i \text{ and some }k}\right|}\nonumber\\
=&\bbP\inb{\left|\inb{(\tilde{m}_{\one}, \tilde{m}_{\two}):(\vecX_{\tilde{m}_{\one}}, \vecY_{\tilde{m}_{\two}}, \vecy)\in T^n_{\tilde{X}\tilde{Y}Y}}\right|> \exp{\inb{n\inp{|R_{\one}- I(\tilde{X};\tilde{Y}Y)|^{+}+|R_{\two}-I(\tilde{Y};Y)|^{+}} +\frac{\epsilon}{4}}}}\nonumber\\
\leq& \inp{|T_{\tilde{Y}|Y}(\vecy)|+1}\exp\inb{-\frac{1}{2}\exp\inb{\frac{n\epsilon}{8}}},\label{eq:conditioning}
\end{align}
where the last inequality can be obtained from the definition of event $\cE$ and \eqref{eq:prob1} where we replace $(\vecX, \vecY)$ with $\vecY$, $(\vecx, \vecy)$ with $\vecy$, and $\epsilon$ with $\epsilon/4$.

For $\vecy_i\in T^n_{\two}, \, i = 1, \ldots, N_{\two}$, we will apply Lemma~\ref{lemma:A1} on $f_i^{[\vecy_1,\vecy_2,\ldots,\vecy_{N_{\two}}]}$ with $(\vecX_1,\ldots,\vecX_{N_{\one}})$ as the random variables $(Z_1, \ldots, Z_N)$. We will first compute the value of $a$ in \eqref{eq:A1}. We note that, for $i\in [1:N_{\one}]$,  $E\insq{f^{[\vecy_1,\vecy_2,\ldots,\vecy_{N_{\two}}]}_{i}\inp{\vecX_1,\vecX_2,\ldots, \vecX_i}\Big|\vecX_1,\vecX_2,\ldots, \vecX_{i-1}}$, being a random function of $(\vecX_1,\vecX_2,\ldots, \vecX_{i-1})$, is a random variable. We will compute it for $(\vecX_1,\vecX_2,\ldots, \vecX_{i-1})= (\vecx_1,\vecx_2,\ldots, \vecx_{i-1})$.

\begin{align*}
&E\insq{f^{[\vecy_1,\vecy_2,\ldots,\vecy_{N_{\two}}]}_{i}\inp{\vecX_1,\vecX_2,\ldots, \vecX_i}\Big|(\vecX_1,\vecX_2,\ldots, \vecX_{i-1}) = (\vecx_1,\vecx_2,\ldots, \vecx_{i-1})}\\
&=\bbP\inp{\vecX_i\in \cup_{(j,k)\in A_i} T^n_{X|\tilde{X}\tilde{Y}Y}(\vecx_j,\vecy_k,\vecy)}\\
&\stackrel{\text{(a)}}{\leq} |A_{i}|\frac{\exp\inb{nH(X|\tilde{X}\tilde{Y}Y)}}{(n+1)^{-|\cX|}\exp(nH(X))}\\
%&=|A_{i}|\frac{\exp\inb{nH(X|\tilde{X}\tilde{Y}Y)}}{(n+1)^{-|\cX|}\exp(nH(X))}\\
%&\leq (n+1)^{|\cX|}\exp\Big\{n\inp{|R_{\one}- I(\tilde{X};\tilde{Y}Y)|^{+}+|R_{\two}-I(\tilde{Y};Y)|^{+}}\\
%&\qquad \qquad \qquad -I(X;\tilde{X}\tilde{Y}Y)+\frac{\epsilon}{4}\Big\}\\
& = (n+1)^{|\cX|}\exp\inb{n\inp{|R_{\one}- I(\tilde{X};\tilde{Y}Y)|^{+}+|R_{\two}-I(\tilde{Y};Y)|^{+}} -I(X;\tilde{X}\tilde{Y}Y)+\frac{\epsilon}{4}},
\end{align*}
where (a) follows by union bound over $(j,k)\in A_{i}$ and by noting that $|T^n_{\one}| = |T^n_{X}|$. For all $i\in [1:N_{\one}]$, this upper bound holds for every realization of $(\vecX_1,\vecX_2,\ldots, \vecX_{i-1})$. Thus, in \eqref{eq:A1}, we may take $a=(n+1)^{|\cX|}\exp$$\Big\{n\Big(|R_{\one}- I(\tilde{X};\tilde{Y}Y)|^{+}$$+|R_{\two}-I(\tilde{Y};Y)|^{+}\Big) -I(X;\tilde{X}\tilde{Y}Y)+\frac{\epsilon}{4}\Big\}$.
If $I(X;\tilde{X}\tilde{Y}Y) >|R_{\one}- I(\tilde{X};\tilde{Y}Y)|^{+}+|R_{\two}-I(\tilde{Y};Y)|^{+} +\epsilon$ (as postulated in~\eqref{codebook:2b}), \eqref{eq:A1} holds with $a = (n+1)^{|\cX|}\exp\inb{-\frac{3}{4}n\epsilon}$. For $t = \exp\inb{\frac{-n\epsilon}{2}}$ and $n\geq n_2(\epsilon)$ with $n_2(\epsilon) \defineqq \min{\inb{n:(n+1)^{|\cX|}\log e<\frac{1}{2}\exp{\inb{\frac{n\epsilon}{4}}}}},$ 
we get
\begin{align}
\bbP&\inb{\frac{1}{N_{\one}}\sum_{i=1}^{N_{\one}}f_i^{[\vecy_1,\vecy_2,\ldots,\vecy_{N_{\two}}]}\inp{\vecX_1,\vecX_2,\ldots, \vecX_i}>\exp\inb{\frac{-n\epsilon}{2}}}\nonumber \\
&\leq \exp\inb{-\frac{N_{\one}}{2}\exp\inb{-\frac{n\epsilon}{2}}}\nonumber\\
&\leq \exp\inb{-\frac{1}{2}\exp\inb{\frac{n\epsilon}{2}}}, \nonumber
\end{align}
where the last inequality uses the assumption that $N_\one\geq \exp\inb{n\epsilon}$.  \\
Averaging over $(\vecY_1, \ldots, \vecY_{\two})$, we get
\begin{align}
\bbP&\inb{\frac{1}{N_{\one}}\sum_{i=1}^{N_{\one}}f_i^{[\vecY_1,\vecY_2,\ldots,\vecY_{N_{\two}}]}\inp{\vecX_1,\vecX_2,\ldots, \vecX_i}>\exp\inb{\frac{-n\epsilon}{2}}}\nonumber \\
&\leq \exp\inb{-\frac{1}{2}\exp\inb{\frac{n\epsilon}{2}}}. \label{eq:prob}
\end{align}
Let events $\cF_1$ and $\cF_2$ be defined as 
\begin{align*}
\cF_1 &= \inb{\frac{1}{N_{\one}}\left|\inb{i:\vecX_i\in T^n_{X|\tilde{X}\tilde{Y}Y}(\vecX_j, \vecY_k, \vecy)\text{ for some}, j<i \text{ and }k}\right|>\exp\inb{\frac{-n\epsilon}{2}}},\\
\cF_2 &= \inb{\sum_{i=1}^{N_\one}f^{[\vecY_1,\vecY_2,\ldots,\vecY_{N_{\two}}]}_{i}\inp{\vecX_1,\vecX_2,\ldots, \vecX_i}\neq \left|\inb{i:\vecX_i\in T^n_{X|\tilde{X}\tilde{Y}Y}(\vecX_j,\vecY_k,\vecy) \text{ for some }j<i \text{ and some }k}\right|},\\
\cF_3 &= \inb{\sum_{i=1}^{N_\one}f^{[\vecY_1,\vecY_2,\ldots,\vecY_{N_{\two}}]}_{i}\inp{\vecX_1,\vecX_2,\ldots, \vecX_i}>\exp\inb{\frac{-n\epsilon}{2}}}.
\end{align*}
We are interested in $\bbP\inp{\cF_1}$. We see that 
\begin{align*}
\bbP\inp{\cF_1}& = \bbP\inp{\cF_1\cap \cF_2} +\bbP\inp{\cF_1\cap \cF^{c}_2}\\
& \leq \bbP\inp{\cF_2} +\bbP\inp{\cF_1\cap \cF^{c}_2}\\
& \leq \bbP\inp{\cF_2} +\bbP\inp{\cF_3}\\
&\stackrel{\text{(a)}}{\leq}\inp{|T_{\tilde{Y}|Y}(\vecy)|+1}\exp\inb{-\frac{1}{2}\exp\inb{\frac{n\epsilon}{8}}}+\exp\inb{-\frac{1}{2}\exp\inb{\frac{n\epsilon}{2}}}\\
& \leq \inp{|T_{\tilde{Y}|Y}(\vecy)|+2}\exp\inb{-\frac{1}{2}\exp\inb{\frac{n\epsilon}{8}}},
\end{align*}
where (a) follows from \eqref{eq:conditioning} and \eqref{eq:prob}. Thus, 
\begin{align*}
\bbP&\inp{\frac{1}{N_{\one}}\left|\inb{i:\vecX_i\in T^n_{X|\tilde{X}\tilde{Y}Y}(\vecX_j, \vecY_k, \vecy)\text{ for some}, j<i \text{ and }k}\right|>\exp\inb{\frac{-n\epsilon}{2}}}\\
&\leq \inp{|T_{\tilde{Y}|Y}(\vecy)|+2}\exp\inb{-\frac{1}{2}\exp\inb{\frac{n\epsilon}{8}}}.
\end{align*}
By symmetry, we get the same upper bound when $j>i$. Thus,
\begin{align*}
\bbP&\inb{\frac{\left|\inb{\mo:(\vecX_{\mo}, \vecX_{\tilde{m}_{\one}}, \vecY_{\mt}, \vecy)\in T^n_{X\tilde{X}\tilde{Y}Y} \text{ for some }\tilde{m}_{\one}\neq \mo\text{ and some }\mt} \right|}{N_{\one}} >\exp\inb{-n\epsilon/2} }\\
&< 2\inp{|T_{\tilde{Y}|Y}(\vecy)|+2}\exp\inb{-\frac{1}{2}\exp\inb{\frac{n\epsilon}{8}}}.
\end{align*}
This completes the analysis for \eqref{codebook:2b}. 
%Analysis for \eqref{codebook:2a} by appropriately defining~\eqref{func:f_i} and following similar arguments.

\noindent\underline{\em Analysis of \eqref{codebook:5}}\\
We will split the analysis in two parts as suggested by the inequalities below.
\begin{align*}
&\bbP\inb{\left|\inb{(\tilde{m}_{\two 1}, \tilde{m}_{\two 2}):(\vecx', \vecy_{\tilde{m}_{\two 1}}, \vecy_{\tilde{m}_{\two 2}}, \vecy')\in T^n_{X'\tilde{Y}_1\tilde{Y}_2Y'} }\right|
> \exp\inb{n\inp{|R_{\two}-I(\tilde{Y}_1;X'Y')|^{+}+|R_{\two}-I(\tilde{Y}_2;\tilde{Y}_1 X'Y')|^{+}+\epsilon}}}\\
&\leq \bbP\Big\{\left|\inb{(\tilde{m}_{\two 1}, \tilde{m}_{\two 2}):\tilde{m}_{\two 1}\neq  \tilde{m}_{\two 2},(\vecx', \vecY_{\tilde{m}_{\two 1}}, \vecY_{\tilde{m}_{\two 2}}, \vecy')\in T^n_{X'\tilde{Y}_1\tilde{Y}_2Y'} }\right|\\
 &\qquad\qquad>1/2\exp\inb{n\inp{|R_{\two}-I(\tilde{Y}_1;X'Y')|^{+}+|R_{\two}-I(\tilde{Y}_2;\tilde{Y}_1 X'Y')|^{+}+\epsilon}}\Big\}\\
&+\bbP\Big\{\left|\inb{(\tilde{m}_{\two 1}, \tilde{m}_{\two 2}):\tilde{m}_{\two 1}= \tilde{m}_{\two 2},(\vecx', \vecY_{\tilde{m}_{\two 1}}, \vecY_{\tilde{m}_{\two 2}}, \vecy')\in T^n_{X'\tilde{Y}_1\tilde{Y}_2Y'} }\right|\\
&\qquad\qquad> 1/2\exp\inb{n\inp{|R_{\two}-I(\tilde{Y}_1;X'Y')|^{+}+|R_{\two}-I(\tilde{Y}_2;\tilde{Y}_1 X'Y')|^{+}+\epsilon}}\Big\}\\
&\leq \bbP\Big\{\left|\inb{(\tilde{m}_{\two 1}, \tilde{m}_{\two 2}):\tilde{m}_{\two 1}\neq  \tilde{m}_{\two 2},(\vecx', \vecY_{\tilde{m}_{\two 1}}, \vecY_{\tilde{m}_{\two 2}}, \vecy')\in T^n_{X'\tilde{Y}_1\tilde{Y}_2Y'} }\right|\\
 &\qquad\qquad>\exp\inb{n\inp{|R_{\two}-I(\tilde{Y}_1;X'Y')|^{+}+|R_{\two}-I(\tilde{Y}_2;\tilde{Y}_1 X'Y')|^{+}+\epsilon'}}\Big\}\\
&+\bbP\Big\{\left|\inb{(\tilde{m}_{\two 1}, \tilde{m}_{\two 2}):\tilde{m}_{\two 1}= \tilde{m}_{\two 2},(\vecx', \vecY_{\tilde{m}_{\two 1}}, \vecY_{\tilde{m}_{\two 2}}, \vecy')\in T^n_{X'\tilde{Y}_1\tilde{Y}_2Y'} }\right|\\
&\qquad\qquad> \exp\inb{n\inp{|R_{\two}-I(\tilde{Y}_1;X'Y')|^{+}+|R_{\two}-I(\tilde{Y}_2;\tilde{Y}_1 X'Y')|^{+}+\epsilon'}}\Big\}
\end{align*}
for $\epsilon' = \epsilon/2$.
%Note that when $\tilde{Y}_1 \neq \tilde{Y}_2$, then if $(\tilde{m}_{\two 1}, \tilde{m}_{\two 2}):(\vecx', \vecY_{\tilde{m}_{\two 1}}, \vecY_{\tilde{m}_{\two 2}}, \vecy')\in T^n_{X'\tilde{Y}_1\tilde{Y}_2Y'}$ then $\tilde{m}_{\two 1}\neq \tilde{m}_{\two 2}$.
We first consider the case when $\tilde{m}_{\two 1}\neq \tilde{m}_{\two 2}$.

We follow arguments similar to those for \eqref{codebook:3b} and get the upper bound. We define
\begin{align}\label{func:tg}
\tilde{g}_i(\vecy_1, \vecy_2, \ldots, \vecy_i) \defineqq \begin{cases}
    1, & \text{if } \vecy_i\in T^n_{\tilde{Y}_1|{X}'Y'}(\vecx',\vecy') \\
    0, & \text{otherwise.}
   \end{cases}
\end{align}
For $\tilde{\vecy}\in T^n_{\tilde{Y}_1|{X}'Y'}(\vecx',\vecy')$,
\begin{align}\label{func:th}
\tilde{h}^{\tilde{\vecy}}_i(\vecy_1, \vecy_2, \ldots, \vecy_i) \defineqq \begin{cases}
    1, & \text{if } \vecy_i\in T^n_{\tilde{Y}_2|\tilde{Y}_1X'Y'}(\tilde{\vecy}, \vecx',\vecy') \\
    0, & \text{otherwise.}
   \end{cases}
\end{align}
Define events $\tilde{\cE}$ and $ \tilde{\cE}_1$ as 
\begin{align*}
\tilde{\cE} = &\Big\{\left|\inb{(\tilde{m}_{\two 1}, \tilde{m}_{\two 2}):\tilde{m}_{\two 1}\neq \tilde{m}_{\two 2},\,(\vecx', \vecY_{\tilde{m}_{\two 1}}, \vecY_{\tilde{m}_{\two 2}}, \vecy')\in T^n_{X'\tilde{Y}_1\tilde{Y}_2Y'}}\right|\\
&\qquad> \exp\inb{n\inp{|R_{\two}- I(\tilde{Y}_2;X'Y'\tilde{Y}_1)|^{+}+|R_{\two}-I(\tilde{Y}_1;X'Y')|^{+}+\epsilon'}}\Big\},\\
\tilde{\cE}_1 &= \inb{\sum_{i = 1}^{N_{\two}}\tilde{g}_i(\vecY_1, \vecY_2, \ldots, \vecY_i)>\exp{\inb{n\inp{|R_{\two}-I(\tilde{Y}_1;X'Y')|^{+}+\frac{\epsilon'}{2}}}}}.
\end{align*}
Let $R'_{\two}\defineqq \frac{\log{\inp{N_{\two}-1}}}{n} =\frac{\log{\inp{2^{nR_{\two}}-1}}}{n} $. For $i\in [1:N_{\two}]$ and $\tilde{\vecy}\in T^n_{\tilde{Y}_1|{X}'Y'}(\vecx',\vecy')$, define events $\tilde{\cE}_{2}^{i, \tilde{\vecy}}$ and $\tilde{\cE}_{2,\dagger}^{i, \tilde{\vecy}}$ as 
\begin{align*}
\tilde{\cE}_{2}^{i, \tilde{\vecy}} &= \inb{\sum_{j = 1, j\neq i}^{N_{\two}} \tilde{h}^{\tilde{\vecy}}_j(\vecY_1, \vecY_2, \ldots, \vecY_j)>\exp{\inb{n\inp{|R_{\two}-I(\tilde{Y}_2;\tilde{Y}_1X'Y')|^{+}+\frac{\epsilon'}{2}}}}}\\
\tilde{\cE}_{2,\dagger}^{i, \tilde{\vecy}} &= \inb{\sum_{j = 1, j\neq i}^{N_{\two}} \tilde{h}^{\tilde{\vecy}}_j(\vecY_1, \vecY_2, \ldots, \vecY_j)>\exp{\inb{n\inp{|R'_{\two}-I(\tilde{Y}_2;\tilde{Y}_1X'Y')|^{+}+\frac{\epsilon'}{2}}}}}
\end{align*}
Note that 
\begin{align*}
&\left|\inb{(\tilde{m}_{\two 1}, \tilde{m}_{\two 2}):\tilde{m}_{\two 1}\neq  \tilde{m}_{\two 2}\text{ and }(\vecx', \vecY_{\tilde{m}_{\two 1}}, \vecY_{\tilde{m}_{\two 2}}, \vecy')\in T^n_{X'\tilde{Y}_1\tilde{Y}_2Y'}}\right|\\
&\qquad \qquad \qquad= \sum_{i = 1}^{N_{\two}}\tilde{g}_i(\vecY_1, \vecY_2, \ldots, \vecY_i)\inp{\sum_{j = 1, j\neq i}^{N_{\two}} \tilde{h}^{\vecY_i}_j(\vecY_1, \vecY_2, \ldots, \vecY_j)}.
\end{align*}
Since,
\begin{align*}
\bbP\inp{\tilde{\cE_2}^{i, \vecY_i}} &= \sum_{\tilde{y}\in \cY^n}\bbP\inp{\vecY_i = \tilde{y}}\bbP\inp{\tilde{\cE_2}^{i, \vecY_i}|\vecY_i = \tilde{y}}\\
& = \sum_{\tilde{y}\in \cY^n}\bbP\inp{\vecY_i = \tilde{y}}\bbP\inp{\tilde{\cE_2}^{i, \tilde{y}}},
\end{align*} and $\tilde{\cE}_{2}^{i, \tilde{\vecy}}\subseteq \tilde{\cE}_{2,\dagger}^{i, \tilde{\vecy}}$ for all $i\in [1:N_{\two}]$,
\begin{align*}
\tilde{\cE} &\subseteq \inp{\cup_{i\in 2^{nR_{\two}}}\cup_{\tilde{\vecy}\in T_{\tilde{Y}_1|X'Y'}(\vecx', \vecy')}\tilde{\cE}_2^{i,\tilde{\vecy}}}\cup \tilde{\cE_1}\\
&\subseteq \inp{\cup_{i\in 2^{nR_{\two}}}\cup_{\tilde{\vecy}\in T_{\tilde{Y}_1|X'Y'}(\vecx', \vecy')}\tilde{\cE}_{2,R'_{\two}}^{i,\tilde{\vecy}}}\cup \tilde{\cE_1}.
\end{align*}

We apply Lemma~\ref{lemma:A1} to \eqref{func:tg} with $(\vecY_1, \ldots, \vecY_{N_{\two}})$ as the random variables $(Z_1, \ldots, Z_N)$. We can show that $a= (n+1)^{|\cY|}\exp{\inp{-nI(\tilde{Y}_1;X'Y')}}$ satisfies \eqref{eq:A1}. We take $t =\frac{1}{N_{\two}}\exp{\inb{n\inp{|R_{\two}-I(\tilde{Y}_1;X'Y')|^{+}+\frac{\epsilon'}{2}}}}$ and $n\geq n_1(\epsilon')$ (recall that $n_1(\epsilon') = \min{\inb{n:(n+1)^{|\cY|}\log{e}<\frac{1}{2}\exp(\frac{n\epsilon'}{2})}}$). This gives $N_{\two}(t-a\log{e})\geq (1/2)\exp(\frac{n \epsilon'}{2})$ which, when plugged in \eqref{eq:A2}, gives 
\begin{align}\label{eq:te1}
\bbP\inp{\tilde{\cE}_1} \leq \exp\inb{-\frac{1}{2}\exp\inb{\frac{n\epsilon'}{2}}}.
\end{align}

Similarly, for $i\in [1:N_{\two}]$, we can apply Lemma~\ref{lemma:A1} to \eqref{func:th} with $(\vecY_1, \ldots,\vecY_{i-1}, \vecY_{i+1}, \vecY_{N_{\two}})$ as the random variables $(Z_1, \ldots, Z_N)$. We can show that $a= (n+1)^{|\cY|}\exp{\inp{-nI(\tilde{Y}_2;\tilde{Y}_1X'Y')}}$ satisfies \eqref{eq:A1}. Choose, $t =\frac{1}{N_{\two}-1}\exp{\inb{n\inp{|R'_{\two}-I(\tilde{Y}_2;\tilde{Y}_1X'Y')|^{+}+\frac{\epsilon'}{2}}}}$ and $n\geq n_1(\epsilon')$ to obtain

\begin{align}\label{eq:te2}
\bbP\inp{\tilde{\cE}_{2,\dagger}^{i, \tilde{\vecy}}} \leq \exp\inb{-\frac{1}{2}\exp\inb{\frac{n\epsilon'}{2}}}, \, \tilde{y}\in T_{\tilde{Y}_1|X'Y'}(\vecx', \vecy').
\end{align}
Using \eqref{eq:te1} and \eqref{eq:te2}, we see that
\begin{align*}
\bbP\inp{\tilde{\cE}}\leq \inp{2^{nR_{\two}}|T_{\tilde{Y}|X'Y'}(\vecx', \vecy')|+1}\exp\inb{-\frac{1}{2}\exp\inb{\frac{n\epsilon'}{2}}}.
\end{align*}
When $\tilde{m}_{\two 1}= \tilde{m}_{\two 2}$ and $\tilde{Y}_1\neq \tilde{Y}_2$,
\begin{align*}
\left|\inb{(\tilde{m}_{\two 1}, \tilde{m}_{\two 2}):(\vecx', \vecY_{\tilde{m}_{\two 1}}, \vecY_{\tilde{m}_{\two 2}}, \vecy')\in T^n_{X'\tilde{Y}_1\tilde{Y}_2Y'}}\right| = 0 \text{ w.p. }1.
\end{align*} When $\tilde{m}_{\two 1}= \tilde{m}_{\two 2}$ and $\tilde{Y}_1= \tilde{Y}_2$,
\begin{align*}
&\bbP\inb{\left|\inb{(\tilde{m}_{\two 1}, \tilde{m}_{\two 2}):(\vecx', \vecY_{\tilde{m}_{\two 1}}, \vecY_{\tilde{m}_{\two 2}}, \vecy')\in T^n_{X'\tilde{Y}_1\tilde{Y}_2Y'}}\right|> \exp\inb{n\inp{|R_{\two}- I(\tilde{Y}_2;X'Y'\tilde{Y}_1)|^{+}+|R_{\two}-I(\tilde{Y}_1;X'Y')|^{+}+\epsilon'}}}\\
&=\bbP\inb{\left|\inb{\tilde{m}_{\two 1}:(\vecx',  \vecY_{\tilde{m}_{\two 1}}, \vecy')\in T^n_{X'\tilde{Y}_1Y'}}\right|> \exp\inp{n\inp{|R_{\two}-I(\tilde{Y}_1;X'Y')|^{+}+\epsilon'}}}\\
&\leq \exp\inb{-1/2\exp(n\epsilon')}.
\end{align*}
The equality follows from the condition that $R_{\two}\leq H(\tilde{Y}_2)$ and the inequality follows from \cite{CsiszarN88}[(A7)]. 
Thus, 
\begin{align}
&\bbP\inb{\left|\inb{(\tilde{m}_{\two 1}, \tilde{m}_{\two 2}):(\vecx', \vecy_{\tilde{m}_{\two 1}}, \vecy_{\tilde{m}_{\two 2}}, \vecy')\in T^n_{X'\tilde{Y}_1\tilde{Y}_2Y'} }\right|
> \exp\inb{n\inp{|R_{\two}-I(\tilde{Y}_1;X'Y')|^{+}+|R_{\two}-I(\tilde{Y}_2;\tilde{Y}_1 X'Y')|^{+}+\epsilon'}}}\nonumber\\
&\leq \inp{2^{nR_{\two}}|T_{\tilde{Y}|X'Y'}(\vecx', \vecy')|+1}\exp\inb{-\frac{1}{2}\exp\inb{\frac{n\epsilon'}{2}}} + \exp\inb{-1/2\exp(n\epsilon')}\nonumber\\
&\leq \inp{2^{nR_{\two}}|T_{\tilde{Y}|X'Y'}(\vecx', \vecy')|+1}\exp\inb{-\frac{1}{2}\exp\inb{\frac{n\epsilon}{4}}} + \exp\inb{-\frac{1}{2}\exp\inb{\frac{n\epsilon}{2}}}.\label{eq:prob2}
\end{align}

This completes the analysis of \eqref{codebook:5}\\

\noindent\underline{\em Analysis of \eqref{codebook:4}}\\
\iffalse
When $\tilde{Y}_1=\tilde{Y}_{2}$, \eqref{codebook:4} gives (for $R_{\two}\leq H(\tilde{Y}_2)$),
\begin{align*}
&\frac{\left|\inb{\mo:(\vecx_{\mo}, \vecy_{\tilde{m}_{\two 1}}, \vecy_{\tilde{m}_{\two 2}}, \vecy')\in T^n_{X'\tilde{Y}_1\tilde{Y}_2Y'} \text{ for some }\tilde{m}_{\two 1},\tilde{m}_{\two 2}, \tilde{m}_{\two 1}\neq \tilde{m}_{\two 2} } \right|}{N_{\one}} \leq \exp\inb{-n\epsilon/2} \nonumber\\
& \qquad \text{ if }I(X';\tilde{Y}_1\tilde{Y}_2Y')-|R_{\two}-I(\tilde{Y}_1;Y')|^{+} >\epsilon\\
\end{align*}
This is similar to (3.3) in \cite{CsiszarN88}[Lemma~3]. %When $\tilde{Y}_1\neq\tilde{Y}_{2}$, then we do the following analysis.
\fi
Let $A$ be the set of indices $(j,k)\in [1:N_{\two}]\times[1:N_{\two}]$  such that $(\vecy_j, \vecy_k)\in T^{n}_{\tilde{Y}_{1}\tilde{Y}_2|Y'}(\vecy')$ provided $|A|\leq \exp{\inb{n\inp{|R_{\two}- I(\tilde{Y}_2;\tilde{Y}_1Y')|^{+}+|R_{\two}-I(\tilde{Y}_1;Y')|^{+}} +\frac{\epsilon}{4}}}$. Otherwise, $A = \emptyset$.
Let 
\begin{align*}
\tilde{f}_i^{[\vecy_1,\vecy_2,\ldots,\vecy_{N_{\two}}]}\inp{\vecx_1,\vecx_2,\ldots, \vecx_i} = \begin{cases}
    1, & \text{if } \vecx_i\in \cup_{(j,k)\in A} T^n_{X'|\tilde{Y}_1\tilde{Y}_2Y'}(\vecy_j,\vecy_k,\vecy') \\
    0, & \text{otherwise.}
   \end{cases}
\end{align*}
\begin{align}
\bbP&\inb{\sum_{i=1}^{N_\one}\tilde{f}^{[\vecY_1,\vecY_2,\ldots,\vecY_{N_{\two}}]}_{i}\inp{\vecX_1,\vecX_2,\ldots, \vecX_i}\neq \left|\inb{i:\vecX_i\in T^n_{X|\tilde{Y}_1\tilde{Y}_2Y'}(\vecY_j,\vecY_k,\vecy') \text{ for some }j\neq k \right|}}\nonumber\\
=&\bbP\Bigg\{\left|\inb{(\tilde{m}_{\two1}, \tilde{m}_{\two2}):(\vecY_{\tilde{m}_{\two 1}}, \vecY_{\tilde{m}_{\two 2}}, \vecy')\in T^n_{\tilde{Y}_1\tilde{Y}_2Y'}}\right|\\
&\qquad \qquad \qquad\qquad> \exp{\inb{n\inp{|R_{\two}- I(\tilde{Y}_1;\tilde{Y}_2Y')|^{+}+|R_{\two}-I(\tilde{Y}_1;Y')|^{+}} +\frac{\epsilon}{4}}}\Bigg\}\nonumber\\
\leq& \inp{2^{nR_{\two}}|T_{\tilde{Y}|Y'}(\vecy')|+1}\exp\inb{-\frac{1}{2}\exp\inb{\frac{n\epsilon}{8}}}+ \exp\inb{-\frac{1}{2}\exp\inb{{n\epsilon}}}.\label{eq:conditioning2}
\end{align}

where last inequality follows from \eqref{eq:prob2} by replacing $(\vecx', \vecy')$ with $\vecy'$, $(\vecX',\vecY')$ with $\vecY'$ and $\frac{\epsilon}{2}$ (or $\epsilon'$) with $\frac{\epsilon}{4}$.

For $\vecy_i\in T^n_{\one},\, i = 1, \ldots, \vecy_{N_{\two}}$, we will apply Lemma~\ref{lemma:A1} on $\tilde{f}_i^{[\vecy_1,\vecy_2,\ldots,\vecy_{N_{\two}}]}$ with $(\vecX_1,\ldots,\vecX_{N_{\one}})$ as the random variables $(Z_1, \ldots, Z_N)$. We will first compute the value of $a$ in \eqref{eq:A1}. 

\begin{align*}
&E\insq{\tilde{f}^{[\vecy_1,\vecy_2,\ldots,\vecy_{N_{\two}}]}_{i}\inp{\vecX_1,\vecX_2,\ldots, \vecX_i}\Big|(\vecX_1,\vecX_2,\ldots, \vecX_{i-1})}\\
&=\bbP\inp{\vecX_i\in \cup_{(j,k)\in A} T^n_{X'|\tilde{Y}_1\tilde{Y}_2Y'}(\vecy_j,\vecy_k,\vecy')}\\
&\stackrel{\text{(a)}}{\leq} |A|\frac{\exp\inb{nH(X'|\tilde{Y}_1\tilde{Y}_2Y')}}{(n+1)^{-|\cX|}\exp(nH(X'))}\\
%&=|A_{i}|\frac{\exp\inb{nH(X|\tilde{X}\tilde{Y}Y)}}{(n+1)^{-|\cX|}\exp(nH(X))}\\
%&\leq (n+1)^{|\cX|}\exp\Big\{n\inp{|R_{\one}- I(\tilde{X};\tilde{Y}Y)|^{+}+|R_{\two}-I(\tilde{Y};Y)|^{+}}\\
%&\qquad \qquad \qquad -I(X;\tilde{X}\tilde{Y}Y)+\frac{\epsilon}{4}\Big\}\\
& \leq  (n+1)^{|\cX|}\exp\inb{n\inp{|R_{\two}- I(\tilde{Y}_2;\tilde{Y}_1Y')|^{+}+|R_{\two}-I(\tilde{Y}_1;Y')|^{+}} -I(X';\tilde{Y}_1\tilde{Y}_2Y')+\frac{\epsilon}{4}}.
\end{align*}
where (a) follows by union bound over $(j,k)\in A$ and by noting that $|T^n_{\one}| = |T^n_{X'}|$. 
If $I(X';\tilde{Y}_1\tilde{Y}_2Y') >|R_{\two}- I(\tilde{Y}_2;\tilde{Y}_1Y')|^{+}+|R_{\two}-I(\tilde{Y}_1;Y')|^{+} +\epsilon$ (which~\eqref{codebook:4} postulates), \eqref{eq:A1} holds with $a = (n+1)^{|\cX|}\exp\inb{-\frac{3}{4}n\epsilon}$. For $t = \exp\inb{\frac{-n\epsilon}{2}}$ and $n\geq n_2(\epsilon)$ (recall that $n_2(\epsilon) = \min{\inb{n:(n+1)^{|\cX|}\log e<\frac{1}{2}\exp{\inb{\frac{n\epsilon}{4}}}}},$ 
we get
\begin{align}
\bbP&\inb{\frac{1}{N_{\one}}\sum_{i=1}^{N_{\one}}\tilde{f}_i^{[\vecy_1,\vecy_2,\ldots,\vecy_{N_{\two}}]}\inp{\vecX_1,\vecX_2,\ldots, \vecX_i}>\exp\inb{\frac{-n\epsilon}{2}}}\nonumber \\
&\leq \exp\inb{-\frac{N_{\one}}{2}\exp\inb{-\frac{n\epsilon}{2}}}\nonumber\\
&\leq \exp\inb{-\frac{1}{2}\exp\inb{\frac{n\epsilon}{2}}} \nonumber
\end{align}
where the last inequality uses the assumption that $N_\one\geq \exp\inb{n\epsilon}$. Averaging over $(\vecY_1, \ldots, \vecY_{\two})$, we get
\begin{align}
\bbP&\inb{\frac{1}{N_{\one}}\sum_{i=1}^{N_{\one}}\tilde{f}_i^{[\vecY_1,\vecY_2,\ldots,\vecY_{N_{\two}}]}\inp{\vecX_1,\vecX_2,\ldots, \vecX_i}>\exp\inb{\frac{-n\epsilon}{2}}}\nonumber \\
&\leq \exp\inb{-\frac{1}{2}\exp\inb{\frac{n\epsilon}{2}}} \label{eq:tprob}
\end{align}
Let events $\tilde{\cF_1}$ and $\tilde{\cF_2}$ be defined as 
\begin{align*}
\tilde{\cF_1} &= \inb{\frac{1}{N_{\one}}\left|\inb{i:\vecX_i\in T^n_{X|\tilde{Y}_1\tilde{Y}_2Y'}(\vecY_j, \vecY_k, \vecy')\text{ for some}, j}\right|>\exp\inb{\frac{-n\epsilon}{2}}},\\
\tilde{\cF_2} &= \inb{\sum_{i=1}^{N_\one}\tilde{f}^{[\vecY_1,\vecY_2,\ldots,\vecY_{N_{\two}}]}_{i}\inp{\vecX_1,\vecX_2,\ldots, \vecX_i}\neq \left|\inb{i:\vecX_i\in T^n_{X|\tilde{Y}_1\tilde{Y}_2Y'}(\vecY_j,\vecY_k,\vecy') \text{ for some }j}\right|},\\
\tilde{\cF_3} &= \inb{\sum_{i=1}^{N_\one}\tilde{f}^{[\vecY_1,\vecY_2,\ldots,\vecY_{N_{\two}}]}_{i}\inp{\vecX_1,\vecX_2,\ldots, \vecX_i}>\exp\inb{\frac{-n\epsilon}{2}}}.
\end{align*}
We are interested in $\bbP\inp{\tilde{\cF_1}}$. We see that 
\begin{align*}
\bbP\inp{\tilde{\cF}_1}& = \bbP\inp{\tilde{\cF}_1\cap \tilde{\cF}_2} +\bbP\inp{\tilde{\cF}_1\cap \tilde{\cF}^{c}_2}\\
& \leq \bbP\inp{\tilde{\cF}_2} +\bbP\inp{\tilde{\cF}_1\cap \tilde{\cF}^{c}_2}\\
& \leq \bbP\inp{\tilde{\cF}_2} +\bbP\inp{\tilde{\cF}_3}\\
&\stackrel{\text{(a)}}{\leq} \inp{2^{nR_{\two}}|T_{\tilde{Y}|Y'}(\vecy')|+1}\exp\inb{-\frac{1}{2}\exp\inb{\frac{n\epsilon}{8}}}+ \exp\inb{-\frac{1}{2}\exp\inb{{n\epsilon}}}+\exp\inb{-\frac{1}{2}\exp\inb{\frac{n\epsilon}{2}}}\\
& =\inp{2^{nR_{\two}}|T_{\tilde{Y}|Y'}(\vecy')|+3}\exp\inb{-\frac{1}{2}\exp\inb{\frac{n\epsilon}{8}}},
\end{align*}
where (a) follows from \eqref{eq:conditioning2} and \eqref{eq:tprob}.        
\end{proof}

%!TeX root=main.tex

\section{\em Proof of Theorem~\ref{thm:inner_bd}}\label{sec:inner_bd_proof}
\begin{proof}
\noindent {\em Encoding.} 
For some $P_{\one}$ and $P_{\two}$ satisfying $\min_{x\in \cX}P_{\one}(x)>0$ and $\min_{y\in \cY}P_{\two}(y)>0$ respectively, and $\epsilon>0$ (TBD), consider a codebook of rate $(R_{\one}, R_{\two})$ (TBD) as given by Lemma~\ref{lemma:codebook}.   For $\mo \in \cM_{\one}$, $f_{\one}(\mo) = \vecx_{\mo}$ and for $\mt \in \cM_{\two}$, $f_{\two}(\mt) = \vecy_{\mt}$.\\
{\em Decoding.} 
For a parameter $\eta>0$, let $\cD_{\eta}$ be the set of joint distributions defined as
$D_\eta \defineqq  \inb{P_{XYZ}\in \cP^n_{\cX\times \cY \times \cZ}:\, D\inp{P_{XYZ}||P_XP_YW}\leq\eta }$.
 Decoding happens in five steps. In the first step, we populate sets $A_1$ and $B_1$ containing candidate messages for user \one and \two respectively. In steps $2-5$, we sequentially remove the candidates. \\
\begin{description}
\item  [{\em Step 1}:]\label{step:1} Let $A_1 = \{m_{\one}\in\mathcal{M}_{\one}:\inp{f_{\one}(m_{\one}), \vecy, \vecz} \in T^{n}_{XYZ}$ for some $\vecy\in \cY^n$ such that $P_{XYZ}\in \cD_{\eta}$\} and\\
$B_1 = \{m_{\two}\in\mathcal{M}_{\two}:\inp{\vecx, f_{\two}(m_{\two}), \vecz} \in T^{n}_{XYZ}$ for some $\vecx\in \cX^n$ such that $P_{XYZ}\in \cD_{\eta}$\}.

\item [{\em Step 2}:] \label{step:2} Let $C_1 = \{\mo\in A_1:$ {For every } $\tilde{m}_{\two 1},\, \tilde{m}_{\two 2}\in B_1$,  such that for every $\vecy\in \cY^n$ with $\inp{f_{\one}(\mo), \vecy,f_{\two}(\tilde{m}_{\two 1}), f_{\two}(\tilde{m}_{\two 2}),\vecz}\in T^n_{XY\tilde{Y}_1\tilde{Y}_2Z}$ and $P_{XYZ}\in\cD_{\eta}$, $I(\tilde{Y}_1\tilde{Y}_2;XZ|Y)> \eta\}$. Let $A_2 = A_1\setminus C_1$.

\item [{\em Step 3}:]\label{step:3} Let $C_2 = \{\mt\in B_1:$ {For every } $\tilde{m}_{\one 1},\, \tilde{m}_{\one 2}\in A_2$,  such that for every $\vecx\in \cX^n$ with $\inp{\vecx, f_{\two}(\mt),  f_{\one}(\tilde{m}_{\one 1}),f_{\one}(\tilde{m}_{\one 2}),\vecz}\in T^n_{XY\tilde{X}_1\tilde{X}_2Z}$ and $P_{XYZ}\in \cD_{\eta}$ and $I(\tilde{X}_1\tilde{X}_2;YZ|X)> \eta\}$. Let $B_2 = B_1\setminus C_2$.

\item [{\em Step 4}:]\label{step:4} Let $C_3 = \{\mo\in A_2:$ {For every } $(\tilde{m}_{\one},\, \tilde{m_{\two}})\in A_2\times B_2, \, \tilde{m}_{\one}\neq m_{\one}$ such that for every $\vecy\in \cY^n$ with $\inp{f_{\one}(\mo), \vecy, f_{\one}(\tilde{m}_{\one}), f_{\two}(\tilde{m}_{\two}), \vecz}\in T^n_{XY\tilde{X}\tilde{Y}Z}$ and $P_{XYZ}\in \cD_{\eta}$ , $I(\tilde{X}\tilde{Y};XZ|Y)>\eta\}$. Let $A_3 = A_2\setminus C_3$.

\item [{\em Step 5}:]\label{step:5} Let $C_4 = \{\mt\in B_2:$ {For every } $(\tilde{m}_{\one},\, \tilde{m_{\two}})\in A_3\times B_2, \, \tilde{m}_{\two}\neq m_{\two}$ such that for every $\vecx\in \cX^n$ with $\inp{\vecx,f_{\two}(\mt), f_{\one}(\tilde{m}_{\one}), f_{\two}(\tilde{m}_{\two}), \vecz}\in T^n_{XY\tilde{X}\tilde{Y}Z}$ and $P_{XYZ}\in \cD_{\eta}$,\,$I(\tilde{X}\tilde{Y};YZ|X)>\eta\}$. Let $B_3 = B_2\setminus C_4$.

\end{description}
After steps 1-5, the decoded output is as follows.
\begin{align*}
\phi(\vecz) = \begin{cases}(\mo,\mt) &\text{ if }A_3\times B_3 = \{(\mo,\mt)\},\\ \oneb &\text{ if }|A_3| = 0, \, |B_3| \neq 0,\\\twob &\text{ if }|A_3| \neq 0, \, |B_3| = 0\text{ and}\\(1,1)&\text{ otherwise.}
																\end{cases}
\end{align*}

For small enough choice of $\eta>0$, Lemma~\ref{lemma:disambiguity} implies that if $|A_3|, \, |B_3| \geq 1$, then $|A_3|$ = $|B_3|$ = 1. 
Suppose the channel is non-\spoofable.
We start by showing that $P_{e,\na}$ can be upper bounded by sum of $P_{e,\malone}$ and $P_{e,\maltwo}$. So, we only need to analyse the case when a user is malicious. To show this, we note that $\cE_{\mo, \mt}  = \cE_{\mo}\cup\cE_{\mt}$. Thus,
\begin{align*}
&P_{e,\na}\hspace{-0.25em} =
\frac{1}{N_{\one}\cdot N_{\two}} \sum_{(\mo, \mt)\in \mathcal{M}_{\one}\times\mathcal{M}_{\two}}W^n\inp{\cE_{\mo}\cup\cE_{\mt}|f_{\one}(\mo), f_{\two}(\mt)}\\
&\leq\frac{1}{N_{\one}\cdot N_{\two}} \sum_{(\mo, \mt)\in \mathcal{M}_{\one}\times\mathcal{M}_{\two}}\Big(W^n\inp{\cE_{\mo}|f_{\one}(\mo), f_{\two}(\mt)}+W^n\inp{\cE_{\mt}|f_{\one}(\mo), f_{\two}(\mt)}\Big)\\
&=\frac{1}{ N_{\two}}\sum_{\mt\in\cM_{\two}}\inp{\frac{1}{N_{\one}} \sum_{\mo\in \mathcal{M}_{\one}}W^n\inp{\cE_{\mo}|f_{\one}(\mo), f_{\two}(\mt)}}\\
&\qquad \qquad+\frac{1}{N_{\one}}\sum_{\mo\in\cM_{\one}}\inp{ \frac{1}{N_{\two}}\sum_{\mt\in\mathcal{M}_{\two}}W^n\inp{\cE_{\mt}|f_{\one}(\mo), f_{\two}(\mt)}}\\
&\leq P_{e,\malone} +P_{e,\maltwo}.
\end{align*}
So, if $P_{e,\malone}$ and $P_{e,\maltwo}$ are small, $P_{e,\na}$ is also small. Thus, it is sufficient to analyze the cases when one of the user is adversarial.

We consider the case when user \two is malicious while user \one is honest.
Let $\cE$ be defined as
\begin{align*}
\cE =\{\vecz:\phi(\vecz)\in\inb{\cM_{\one}\setminus\{\mo\}\times\cM_{\two}, \oneb, (1,1)}\}.
\end{align*}
Then, the probability of error is
\begin{align*}
P_{e,\maltwo} = \max_{\vecy\in \cY^n}{\frac{1}{N_{\one}}\sum_{\mo\in \cM_{\one}}W^n(\cE|f_{\one}^{n}(\mo),\vecy)}
\end{align*}
For each $\vecy' \in \cY^n$, we will get a uniform upper bound on $P_{e, \maltwo}$ which goes to zero with $n$.
So, let us fix an attack vector $\vecy\in \cY^n$ and analyze
\begin{align*}
P:= {\frac{1}{N_{\one}}\sum_{\mo\in \cM_{\one}}W^n(\cE|f_{\one}^{n}(\mo),\vecy)}.
\end{align*}

For some $\epsilon$ satisfying $0<\epsilon<\eta/3$, let
\begin{align*}
\cH&= \inb{m_{\one}: (\vecx_{m_{\one}}, \vecy)\in \cup_{P_{XY}\in \cP^n_{\cX\times\cY}}T_{XY}^n, I(X;Y)> \epsilon}.
\end{align*}

Then,
\begin{align*}
P &\leq \frac{1}{N_{\one}}|\cH| + \sum_{\mo\in \cH^c}W^n(\cE|f_{\one}^{n}(\mo),\vecy)\\
&:= P_1 +P_2.
\end{align*}
The first term on the RHS, 
\begin{align*}
P_1\leq |\cP^n_{\cX\times \cY}|\times\frac{|\inb{m_{\one}: (\vecx_{m_{\one}}, \vecy)\in T_{XY}^n, \, I(X;Y)> \epsilon}|}{N_{\one}}
\end{align*}
which goes to zero as $n\rightarrow \infty$ by using~\eqref{codebook:1} and noting that there are only polynomially many types.

Using the decoder definition and Lemma~\ref{lemma:disambiguity}, we note that $\cE\subseteq\{\vecz: \mo\notin A_3\}$.
Thus, $W^n(\cE|f_{\one}^{n}(\mo),\vecy) \leq W^n(\{\vecz: \mo\notin A_3\}|f_{\one}^{n}(\mo),\vecy)$. 
\iffalse
The event $\{\vecz: \mo\notin A_3\} = \cE_1\cup\cE_2\cup\cE_3$ where
\begin{align*}
\cE_1 &= \{\vecz: \mo\notin A_1\},\\
\cE_2 &= \{\vecz: \mo\in A_1\cap A_2^c\}, \text{and}\\
\cE_3 &= \{\vecz:\mo \in A_1\cap A_2 \cap A_3^c\}.
\end{align*} 
\fi

For $\vecy'\in \cY^n$, let $\cE_1(\vecy')$ be defined as
\begin{align*}
\cE_1(\vecy') = \{\vecz: (\vecx_{\mo}, \vecy', \vecz)\in T^n_{XYZ} \text{ such that }P_{XYZ}\in D_{\eta}\}
\end{align*}
Then $\inp{\cup_{\tilde{\vecy}\in \cY^n}\cE_1(\tilde{\vecy})}^c=\{\vecz: \mo\notin A_1\}$. 
Note that $\inp{\cup_{\tilde{\vecy}\in \cY^n}\cE_1(\tilde{\vecy})}^c \subseteq \cE_1(\vecy)^c$. Then,
\begin{align*}
P_2\leq &  {\frac{1}{N_{\one}}\sum_{\mo\in \cH^c}W^n(\cE|f_{\one}^{n}(\mo),\vecy)}\\
=&\frac{1}{N_{\one}}\sum_{\mo\in \cH^c}W^n(\inp{\cE_1(\vecy)^c\cap\cE}\cup\inp{\cE_1(\vecy)\cap\cE}|f_{\one}^{n}(\mo),\vecy) \\
\leq &\frac{1}{N_{\one}}\sum_{\mo\in \cH^c}W^n(\inp{\cE_1(\vecy)^c}|f_{\one}^{n}(\mo),\vecy) + W^n(\inp{\cE_1(\vecy)\cap\cE}|f_{\one}^{n}(\mo),\vecy) \\
=&\frac{1}{N_{\one}}\sum_{m_{\one}\in \cH^c}\inp{\sum_{P_{XYZ}\in \cD_{\eta}^c}\sum_{\vecz\in T^n_{Z|XY}(\vecx_{m_{\one}},\vecy)}W^n(\vecz|\vecx_{m_{\one}}, \vecy)} \\ +& \frac{1}{N_{\one}}\sum_{m_{\one}\in \cH^c}\inp{\sum_{P_{XYZ}\in \cD_{\eta}}\sum_{\vecz\in T^n_{Z|XY}(\vecx_{m_{\one}},\vecy)\cap\cE}W^n(\vecz|\vecx_{m_{\one}}, \vecy)}\\
=:&P_{2a} +P_{2b}
\end{align*}

For any $\mo \in \cH^c$,

\begin{align*}
\sum_{P_{XYZ}\in \cD_{\eta}^c}\sum_{\vecz\in T^n_{Z|XY}(\vecx_{m_{\one}},\vecy)}W^n(\vecz|\vecx_{m_{\one}}, \vecy) &\leq |\cD_{\eta}^c|\exp{\inp{-nD(P_{XYZ}||P_{XY}W)}}\\
& = |\cD_{\eta}^c|\exp{\inp{-n\inp{D(P_{XYZ}||P_{X}P_{Y}W) -I(X;Y)}}}\\
& \leq |\cD_{\eta}^c|\exp{\inp{-n\inp{\eta-\epsilon}}} 
\end{align*}
Thus, 
\begin{align*}
P_{2a}&\leq \frac{|\cH^c|}{N_{\one}}|\cD_{\eta}^c|\exp{\inp{-n\inp{\eta-\epsilon}}}\\
&\rightarrow 0 \text{ as }\epsilon<\eta/3 \text{ and $|\cD_{\eta}^c|$ grows as a polynomial in $n$.}
\end{align*}

We are left to analyze 
\begin{align*}
P_{2b} = \frac{1}{N_{\one}}\sum_{m_{\one}\in \cH^c}\inp{\sum_{P_{XYZ}\in \cD_{\eta}}\sum_{\vecz\in T^n_{Z|XY}(\vecx_{m_{\one}},\vecy)\cap\cE}W^n(\vecz|\vecx_{m_{\one}}, \vecy)}.
\end{align*} 
Let

%%%%%%%%%%%%%%%%%%%%%%%%%%%%%%%%%%%%%%%%%%%%%%%%%%%%%%%%%%%%%%%%%%%%%%%%%%%%%%%%%%%%%%%%%%%
\iffalse
For $(\vecx_{\mo}, \vecy, \vecz)\in P_{XYZ}$ such that $P_{XYZ}\in \cD_{\eta}$ and $\mo\in \cH^c$, $\phi_{\one}(\vecz)\notin\inb{m_{\one}, \two}$ when one of the following happens (follows from Lemma~\ref{lemma:disambiguity}).
\begin{itemize}
	\item $|A| = |B| = 1$, but $\mo\notin A$.
	\item $|A| = 0$.
\end{itemize}
To formalize this, we define the following sets. For $\mo\in \cM_{\one}$,
\begin{align*}
\cG_{\mo} &= \inb{\vecz: (\vecx_{\mo}, \vecy, \vecz)\in P_{XYZ}, P_{XYZ}\in \cD_{\eta}, I(X, Y)\leq \epsilon}\\
\cG_{\mo,0} &= \cG_{\mo} \cap \inb{\vecz: \phi_{\one}(\vecz)\notin\inb{m_{\one}, \two} }\\
\cG_{\mo,1} &= \cG_{\mo} \cap\inb{\vecz:  |A| = |B| = 1 , \mo\notin A }\\
\cG_{\mo,2} &= \cG_{\mo} \cap\inb{\vecz:  |A| = 0}\\
\cG_{\mo,3} &= \cG_{\mo} \cap\inb{\vecz: \mo\notin A }
\end{align*}
We are interested in $\cG_{\mo,0}$. Note that $\cG_{\mo,0} \subseteq \cG_{\mo,1}\cup \cG_{\mo,2}\subseteq \cG_{\mo,3}$. So, it suffices to upper bound the probability of $\cG_{\mo,3}$ when $\vecx_{\mo}$ is sent by user \one and $\vecy$ by user \two.

 From the definition of $A$, we see that $\cG_{\mo,3}$ is such that $\mo$ was included in $A$ in step 1 but got removed in either step 2 or 5. We capture this by defining the following sets of distributions:
\fi
%%%%%%%%%%%%%%%%%%%%%%%%%%%%%%%%%%%%%%%%%%%%%%%%%%%%%%%%%%%%%%%%%%%%%%%%%%%%%%%%%%%%
\begin{align*}
\cP_1^{\eta}& =  \{P_{X\tilde{X}\tilde{Y}YZ}\in \cP^n_{\cX\times\cX\times\cY\times\cY\times\cZ}: P_{XYZ}\in \cD_{\eta},I(X;Y)\leq \epsilon,\, P_{\tilde{X}Y'Z}\in \cD_{\eta} \text{ for some }Y',\\
&\qquad  \, P_{X'\tilde{Y}Z}\in \cD_{\eta} \text{ for some }X', P_{X}=P_{\tilde{X}}=P_{\one}, P_{\tilde{Y}} = P_{\two}, \, I(\tilde{Y};X)\leq \eta, \, I(\tilde{Y};\tilde{X})\leq \eta\\
&\qquad \text{ and }I(\tilde{X}\tilde{Y};XZ|Y)\geq\eta\}\\
\cP_2^{\eta}& = \{P_{X\tilde{Y}_1\tilde{Y}_2YZ}\in \cP^n_{\cX\times\cY\times\cY\times\cY\times\cZ}: P_{XYZ}\in \cD_{\eta},I(X;Y)\leq \epsilon,\, P_{X'_1\tilde{Y}_1Z}\in \cD_{\eta} \text{ for some }X'_1,\\
&\qquad \, P_{X'_2\tilde{Y}_2Z}\in \cD_{\eta} \text{ for some }X'_2, P_{X}=P_{\one}, P_{\tilde{Y}_1}=P_{\tilde{Y}_2} = P_{\two}\text{ and }I(\tilde{Y}_1\tilde{Y}_2;XZ|Y)\geq\eta\}.
\end{align*}
For $P_{X\tilde{X}\tilde{Y}YZ}\in \cP_1^{\eta}$ and $P_{X\tilde{Y}_1\tilde{Y}_2YZ}\in \cP_2^{\eta} $, let
\begin{align*}
\cE_{\mo,1}(P_{X\tilde{X}\tilde{Y}YZ}) & = \big\{\vecz: \exists(\tilde{m}_{\one},\, \tilde{m}_{\two})\in \cM_{\one}\times \cM_{\two}, \, \tilde{m}_{\one}\neq m_{\one}, \,  \inp{\vecx_{\mo},\vecx_{\tilde{m}_{\one}}, \vecy,   \vecy_{\tilde{m}_{\two}}, \vecz}\in T^n_{X\tilde{X}Y\tilde{Y}Z} \big\} \text{ and }\\
\cE_{\mo,2}(P_{X\tilde{Y}_1\tilde{Y}_2YZ}) & = \big\{\vecz: \exists \tilde{m}_{\two 1},\, \tilde{m}_{\two 2}\in \cM_{\two}, \,\inp{\vecx_{\mo},  \vecy_{\tilde{m}_{\two 1}}, \vecy_{\tilde{m}_{\two 2}},\vecy,\vecz}\in T^n_{X\tilde{Y}_1\tilde{Y}_2YZ}\big\}.
\end{align*}
Note that the extra conditions $I(\tilde{Y};X)\leq \eta$ and $ I(\tilde{Y};\tilde{X})\leq \eta$ in $\cP_{1}^{\eta}$ are due the Lemma~\ref{lemma:indep} (stated below) and using the decoder definition where we only consider $\tilde{m}_{\two}$ which are in $B_2$, that is, they have passed the check in {\em Step 3}.

\begin{lemma}\label{lemma:indep}
For a distribution $P_{XY\tilde{X}Y'X'\tilde{Y}Z}\in \cP^n_{\cX\times\cY\times\cX\times\cY\times\cX\times\cY\times\cZ}$ satisfying
\begin{enumerate}[label=(\Alph*)]
	\item $P_{XYZ}\in D_{\eta}$
	\item $P_{\tilde{X}Y'Z}\in D_{\eta}$
	\item $P_{X'\tilde{Y}Z}\in D_{\eta}$
	\item $I(X\tilde{Y};\tilde{X}Z|Y')<\eta$
\end{enumerate}  
The following holds:  $I(\tilde{Y}X;\tilde{X})\leq \eta$.
\end{lemma}
The proof of this Lemma follows from arguments in the proof of Lemma~\ref{lemma:disambiguity}. In particular, the claim follows from \eqref{disambeq:2}.
\begin{align*}
\cE_1(\vecy)\cap\cE&= \{\vecz:\mo\in A_1\cap A_2^c\}\cup\{\vecz: \mo\in A_2\cap A_3^c\}\\
&=\inp{\cup_{P_{X\tilde{Y}_1\tilde{Y}_2YZ}\in \cP_2^{\eta}}\cE_{\mo,2}(P_{X\tilde{Y}_1\tilde{Y}_2YZ})}\cup\inp{\cup_{P_{X\tilde{X}\tilde{Y}YZ}\in \cP_1^{\eta}}\cE_{\mo,1}(P_{X\tilde{X}\tilde{Y}YZ})} 
\end{align*}

Thus,
\begin{align}
P_{2b} =& \frac{1}{N_{\one}}\sum_{m_{\one}\in \cH^c}\inp{\sum_{P_{XYZ}\in \cD_{\eta}}\sum_{\vecz\in T^n_{Z|XY}(\vecx_{m_{\one}},\vecy)\cap\cE}W^n(\vecz|\vecx_{m_{\one}}, \vecy)}\nonumber\\
\leq&\frac{1}{N_{\one}}\sum_{\mo\in \cH^c}\sum_{P_{X\tilde{X}\tilde{Y}YZ}\in \cP_1^{\eta}}  W^n\inp{\cE_{\mo,1}(P_{X\tilde{X}\tilde{Y}YZ})|\vecx_{\mo}, \vecy} \nonumber\\
&\qquad \qquad+ \frac{1}{N_{\one}}\sum_{\mo\in \cH^c}\sum_{P_{X\tilde{Y}_1\tilde{Y}_2YZ}\in \cP_2^{\eta}}  W^n\inp{\cE_{\mo,2}(P_{X\tilde{Y}_1\tilde{Y}_2YZ})|\vecx_{\mo}, \vecy}. \label{err:upperbound}
\end{align}
We see that $|\cP_1^{\eta}|$ and $|\cP_2^{\eta}|$ are at most polynomial and clearly $|\cH^c|\leq N_{\one}$. So, it will  suffice to uniformly upper bound $W^n\inp{\cE_{\mo,1}(P_{X\tilde{X}Y\tilde{Y}Z})|\vecx_{\mo}, \vecy}$ and $W^n\inp{P_{X\tilde{Y}_1\tilde{Y}_2YZ})|\vecx_{\mo}, \vecy}$ by a term exponentially decreasing in $n$ for all $P_{X\tilde{X}Y\tilde{Y}Z}\in \cP_1^{\eta}$ and $P_{X\tilde{Y}_1\tilde{Y}_2YZ}\in \cP_2^{\eta} $. 
We start with the first term in the RHS of~\eqref{err:upperbound}. By using~\eqref{codebook:2b}, we see that for $P_{X\tilde{X}\tilde{Y}YZ}\in \cP_1^{\eta}$ such that
\begin{align*}
I\inp{X;\tilde{X}\tilde{Y}Y}> |R_{\one}- I(\tilde{X};\tilde{Y}Y)|^{+}+|R_{\two}-I(\tilde{Y};Y)|^{+}+\epsilon\,
\end{align*}
\begin{align*}
\frac{\left|\inb{\mo:(\vecx_{\mo}, \vecx_{\tilde{m}_{\one}}, \vecy_{\mt}, \vecy)\in T^n_{X\tilde{X}\tilde{Y}Y} \text{ for some }\tilde{m}_{\one}\neq \mo\text{ and some }\mt} \right|}{N_{\one}} \leq \exp\inb{-n\epsilon/2}.
\end{align*}
So, 
\begin{align*}
&\frac{1}{N_{\one}}\sum_{\mo\in \cH^c} W^n\inp{\cE_{\mo,1}(P_{X\tilde{X}Y\tilde{Y}Z})|\vecx_{\mo}, \vecy} \\
& = \frac{1}{N_{\one}}\sum_{\substack{\mo:(\vecx_{\mo}, \vecx_{\tilde{m}_{\one}}, \vecy_{\mt}, \vecy)\in T^n_{X\tilde{X}\tilde{Y}Y},\\ \tilde{m}_{\one}\in \cM_{\one},\tilde{m}_{\one}\neq \mo,\mt\in\cM_{\two}}} \sum_{\vecz\in T^{n}_{Z|X\tilde{X}Y\tilde{Y}}(\vecx_{\mo},\vecx_{\tilde{m}_{\one}},\vecy,\vecy_{\tilde{m}_{\two}})}W^n\inp{\vecz|\vecx_{\mo}, \vecy}\\
&\leq \exp\inb{-n\epsilon/2}.
\end{align*}
Thus, it is sufficient to consider distributions $P_{X\tilde{X}\tilde{Y}YZ}\in \cP_1^{\eta}$ for which 
\begin{align}
I\inp{X;\tilde{X}\tilde{Y}Y}\leq |R_{\one}- I(\tilde{X};\tilde{Y}Y)|^{+}+|R_{\two}-I(\tilde{Y};Y)|^{+}+\epsilon\label{eq:1}
\end{align}
For $P_{X\tilde{X}\tilde{Y}YZ}\in \cP_1^{\eta}$ satisfying~\eqref{eq:1},
\begin{align}
&\sum_{\vecz\in \cE_{\mo,1}(P_{X\tilde{X}Y\tilde{Y}Z})}W^n(\vecz|\vecx_{\mo}, \vecy)\nonumber\\
&\qquad\leq\sum_{\substack{\tilde{m}_{\one}, \tilde{m}_{\two}:\\(\vecx_{\mo}, \vecx_{\tilde{m}_{\one}}, \vecy_{\tilde{m}_{\two}},\vecy)\in T^{n}_{X\tilde{X}\tilde{Y}Y}}}\sum_{\vecz:(\vecx_{\mo}, \vecx_{\tilde{m}_{\one}}, \vecy_{\tilde{m}_{\two}},\vecy, \vecz)\in T^{n}_{X\tilde{X}\tilde{Y}YZ}}W^n(\vecz|\vecx_{\mo}, \vecy)\nonumber\\
&\qquad\leq \sum_{\substack{\tilde{m}_{\one}, \tilde{m}_{\two}:\\(\vecx_{\mo}, \vecx_{\tilde{m}_{\one}}, \vecy_{\tilde{m}_{\two}},\vecy)\in T^{n}_{X\tilde{X}\tilde{Y}Y}}}\frac{|T^{n}_{Z|X\tilde{X}\tilde{Y}Y}(\vecx_{\mo},\vecx_{\tilde{m}_{\one}}, \vecy_{\tilde{m}_{\two}}, \vecy)|}{|T^n_{Z|XY}(\vecx_{\mo},\vecy)|}\nonumber\\
&\qquad \leq \sum_{\substack{\tilde{m}_{\one}, \tilde{m}_{\two}:\\(\vecx_{\mo}, \vecx_{\tilde{m}_{\one}}, \vecy_{\tilde{m}_{\two}},\vecy)\in T^{n}_{X\tilde{X}\tilde{Y}Y}}}\frac{\exp\inp{nH(Z|X\tilde{X}\tilde{Y}Y)}}{(n+1)^{-|\cX||\cY||\cZ|}\exp\inp{nH(Z|XY)}}\nonumber\\
&\qquad \leq \sum_{\substack{\tilde{m}_{\one}, \tilde{m}_{\two}:\\(\vecx_{\mo}, \vecx_{\tilde{m}_{\one}}, \vecy_{\tilde{m}_{\two}},\vecy)\in T^{n}_{X\tilde{X}\tilde{Y}Y}}}\exp\inp{-n\inp{I(Z;\tilde{X}\tilde{Y}|XY)-\epsilon}} \text{ for large }n.\nonumber\\
&\qquad\stackrel{\text{(a)}}{\leq}\exp\inp{n\inp{|R_{\one}- I(\tilde{X};\tilde{Y}XY)|^{+}+|R_{\two}-I(\tilde{Y};XY)|^{+}-I(Z;\tilde{X}\tilde{Y}|XY)+2\epsilon}}\label{eq:upperbound2}
\end{align}
where (a) follows using~\eqref{codebook:3b}. We will separately consider the following cases which together cover all possibilities.
\begin{enumerate}
	\item $R_{\one}\leq I(\tilde{X};\tilde{Y}Y)$ and $R_{\two}\leq I(\tilde{Y};Y)$ \label{case1}
	\item $  I(\tilde{X};\tilde{Y}Y)<R_{\one}$ and $R_{\two}\leq I(\tilde{Y};XY)$\label{case2}
	\item $R_{\one}\leq  I(\tilde{X};\tilde{Y}XY) $ and $I(\tilde{Y};Y)<R_{\two}$\label{case3}
	\item $I(\tilde{X};\tilde{Y}XY)<R_{\one} $ and $I(\tilde{Y};XY)<R_{\two}$\label{case4}
\end{enumerate}

\noindent\underline{Case \ref{case1}: $R_{\one}\leq I(\tilde{X};\tilde{Y}Y)$ and $R_{\two}\leq I(\tilde{Y};Y)$}\\
In this case, \eqref{eq:1} implies that $I(X;\tilde{X}\tilde{Y}Y) \leq \epsilon$ Thus, using the condition $I(XZ;\tilde{X}\tilde{Y}|Y)\geq \eta$ from definition of $\cP_1^{\eta}$, we see that
\begin{align*}
I(Z;\tilde{X}\tilde{Y}|XY) &= I(XZ;\tilde{X}\tilde{Y}|Y)-I(X;\tilde{X}\tilde{Y}|Y)\\
&{\geq} \eta-\epsilon.
\end{align*} 
This implies that
\begin{align*}
\sum_{\vecz\in \cE_{\mo,1}(P_{X\tilde{X}Y\tilde{Y}Z})}W^n(\vecz|\vecx_{\mo}, \vecy)&\leq \exp\inp{-n\inp{\eta-3\epsilon}}\\
&\rightarrow 0\text{ because }\eta>3\epsilon.
\end{align*}

\noindent \underline{Case ~\ref{case2}: $  I(\tilde{X};\tilde{Y}Y)<R_{\one}$ and $R_{\two}\leq I(\tilde{Y};XY)$}\\
 Using~\eqref{eq:1}, we have
 \begin{align*}
 R_{\one} - I(\tilde{X};\tilde{Y}Y)-I(X;\tilde{X}\tilde{Y}Y)+\epsilon\geq -|R_{\two}-I(\tilde{Y};Y)|^{+},\\
 R_{\one} - I(\tilde{X};\tilde{Y}XY)+\epsilon\geq I(X;\tilde{Y}Y) -|R_{\two}-I(\tilde{Y};Y)|^{+}.
 \end{align*}
 We will argue that the RHS is non-negative. When $R_{\two}\leq I(\tilde{Y};Y)$, RHS is $I(X;\tilde{Y}Y)$ which is non-negative. When $I(\tilde{Y};Y)<R_{\two}\leq I(\tilde{Y};XY)$
 \begin{align*}
  I(X;\tilde{Y}Y) -|R_{\two}-I(\tilde{Y};Y)|^{+} &=I(X;\tilde{Y}Y) -R_{\two}+I(\tilde{Y};Y)\\
  & = I(X;Y)+I(X;\tilde{Y}|Y)-R_{\two} +I(\tilde{Y};Y)\\
  &=I(\tilde{Y};XY)-R_{\two}+I(X;Y)\geq 0.
  \end{align*} 
 So, again the RHS is non-negative and $ R_{\one} \geq  I(\tilde{X};\tilde{Y}XY)-\epsilon$. Hence $ |R_{\one} -  I(\tilde{X};\tilde{Y}XY)|^{+}$ $\leq R_{\one} -  I(\tilde{X};\tilde{Y}XY) +\epsilon$. Thus, 
 
 \begin{align}
 &\sum_{\vecz\in \cE_{\mo,1}}W^n(\vecz|\vecx_{\mo}, \vecy)\leq \exp\inp{n\inp{R_{\one}-I(\tilde{X};\tilde{Y}XY) -I(Z;\tilde{X}\tilde{Y}|XY)+3\epsilon}}\nonumber \\
 &\qquad \qquad= \exp\inp{n\inp{R_{\one}-I(\tilde{X};Z\tilde{Y}XY) -I(Z;\tilde{Y}|XY)+3\epsilon}}\nonumber \\
 &\qquad \qquad\leq \exp\inp{n\inp{R_{\one}-I(\tilde{X};Z\tilde{Y}) +3\epsilon}}
 \end{align}
Taking limit $\cP^{\eta}_1\rightarrow \cP_1^{0}$, we get the following rate bound
 \begin{align}
 R_{\one}\leq \min_{P_{X\bar{X}\bar{Y}YZ}\in \cP_1^0, \bar{X}\indep\bar{Y}}{I(\bar{X};Z|\bar{Y})}\label{eq:rbound1}
\end{align}

Thus, 

\noindent \underline{Case~\ref{case3} $R_{\one}\leq  I(\tilde{X};\tilde{Y}XY) $ and $I(\tilde{Y};Y)<R_{\two}$}\\
Using~\eqref{eq:1}, we obtain that 
 \begin{align*}
 R_{\two} - I(\tilde{Y};Y)-I(X;\tilde{X}\tilde{Y}Y)+\epsilon\geq -|R_{\one}-I(\tilde{X};\tilde{Y}Y)|^{+},\\
 R_{\two} - I(\tilde{Y};XY)+\epsilon\geq I(X;Y)+I(X;\tilde{X}|\tilde{Y}Y) -|R_{\one}-I(\tilde{X};\tilde{Y}Y)|^{+}.
 \end{align*}

We will argue that RHS is non-negative. When $R_{\one}\leq I(\tilde{X};\tilde{Y}Y)$, it is clearly true. When $I(\tilde{X};\tilde{Y}Y)<R_{\one} \leq I(\tilde{X};\tilde{Y}XY)$, then 
\begin{align*}
I(X;\tilde{X}|\tilde{Y}Y) -|R_{\one}-I(\tilde{X};\tilde{Y}Y)|^{+} &=I(X;\tilde{X}|\tilde{Y}Y) -R_{\one}+I(\tilde{X};\tilde{Y}Y)\\
&=I(\tilde{X};\tilde{Y}XY)-R_{\one} \geq 0.
\end{align*}
Thus, for $R_{\one}\leq  I(\tilde{X};\tilde{Y}XY) $ and $I(\tilde{Y};Y)<R_{\two}$, $R_{\two} - I(\tilde{Y};XY)+\epsilon\geq0$. This imples that $|R_{\two} - I(\tilde{Y};XY)|^+\leq R_{\two} - I(\tilde{Y};XY) +\epsilon$. So,
 \begin{align}
 &\sum_{\vecz\in \cE_{\mo,1}(P_{X\tilde{X}Y\tilde{Y}Z})}W^n(\vecz|\vecx_{\mo}, \vecy)\leq \exp\inp{n\inp{R_{\two}-I(\tilde{Y};XYZ)-I(Z;\tilde{X}|XY\tilde{Y})+3\epsilon}}\nonumber \\
 &\qquad\qquad\qquad\qquad \rightarrow 0\nonumber \\
 &\text{if }R_{\two}<I(\tilde{Y};XYZ)+I(Z;\tilde{X}|XY\tilde{Y})-3\epsilon.\nonumber
 \end{align}

Thus, 
 \begin{align}
 R_{\two}\leq \min_{P_{XY\bar{X}\bar{Y}Z}\in P_1^0, X\indep\bar{Y}}{I(\bar{Y};Z|X)}\label{eq:rbound2}
\end{align}

\noindent\underline{Case 4: $I(\tilde{X};\tilde{Y}XY)<R_{\one} $ and $I(\tilde{Y};XY)<R_{\two}$}\\
 \begin{align*}
 &\sum_{\vecz\in \cE_{\mo,1}(P_{X\tilde{X}Y\tilde{Y}Z})}W^n(\vecz|\vecx_{\mo}, \vecy)\leq \exp\inp{n\inp{R_{\one}- I(\tilde{X};\tilde{Y}XY)+R_{\two}-I(\tilde{Y};XY)-I(Z;\tilde{X}\tilde{Y}|XY)+2\epsilon}}\\
 &\qquad \qquad\leq \exp\inp{n\inp{R_{\one}+R_{\two}-I(\tilde{X}\tilde{Y};XYZ)-I(\tilde{X};\tilde{Y})+3\epsilon}} \\
 &\qquad\qquad\rightarrow 0 \\
 &\text{if }R_{\one}+R_{\two}<I(\tilde{X}\tilde{Y};XYZ)+I(\tilde{X};\tilde{Y})-3\epsilon.
 \end{align*}
Thus, 
 \begin{align}
 R_{\one}+R_{\two}\leq \min_{P_{XY\bar{X}\bar{Y}Z}\in P_1^0, \bar{X}\indep\bar{Y}}{I(\bar{X}\bar{Y};Z)}\label{eq:rbound3}
\end{align}

Collecting \eqref{eq:rbound1}, \eqref{eq:rbound2} and \eqref{eq:rbound3}, the first term in the RHS of ~\eqref{err:upperbound} goes to zero as $n\rightarrow \infty$ if:
\begin{align}
R_{\one}&\leq \min_{P_{XY\bar{X}\bar{Y}Z}\in P_1^0, \bar{X}\indep\bar{Y}}{I(\bar{X};Z|\bar{Y})}\label{eq:Rbound1}\\
R_{\two}&\leq \min_{P_{XY\bar{X}\bar{Y}Z}\in P_1^0, X\indep\bar{Y}}{I(\bar{Y};Z|X)}\label{eq:Rbound2}\\
R_{\one}+R_{\two}&\leq \min_{P_{XY\bar{X}\bar{Y}Z}\in P_1^0, \bar{X}\indep\bar{Y}}{I(\bar{X}\bar{Y};Z)}\label{eq:Rbound3}
\end{align}
where $P_1^{0}$ is
\begin{align*}
\cP_1^0& =  \{P_{X\tilde{X}\tilde{Y}YZ}\in \cP^n_{\cX\times\cY\times\cX\times\cY\times\cZ}: P_{XYZ}\in \cD_{0},\, P_{\tilde{X}Y'Z}\in \cD_{0} \text{ for some }Y',\\
&\qquad  \, P_{X'\tilde{Y}Z}\in \cD_{0} \text{ for some }X',\, P_{X}=P_{\tilde{X}}=P_{\one}, P_{\tilde{Y}} = P_{\two}\text{ and }I(\tilde{X};\tilde{Y})=0,\,I(X;\tilde{Y})=0\}
\end{align*}

Now, we move on to the second term in the RHS of~\eqref{err:upperbound}. We see that by using~\eqref{codebook:4}, it is sufficient to consider distribution $P_{XY\tilde{Y}_1\tilde{Y}_2Z}\in \cP_2^\eta$ for which 
\begin{align}
I\inp{X;\tilde{Y}_1\tilde{Y}_2Y}\leq|R_{\two}-I(\tilde{Y}_1;Y)|^{+}+|R_{\two}-I(\tilde{Y}_2;\tilde{Y}_1 Y)|^{+} +\epsilon.\label{eq:2}
\end{align}
For $P_{XY\tilde{Y}_1\tilde{Y}_2Z}\in \cP_2^\eta$ satisfying~\eqref{eq:2},

\begin{align}
&\sum_{\vecz\in \cE_{\mo,\vecy, 2}(P_{XY\tilde{Y}_1\tilde{Y}_2Z})}W^n(\vecz|\vecx_{\mo}, \vecy)\nonumber\\
&\qquad\leq\sum_{\substack{\tilde{m}_{\two 1},  \tilde{m}_{\two 2}:\\(\vecx_{\mo}, \vecy_{\tilde{m}_{\two 1}}, \vecy_{\tilde{m}_{\two 2}},\vecy)\in T^{n}_{X\tilde{Y}_1\tilde{Y}_2Y}}}\sum_{\vecz:(\vecx_{\mo}, \vecy_{\tilde{m}_{\two 1}}, \vecy_{\tilde{m}_{\two 2}},\vecy, \vecz)\in T^{n}_{X\tilde{Y}_1\tilde{Y}_2YZ}}W^n(\vecz|\vecx_{\mo}, \vecy)\nonumber\\
&\qquad\leq \sum_{\substack{\tilde{m}_{\two 1}, \tilde{m}_{\two 2}:\\(\vecx_{\mo}, \vecy_{\tilde{m}_{\two 1}}, \vecy_{\tilde{m}_{\two 2}},\vecy)\in T^{n}_{X\tilde{Y}_1\tilde{Y}_2Y}}}\frac{|T^{n}_{Z|X\tilde{Y}_1\tilde{Y}_2Y}(\vecx_{\mo},\vecy_{\tilde{m}_{\two 1}}, \vecy_{\tilde{m}_{\two 2}}, \vecy)|}{|T^n_{Z|XY}(\vecx_{\mo},\vecy)|}\nonumber\\
&\qquad \leq \sum_{\substack{\tilde{m}_{\two 1}, \tilde{m}_{\two 2}:\\(\vecx_{\mo}, \vecy_{\tilde{m}_{\two 1}}, \vecy_{\tilde{m}_{\two 2}},\vecy)\in T^{n}_{X\tilde{Y}_1\tilde{Y}_2Y}}}\frac{\exp\inp{nH(Z|X\tilde{Y}_1\tilde{Y}_2Y)}}{(n+1)^{-|\cX||\cY||\cZ|}\exp\inp{nH(Z|XY)}}\nonumber\\
&\qquad \leq \sum_{\substack{\tilde{m}_{\two 1}, \tilde{m}_{\two 2}:\\(\vecx_{\mo}, \vecy_{\tilde{m}_{\two 1}}, \vecy_{\tilde{m}_{\two 2}},\vecy)\in T^{n}_{X\tilde{Y}_1\tilde{Y}_2Y}}}\exp\inp{-n\inp{I(Z;\tilde{Y}_1\tilde{Y}_2|XY)-\epsilon}} \text{ for large }n.\nonumber\\
&\qquad\stackrel{\text{(a)}}{\leq}\exp\inp{n\inp{|R_{\two}-I(\tilde{Y}_1;XY)|^{+}+|R_{\two}-I(\tilde{Y}_2;\tilde{Y}_1 XY)|^{+}-I(Z;\tilde{Y}_1\tilde{Y}_2|XY)+2\epsilon}},\label{eq:upperbound3}
\end{align}
where (a) follows using~\eqref{codebook:5}.\\

Note that, in the analysis of first term in the RHS of~\eqref{err:upperbound}, if we replace $R_{\one}$ with $R_{\two}$, $\tilde{Y}$ with $\tilde{Y}_1$ and $\tilde{X}$ with $\tilde{Y}_2$, \eqref{eq:upperbound2} changes to \eqref{eq:upperbound3} and the conditions on the distribution \eqref{eq:1} to \eqref{eq:2}. We see that~\eqref{eq:upperbound3} goes to zero when the following hold (cf. \eqref{eq:Rbound1},\eqref{eq:Rbound2},\eqref{eq:Rbound3}):
\begin{align*}
	R_{\two}&<I(\tilde{Y}_{2};Z\tilde{Y}_1XY) +I(Z;\tilde{Y}_1|XY)-3\epsilon\\
	R_{\two}&<I(\tilde{Y}_1;XYZ)+I(Z;\tilde{Y}_2|XY\tilde{Y}_1)-3\epsilon\\
	2R_{\two}&<I(\tilde{Y}_2\tilde{Y}_1;XYZ)+I(\tilde{Y}_2;\tilde{Y}_1)-3\epsilon
\end{align*}
For 
\begin{align*}
\cP_2^0& = \{P_{XY\tilde{Y}_1\tilde{Y}_2Z}\in \cP^n_{\cX\times\cY\times\cY\times\cY\times\cZ}: P_{XYZ}\in \cD_{0},\, P_{X'_1\tilde{Y}_1Z}\in \cD_{0}\text{ for some }X'_1, \\
&\qquad \, P_{X'_2\tilde{Y}_2Z}\in \cD_{0} \text{ for some }X'_2, \,   P_{X}=P_{\one}, P_{\tilde{Y}_1}=P_{\tilde{Y}_2} = P_{\two}\}
\end{align*}

This gives us the following rate bounds
\begin{align}
R_{\two}&\leq \min_{P_{XY\bar{Y}_1\bar{Y}_2Z}\in P_2^0}{I(\bar{Y}_2;Z)}\label{eq:Rbound4}\\
R_{\two}&\leq \min_{P_{XY\bar{Y}_1\bar{Y}_2Z}\in P_2^0}{I(\bar{Y}_1;Z)}\label{eq:Rbound5}\\
2R_{\two}&\leq \min_{P_{XY\bar{Y}_1\bar{Y}_2Z}\in P_2^0}{I(\bar{Y}_2;Z)+I(\bar{Y}_2;Z)}\label{eq:Rbound6}
\end{align}
When user \one is malicious, error will occur either in {\em Step 1} or {\em Step 3} or {\em Step 5}. Error will not happen in {\em Step 1} w.h.p. because of typicality. For {\em Step 3} and {\em Step 5}, we wil get bounds of the form \eqref{eq:Rbound1}, \eqref{eq:Rbound2} and \eqref{eq:Rbound3}. This is because we only consider the candidates which have passes {\em Step 2}. Hence, we get independence conditions from Lemma~\ref{lemma:indep}.

Thus, combining \eqref{eq:Rbound1}, \eqref{eq:Rbound2}, \eqref{eq:Rbound3}, \eqref{eq:Rbound4}, \eqref{eq:Rbound5},\eqref{eq:Rbound6} and bounds from the case when user \one is malicious, we get the following rate region
\\
Let $\cP$ be the set of distribution 
\begin{align*}
\cP = \{P_{XY'X'YZ}: P_{XY'Z}=P_{\one}P_{Y'}W, \,P_{X'YZ}=P_{X'}P_{\two}W, \,X\indep Y\}
\end{align*}
\begin{align*}
R_{\one}&\leq \min_{P_{XY'X'YZ}\in \cP}I(X;Z|Y)\\
R_{\two}&\leq \min_{P_{XY'X'YZ}\in \cP}I(Y;Z)
\end{align*}
This gives us one corner point (given by \eqref{eq:inner_bd_2}) of the rate region, we get the other corner point (given by \eqref{eq:inner_bd_1}) by  changing the order of decoding by performing {\em Step 3} before {\em Step 2}.

\end{proof}

%!TeX root=paper.tex
\section{Proof of Theorem~\ref{thm:outer_bd}}\label{sec:outer_bd_proof}
Consider an $(N_{\one}, N_{\two}, n)$ adversary identifying code $(F^{(n)}_{\one},F^{(n)}_{\two}, \Phi^{(n)})$ (with potential shared randomness between the encoder and the decoder) such that $P_{e}(F^{(n)}_{\one},F^{(n)}_{\two}, \Phi^{(n)}) \leq \epsilon(n)$ where $\epsilon(n)\rightarrow 0$ as $n\rightarrow 0$.
For for all $i\in [1:n]$, let $(Q^i_{X'|X}, Q^i_{Y'|Y})$ be an arbitrary sequence of pairs of channel distributions satisfying \eqref{eq:outer_bound}. Define $\tilde{W}_{i}$ as 
\begin{align*}
\tilde{W}_{i}(z|x,y) \defineqq \sum_{x'}Q^{i}_{X'|X}(x'|x)W(z|x',y) = \sum_{y'}Q^{i}_{Y'|Y}(y'|y)W(z|x,y') 
\end{align*}
for all $x,y,z$. Let $Q^{(n)}_{X'|X}\defineqq \prod_{i=1}^{n}Q^{i}_{X'|X}$, $Q^{(n)}_{Y'|Y}\defineqq \prod_{i=1}^{n}Q^{i}_{Y'|Y}$ and $\tilde{W}^{(n)} = \prod_{i=1}^{n}\tilde{W}_{i}$.

Then, 
\begin{align*}
P_{e, \malone} \geq \frac{1}{N_{\one}\cdot N_{\two}}\sum_{\mo, \mt}\sum_{\vecx}Q^{(n)}_{X'|X}(\vecx|F^{(n)}_{\one}(\mo))W^n\inp{\inb{\vecz:\Phi^{(n)}(\vecz) = \twob}\Big|\vecx,F_{\two}^{(n)}(\mt)}
\end{align*} and

\begin{align*}
P_{e, \maltwo} \geq \frac{1}{N_{\one}\cdot N_{\two}}\sum_{\mo, \mt}\sum_{\vecy}Q^{(n)}_{Y'|Y}(\vecx|F^{(n)}_{\two}(\mt))W^n\inp{\inb{\vecz:\Phi^{(n)}(\vecz) = \oneb}\Big|F_{\one}^{(n)}(\mo), \vecy}.
\end{align*}
Using these two equations, we get
\begin{align*}
2\epsilon(n)&\geq P_{e, \malone} + P_{e, \maltwo} \geq \frac{1}{N_{\one}\cdot N_{\two}}\sum_{\mo, \mt}\Big(\sum_{\vecx}Q^{(n)}_{X'|X}(\vecx|F^{(n)}_{\one}(\mo))W^n\inp{\inb{\vecz:\Phi^{(n)}(\vecz) = \twob}\Big|\vecx,F_{\two}^{(n)}(\mt)}\\
&\qquad+ \sum_{\vecy}Q^{(n)}_{Y'|Y}(\vecy|F^{(n)}_{\two}(\mt))W^n\inp{\inb{\vecz:\Phi^{(n)}(\vecz) = \oneb}\Big|F_{\one}^{(n)}(\mo), \vecy}\Big)\\
& = \frac{1}{N_{\one}\cdot N_{\two}}\sum_{\mo, \mt}\tilde{W}^{(n)}\inp{\inb{\vecz:\Phi^{(n)}(\vecz) \in \{\oneb, \twob\}}\Big|F^{(n)}_{\one}(\mo),F_{\two}^{(n)}(\mt)}.
\end{align*}
Thus, 
\begin{align*}
&\frac{1}{N_{\one}\cdot N_{\two}}\sum_{\mo, \mt}\tilde{W}^{(n)}\inp{\inb{\vecz:\Phi^{(n)}(\vecz) \neq (\mo, \mt)}|F^{(n)}_{\one}(\mo),F_{\two}^{(n)}(\mt)} \\
&= \frac{1}{N_{\one}\cdot N_{\two}}\sum_{\mo, \mt}\tilde{W}^{(n)}\inp{\inb{\vecz:\Phi^{(n)}(\vecz) \in \cM_{\one}\times\cM_{\two}\setminus\{(\mo, \mt)\}}|F^{(n)}_{\one}(\mo),F_{\two}^{(n)}(\mt)}\\
&\qquad + \frac{1}{N_{\one}\cdot N_{\two}}\sum_{\mo, \mt}\tilde{W}^{(n)}\inp{\inb{\vecz:\Phi^{(n)}(\vecz) \in \{\oneb, \twob\}}\Big|F^{(n)}_{\one}(\mo),F_{\two}^{(n)}(\mt)}\\
&\leq \epsilon(n) + 2\epsilon(n)\\
&= 3\epsilon(n).
\end{align*}
Recall that every pair $(Q_{X'|X}, Q_{Y'|Y})$ satisfying \eqref{eq:outer_bound} corresponds to an element in $\tilde{\cW}_W$ which is a convex set (see the discussion in Section~\ref{sec:outer_bd}). Thus, any adversary identifying code for the MAC $W$ with probability of error $\epsilon(n)$ is also a communication code for the AV-MAC $\tilde{\cW}_{W}$ with probability of error at most $3\epsilon(n)$. So, capacity region of $W$ is outer bounded by the capacity region of the AV-MAC $\tilde{\cW}_{W}$. 

The capacity of an AV-MAC only depends on its convex hull \cite{Jahn81}. So, capacity of $\tilde{\cW}_{W}$ is same as capacity of another AV-MAC $\cW_{W}$ which consists of vertices of the convex polytope $\tilde{\cW}_{W}\subseteq \bbR^{|\cX|\times|\cY|\times|\cZ|}$.
The elements in the set $\tilde{\cW}_{W}$ are parameterized by $(Q_{X'|X}, Q_{Y'|Y})$ pairs. It consists of the vertices of the polytope formed using constraints in \eqref{eq:outer_bound} and constraints of the form: (1) $\sum_{x'}P_{X'|X}(x'|x) = 1$ for all $x$, and (2) $P_{X'|X}(x'|x)\geq 0$. There are similar constraints for $P_{Y'|Y}$. Note that there are $|X|^2 + |Y|^2$ inequality constraints. Every point in the resulting polytope satisfies all the equality constraints. We will get faces, edges, vertices etc. depending on the number of additional inequality constraints satisfied at that point. Thus, number of vertices $\leq 2^{|X|^2 + |Y|^2}$.

\section{Examples}
\subsection{Tightness of inner bound for the Binary Erasure MAC}\label{sec:examples}
Recall that for distributions $P_{\one}$ and $P_{\two}$  over $\cX$ and $\cY$, $\cP(P_{\one}, P_{\two}) = \{P_{XY\tilde{X}\tilde{Y}Z}: P_{X\tilde{Y}Z}=P_{\one}\times P_{\tilde{Y}}\times W \text{ for some }P_{\tilde{Y}} \text{ and }P_{\tilde{X}YZ} = P_{\tilde{X}}\times P_{\two}\times W \text{ for some }P_{\tilde{X}}\}$.
Consider $P_{XY\tilde{X}\tilde{Y}Z}\in \cP(P_{\one}, P_{\two})$. 
\begin{align}\label{eq:ex_1}
\bbP(Z = 0) = P_{\one}(0)P_{\tilde{Y}}(0) = P_{\tilde{X}}(0)P_{\two}(0).
\end{align}
\begin{align*}
\bbP(Z=2) &= (1-P_{\one}(0))(1-P_{\tilde{Y}}(0)) =  (1-P_{\tilde{X}}(0))(1-P_{\two}(0))\\
& = 1+ P_{\one}(0)P_{\tilde{Y}}(0) -P_{\one}(0)-P_{\tilde{Y}}(0) = 1+ P_{\tilde{X}}(0)P_{\two}(0) -P_{\tilde{X}}(0)-P_{\two}(0).
\end{align*}
Using \eqref{eq:ex_1}, we get $P_{\one}(0)+P_{\tilde{Y}}(0) = P_{\tilde{X}}(0)+P_{\two}(0)$. Thus, 
\begin{align}\label{eq:ex_2}
P_{\tilde{X}}(0) = P_{\one}(0)+P_{\tilde{Y}}(0) - P_{\two}(0).
\end{align}
Substituting this in \eqref{eq:ex_1}, we get $P_{\one}(0)P_{\tilde{Y}}(0) = P_{\one}(0)P_{\two}(0)+ P_{\tilde{Y}}(0)P_{\two}(0)- P_{\two}(0)P_{\two}(0).$ This implies that 
\begin{align*}
\inp{P_{\one}(0)-P_{\two}(0)}\inp{P_{\tilde{Y}}(0) -P_{\two}(0)} = 0.
\end{align*}
Thus, either $P_{\one}(0)=P_{\two}(0)$ or $P_{\tilde{Y}}(0) = P_{\two}(0)$.
Substituting this in \eqref{eq:ex_2}, we get either $P_{\one}(0)=P_{\two}(0)$ and $P_{\tilde{X}}(0)=P_{\tilde{Y}}(0)$, or $P_{\tilde{Y}}(0) = P_{\two}(0)$ and $P_{\tilde{X}}(0) = P_{\one}(0)$. If we choose $P_{\one}$ and $P_{\two}$ such that $P_{\one}\neq P_{\two}$, then for every $P_{XY\tilde{X}\tilde{Y}Z}\in \cP(P_{\one}, P_{\two})$, $P_{\tilde{Y}} = P_{Y} = P_{\two}$ and $P_{\tilde{X}} = P_{X} = P_{\one}$. 

We know from the definition of $\cP(P_{\one}, P_{\two})$, that $X\indep \tilde{Y}$ and $\tilde{X}\indep Y$. We now analyse the case when there is further restriction of $X\indep Y$ on the distributions. From the definition of $\cP(P_{\one}, P_{\two})$, we note that $P_{\tilde{X}Y|X\tilde{Y}}(0,0|0,0) = 1$ and $P_{\tilde{X}Y|X\tilde{Y}}(1,1|1,1) = 1$.  Let $P_{\tilde{X}Y|X\tilde{Y}}(0,1|0,1) = \alpha$ and $P_{\tilde{X}Y|X\tilde{Y}}(1,0|0,1) = 1-\alpha$ (Note that $P_{\tilde{X}Y|X\tilde{Y}}((0,0) |0,1)) = P_{\tilde{X}Y|X\tilde{Y}}((1,1) |0,1)) = 0$ by definition of $\cP(P_{\one}, P_{\two})$). Similarly, let $P_{\tilde{X}Y|X\tilde{Y}}(1, 0|1,0) = \beta$ and $P_{\tilde{X}Y|X\tilde{Y}}(0, 1|1,0) = 1-\beta$. 
Thus, $P_{XY}(0,0) = P_{X\tilde{Y}}(0, 0)P_{\tilde{X}Y|X\tilde{Y}}(0,0|0,0) + P_{X\tilde{Y}}(0, 1)P_{\tilde{X}Y|X\tilde{Y}}(1,0|0,1)= P_{X}(0)P_{\tilde{Y}}(0)\cdot 1 + P_{X}(0)P_{\tilde{Y}}(1)\cdot (1-\alpha)$. Also, $P_{XY}(0,0) = P_{X}(0)P_Y(0) = P_{X}(0)P_{\tilde{Y}}(0)$ (The last equality follows from the choice of $P_X$ and $P_Y \neq P_X$). This implies that $\alpha = 1$. By evaluating  $P_{XY}(1,1)$, we can show that $\beta = 1$. This implies that $\tilde{X} = X$ and $\tilde{Y} = Y$.

We choose $P_{\one}$ and $P_{\two}$ arbitrarily close to uniform distributions such that $P_{\one}\neq P_{\two}$. Following the arguments above, it is easy to see that the rate pairs given by \eqref{eq:inner_bd_1} and \eqref{eq:inner_bd_2} are arbitrarily close to $(0.5, 1)$ and $(1, 0.5)$ respectively.

\subsection{Binary erasure MAC is not spoofable}\label{sec:BEC_not_spoof}
Suppose the channel is  \one-{\em \spoofable}, that is, there exist distributions ${Q_{Y|\tilde{X},\tilde{Y}}}$ and ${Q_{X|\tilde{X},X'}}$ such that $\forall\,x', \,\tilde{x},\, \tilde{y},\, z,$
\begin{align*}
&\sum_{y}Q_{Y|\tilde{X},\tilde{Y}}(y|\tilde{x},\tilde{y})\mch(z|x',y) \nonumber\\
&= \sum_{y}Q_{Y|\tilde{X},\tilde{Y}}(y|x',\tilde{y})\mch(z|\tilde{x},y) \nonumber\\
& = \sum_{x}Q_{X|\tilde{X},X'}(x|\tilde{x},x')\mch(z|x,\tilde{y}).
\end{align*}
For $(x', \tilde{x}, \tilde{y}, z) = (1, 0, 1, 2)$, this gives $Q_{Y|\tilde{X}\tilde{Y}}(1|0,1) = 0 = Q_{X|\tilde{X}X'}(1|0,1)$ and for $(x', \tilde{x}, \tilde{y}, z) = (1, 0, 0, 0)$, we get $0 = Q_{Y|\tilde{X}\tilde{Y}}(0|1,0) = Q_{X|\tilde{X}X'}(0|0,1)$. However, $Q_{X|\tilde{X}X}(1|0,1) = Q_{X|\tilde{X}X}(0|0,1) = 0$ is not possible. Thus, the channel is not \one-\spoofable. Similarly, we can show that the channel is not \two-\spoofable.

\subsection{Binary additive MAC is not overwritable}\label{sec:BAC_not_over}
Suppose binary additive MAC $Z = X\oplus Y$ is  \two-overwritable. Let $P_{X'|X,Y}$ be the overwriting attack by user \one which satisfies \eqref{eq:overwritable}. Then for all $y'$
\begin{align*}
P_{X'|X,Y}(0|1,1)W(0|0,y')+P_{X'|X,Y}(1|1,1)W(0|1,y') = W(0|1,1) = 1.
\end{align*}
For $y' = 0$ and $1$, this implies that $P_{X'|X,Y}(0|1,1)=1$ and $P_{X'|X,Y}(1|1,1)=1$ respectively, which is not possible simultaneously. Thus, the channel cannot be \two-overwritable. Similarly, we can argue that the channel is not \one-overwritable.

\subsection{Capacity of $(Z_1,Z_2)=(X_1+Y_1, X_2\oplus Y_2)$ under different decoding guarantees}\label{sec:ex3}
We will first show that this channel is \two-symmetrizable, that is, there exists distribution $P_{X|Y}$ such that
\begin{align*}
\sum_{x'\in \cX}P_{X|Y}(x|y')W(z|x,y) = \sum_{x'\in \cX}P_{X|Y}(x|y)W(z|x,y')
\end{align*}
for all $x,y,z$.
Consider $P_{X|Y}((x_1,x_2)|(y_1, y_2)) = 1$ when $(x_1,x_2)= (y_1, y_2)$. Then for $y' = (y_1', y_2')$, $y = (y_1, y_2)$ and $z = (y_1' + y_1, y_2'\oplus y_2)$, both LHS and RHS of the above equation evaluate to $1$, and for every other $z$, they evaluate to $0$. So, the channel is \two-symmetrizable. Similarly, we can show that the channels in \one-symmetrizable. 

Next, we show that this channel is not overwritable. Suppose the channel is \two-overwritable. Let $P_{X'|X,Y}$ be the overwriting attack by user \one which satisfies \eqref{eq:overwritable}. Then for all $(y'_1, y_2')$,
\begin{align*}
\sum_{(x_1',x_2')}P_{X'|X,Y}((x_1',x_2')|(1,1),(1,1))W((2,0)|(x_1',x_2'),(y'_1, y_2'))= W((2,0)|(1,1),(1,1)).
%+P_{X'|X,Y}((0,1)|(1,1),(1,1))W((2,0)|(0,1),(y'_1, y_2'))\\
%&\qquad+P_{X'|X,Y}((1,0)|(1,1),(1,1))W((2,0)|(1,0),(y'_1, y_2'))+P_{X'|X,Y}((1,1)|(1,1),(1,1))W((2,0)|(1,1),(y'_1, y_2'))\\
%& = .
\end{align*}
However, for $(y_1',y_2') = (0,0)$, LHS evaluates to $0$ whereas RHS evaluates to $1$. Hence, the channel is not \two-overwritable. Similarly, we can show that the channel is not \one-overwritable. 

For the capacity region $\cC$, continuing the discussion in Section~\ref{sec:examples} (following Example~\ref{ex:3}), the given attack distributions satisfying \eqref{eq:outer_bound}, gives an outer bound which is the capacity of the binary erasure MAC. This outer bound is also achievable using an adversary identifying code for the binary erasure channel in the first component $Z_1 = X_1+Y_1$. The inputs $X_2$ and $Y_2$ can be chosen arbitrarily.

}

\begin{thebibliography}{00}
\bibitem{NehaBDPITW19} N. Sangwan, M. Bakshi, B. Dey, and V. Prabhakaran, ``Multiple access channels with byzantine users,'' in Proc. IEEE Information Theory Workshop (ITW), 2019.
\bibitem{YHKEG}A. El Gamal and Y.-H. Kim, {\em Network Information Theory.} Cambridge University Press, 2011.
\bibitem{NehaBDPISIT19} N. Sangwan, M. Bakshi, B. Dey, and V. Prabhakaran, ``Multiple access channels with adversarial users,'' in Proc. IEEE International Symposium on Information Theory (ISIT), 2019.
\bibitem{KK2} O.~{Kosut} and J.~{Kliewer}, ``{Network equivalence for a joint compound-arbitrarily-varying network model},'' in Proc. IEEE Information Theory Workshop (ITW), 2016 
\bibitem{KosutKITW18} O. Kosut and J. Kliewer, ``Authentication capacity of adversarial channels,'' in Proc. {IEEE Information Theory Workshop (ITW)}, 2018. 
\bibitem{BBT60} D. Blackwell, L. Breiman, and A. J. Thomasian, ``The capacities of certain channel classes under random coding,'' {\em Annals Math. Stat.}, 31:558-567, 1960.
\bibitem{CsiszarN88} I. Csisz\'ar and P. Narayan, ``The capacity of the arbitrarily varying channel revisited: positivity, constraints,''  {\em IEEE Trans. Inform. Theory} 34(2):181-193, Mar. 1988.
\bibitem{Jahn81} J. H. Jahn, ``Coding of arbitrarily varying multiuser channels,'' {\em IEEE Trans. Inform. Theory}, 27:212--226, 1981.
\bibitem{survey} A. Lapidoth and P. Narayan, ``Reliable communication under channel uncertainty,'' {\em IEEE Transactions on Information Theory}, vol. 44, no. 6, pp. 2148–2177, 1998.
\bibitem{Jaggi7} S.  Jaggi,  M.  Langberg,  S.  Katti,  T.  Ho,  D.  Katabi  and  M.  Me\'dard, ``Resilient network coding in the presence of byzantine adversaries,'' in Proc.  INFOCOM 2007, pp. 616-624. 
\bibitem{KTong} O. Kosut, L. Tong and D. N. C. Tse, ``Polytope Codes Against Adversaries in Networks,'' {\em IEEE Trans. Inform. Theory}, 60:3308-44, 2014.
\bibitem{Yener} X.  He  and  A.  Yener,  ``Strong  secrecy  and  reliable  byzantine  detection in  the  presence  of  an  untrusted  relay,'' {\em IEEE  Trans.  Inform.  Theory}, 59(1):177-192, Jan. 2013.
\bibitem{SimmonsCRYPTO84} G. J. Simmons, ``Authentication theory/coding theory,'' in Proc. Advances in Cryptology-CRYPTO, 1984.
\bibitem{MaurerIT00} U. M. Maurer, ``Authentication theory and hypothesis testing,'' {\em IEEE Trans. Inform. Theory}, 46(4):1350-1356, Jul. 2000. 
\bibitem{Gungor16} O. Gungor and C. E. Koksal, ``On the Basic Limits of RF-Fingerprint-Based Authentication,'' {\em IEEE Trans. Inform. Theory}, 62(8):4523-4543, Aug. 2016.
\bibitem{LaiEPIT09} L. Lai, H. El Gamal and H. V. Poor, ``Authentication over noisy channels,'' {\em IEEE Trans. Inform. Theory}, 55(2):906-916, Feb. 2009.
\bibitem{Jiang14} S. Jiang, ``Keyless Authentication in a Noisy Model,'' {\em IEEE Trans. on Information Forensics and Security}, 9(6):1024-1033, June 2014.
\bibitem{TuLIT18} W. Tu and L. Lai, ``Keyless authentication and authenticated capacity,'' {\em IEEE Trans. Inform. Theory}, 64(5):3696-3714, May 2018. 
\bibitem{Graves16} E.  Graves,  P.  Yu,  and  P.  Spasojevic,  ``Keyless  authentication  in  the presence  of  a  simultaneously  transmitting  adversary,''  in Proc. {IEEE Information Theory Workshop (ITW)}, 2016.
\bibitem{BKKGYu} {A. Beemer, O. Kosut, J. Kliewer, E. Graves and P. Yu}, ``{Structured Coding for Authentication in the Presence of a Malicious Adversary},'' in Proc. {IEEE International Symposium on Information Theory (ISIT)}, 2019.
\bibitem{BeemerCNS20} A. Beemer, E. Graves, J. Kliewer, O. Kosut, and P. Yu, ``Authentication and Partial Message Correction over Adversarial Multiple-Access Channels,'' in Proc. {IEEE Conference on Communications and Network Security (CNS),} 2020.
\shortonly{\textcolor{blue}{\bibitem{lver} {\tt http://www.tifr.res.in/\%7Eneha\_010/BMAC.pdf}}}
\bibitem{AhlswedeC99} R. Ahlswede and N. Cai, ``Arbitrarily varying multiple-access channels. I. Ericson's symmetrizability is adequate, Gubner's conjecture is true,'' {\em IEEE Trans. Inform. Theory} 45(2):742--749, 1999.
\bibitem{PeregS19} U. Pereg and Y. Steinberg, ``The capacity region of the arbitrarily varying MAC: with and without constraints,'' arXiv:1901:00939, 2019.
%\bibitem{Gubner}J. A. Gubner, ``On the deterministic-code capacity of the multiple-access arbitrarily  varying  channel,'' {\em IEEE Trans. Inform. Theory},  36:262–275,1990.






\end{thebibliography}
\end{document}